\documentclass{article}
\pdfoutput=1

\usepackage[utf8]{inputenc}
\usepackage[margin=1in]{geometry}
\usepackage{color, colortbl}
\usepackage[dvipsnames]{xcolor}
\usepackage{graphicx,epstopdf}
\usepackage{amsmath,amssymb,amsthm}
\usepackage{algorithm}
\usepackage{algorithmic}
\usepackage{bm}
\usepackage[caption=false]{subfig}
\usepackage{appendix}
\usepackage{multirow}
\usepackage{mathtools}
\usepackage{braket}
\usepackage{csquotes}
\usepackage[colorlinks=true,linkcolor=blue,bookmarks=true,breaklinks=true]{hyperref}
\usepackage[numbers,comma,sort&compress]{natbib}
\usepackage[english]{babel}
\usepackage[qm]{qcircuit}
\usepackage[capitalize,nameinlink]{cleveref}
\usepackage{tikz}
\usetikzlibrary{shapes.geometric}
\usetikzlibrary{arrows.meta}
\usetikzlibrary{positioning}
\usetikzlibrary{calc}
\usepackage{array}
\usepackage[normalem]{ulem}
\usepackage{makecell}
\usepackage{authblk}

\definecolor{color1}{HTML}{FFDAB9}
\definecolor{color2}{HTML}{FFA07A}
\definecolor{color3}{HTML}{FF7F50}
\definecolor{color4}{HTML}{FF6347}

\definecolor{amber}{rgb}{1.0, 0.75, 0.0}

\newtheorem{thm}{\protect\theoremname}
  \theoremstyle{plain}
  \newtheorem{lem}[thm]{\protect\lemmaname}
  \theoremstyle{remark}
  
  \theoremstyle{plain}
  \newtheorem*{lem*}{\protect\lemmaname}
  \theoremstyle{plain}
  \newtheorem{prop}[thm]{\protect\propositionname}
  \theoremstyle{plain}
  \newtheorem{cor}[thm]{\protect\corollaryname}
  
  \newtheorem{defn}[thm]{Definition}

  \newtheorem{result}[thm]{Result}
  
  \newcolumntype{x}[1]{>{\centering\let\newline\\\arraybackslash\hspace{0pt}}p{#1}}
  
\PassOptionsToPackage{USenglish}{babel}
\usepackage[USenglish]{babel}
  \providecommand{\corollaryname}{Corollary}
  \providecommand{\lemmaname}{Lemma}
  \providecommand{\propositionname}{Proposition}
  \providecommand{\remarkname}{Remark}
\providecommand{\theoremname}{Theorem}

\newcommand{\abs}[1]{\left\lvert#1\right\rvert}
\newcommand{\norm}[1]{\left\lVert#1\right\rVert}

\newcommand{\ud}{\,\mathrm{d}}
\newcommand{\Or}{\mathcal{O}}

\DeclareMathOperator{\poly}{poly}

\renewcommand{\Re}{\mathop{\mathrm{Re}}}
\renewcommand{\Im}{\mathop{\mathrm{Im}}}
\newcommand{\sgn}{\mathop{\mathrm{sgn}}}

\newcommand{\REV}[1]{#1}

\let\originalleft\left
\let\originalright\right
\renewcommand{\left}{\mathopen{}\mathclose\bgroup\originalleft}
\renewcommand{\right}{\aftergroup\egroup\originalright}

\begin{document}

\title{Quantum algorithm for linear non-unitary dynamics with near-optimal dependence on all parameters}

\author[1]{Dong An}
\author[1,2]{Andrew M. Childs}
\author[3,4,5]{Lin Lin}
\affil[1]{Joint Center for Quantum Information and Computer Science \authorcr University of Maryland, College Park, MD 20742, USA}
\affil[2]{Department of Computer Science and
Institute for Advanced Computer Studies \authorcr University of Maryland, College Park, MD 20742, USA}
\affil[3]{Department of Mathematics, University of California, Berkeley, CA 94720, USA}
\affil[4]{Applied Mathematics and Computational Research Division \authorcr Lawrence Berkeley National Laboratory, Berkeley, CA 94720, USA}
\affil[5]{Challenge Institute for Quantum Computation, University of California, Berkeley, CA 94720, USA}

\date{}

\maketitle

\begin{abstract}
We introduce a family of identities that express general linear non-unitary evolution operators as a linear combination of unitary evolution operators, each solving a Hamiltonian simulation problem. This formulation can exponentially enhance the accuracy of the recently introduced linear combination of Hamiltonian simulation (LCHS) method [An, Liu, and Lin, Physical Review Letters, 2023]. For the first time, this approach enables quantum algorithms to solve linear differential equations with both optimal state preparation cost and near-optimal scaling in matrix queries on all parameters. 
\end{abstract}

\tableofcontents


\section{Introduction}

Accurate and efficient simulation of large-scale differential equations presents significant challenges across various computational scientific domains, from physics and chemistry to biology and economics.
Consider a system of linear ordinary differential equations (ODE) in the form
\begin{equation}\label{eqn:ODE}
    \frac{\ud u(t)}{\ud t} = -A(t) u(t) + b(t), \quad u(0) = u_0. 
\end{equation}
Here $A(t) \in \mathbb{C}^{N\times N}$, and $u(t), b(t) \in \mathbb{C}^{N}$. 
The solution of~\cref{eqn:ODE} can be represented as  
\begin{equation}\label{eqn:ODE_solu}
    u(t) = \mathcal{T}e^{-\int_0^t A(s) \ud s} u_0 + \int_0^t \mathcal{T}e^{-\int_s^t A(s') \ud s'} b(s) \ud s,
\end{equation}
where $\mathcal{T}$ denotes the time-ordering operator. 

In practice, the dimension $N$ can be very large, especially for discretized differential equations in high-dimensional spaces. This poses significant computational challenges for classical algorithms. 
 Recent advancements in quantum algorithms offer promising new avenues for substantially improving the simulation of differential equations.
Quantum algorithms for~\cref{eqn:ODE} aim at preparing an $\epsilon$-approximation of the quantum state $\ket{u(T)} = u(T)/\|u(T)\|$, encoding the normalized solution in its amplitudes. 

Among all ODEs, one special case is the Schr\"odinger equation, in which the coefficient matrix $A(t) = iH(t)$ is anti-Hermitian and $b(t) \equiv 0$. 
Simulating the Schr\"odinger equation on a quantum computer, also known as the Hamiltonian simulation problem, is among the most significant applications of quantum computers, and may be one of the first to achieve practical quantum advantage. 
In recent years, there has been significant progress on designing and analyzing efficient quantum algorithms for Hamiltonian simulation in both the time-independent~\cite{BerryAhokasCleveEtAl2007,BerryChilds2012,BerryCleveGharibian2014,BerryChildsCleveEtAl2014,BerryChildsCleveEtAl2015,BerryChildsKothari2015,LowChuang2017,ChildsMaslovNamEtAl2018,ChildsOstranderSu2019,Campbell2019,Low2019,ChildsSu2019,ChildsSuTranEtAl2020,ChenHuangKuengEtAl2020,SahinogluSomma2020} and time-dependent~\cite{HuyghebaertDeRaedt1990,WiebeBerryHoyerEtAl2010,PoulinQarrySommaEtAl2011,WeckerHastingsWiebeEtAl2015,LowWiebe2019,BerryChildsSuEtAl2020,AnFangLin2021,AnFangLin2022} cases. 

In this work, we study quantum algorithms for general ODEs beyond Hamiltonian simulation with general coefficient matrix $A(t)$ and inhomogeneous term $b(t)$. 
As the Hamiltonian simulation problem can readily be efficiently solved on quantum computers, a natural approach to general ODEs is to reduce them to the case of Hamiltonian simulation. 
Such an idea has been realized in recent work~\cite{AnLiuLin2023}, which starts from the Cartesian decomposition~\cite[Chapter I]{Bhatia1997}
\begin{equation}\label{eqn:A_cartesian_1}
A(t) = L(t) + iH(t).
\end{equation}
Here the Hermitian matrices
\begin{equation}\label{eqn:A_cartesian_2}
L(t) = \frac{A(t)+A(t)^{\dagger}}{2}, \quad H(t) = \frac{A(t)-A(t)^{\dagger}}{2i}
\end{equation}
are called the \emph{real} and \emph{imaginary} part of $A(t)$, respectively. 
Throughout the paper we assume $L(t)$ is positive semi-definite for all $t$, which guarantees the asymptotic stability of the dynamics (see~\cref{app:stability} for details). 
Under this assumption, Ref.~\cite{AnLiuLin2023} establishes an identity, expressing any non-unitary evolution operator as a linear combination of Hamiltonian simulation (LCHS) problems:
\begin{equation}\label{eqn:LCHS_original}
    \mathcal{T} e^{-\int_0^t A(s) \ud s} = \int_{\mathbb{R}} \frac{1}{\pi (1+k^2)} \mathcal{T} e^{-i \int_0^t (kL(s)+H(s)) \ud s} \ud k. 
\end{equation}
Notice that $\mathcal{T} e^{-i \int_0^t (kL(s)+H(s)) \ud s}$ is the unitary time-evolution operator for the Hamiltonian $kL(s)+H(s)$, and thus~\cref{eqn:LCHS_original} can be viewed as a weighted continuous summation of unitary operators. 
An efficient quantum algorithm results from a suitable discretization of~\cref{eqn:LCHS_original} and the quantum linear combination of unitaries (LCU) subroutine~\cite{BerryChildsCleveEtAl2014,GilyenSuLowEtAl2019}. 

Prior to LCHS, substantial efforts have been made in addressing general ODEs~\cite{Berry2014,BerryChildsOstranderEtAl2017,ChildsLiu2020,Krovi2022,BerryCosta2022,FangLinTong2022,JinLiuYu2022}. 
The mechanism of LCHS stands distinct from the methodologies of these previous approaches, providing two notable advantages. 
First, the LCHS formula~\cref{eqn:LCHS_original} does not rely on the spectral mapping theorem, \REV{which states that the spectrum of a matrix polynomial is the spectrum of the matrix transformed by the same polynomial}. Compared to transformation-based methods such as quantum singular value transformation (QSVT)~\cite{GilyenSuLowEtAl2019}, the applicable regime of LCHS is greatly broadened to ODEs with general time-dependent matrix $A(t)$. 
Second, the LCHS method does not rely on converting the ODE system into a dilated linear system of equations. As a result, the state preparation cost (i.e., the number of queries to prepare the initial state $u_0$) is significantly reduced, and LCHS achieves optimal state preparation cost. 
(We present a more comprehensive review of related works below, in~\cref{sec:intro_comp}.)

The main drawback of the original LCHS method based on~\cref{eqn:LCHS_original} is the slow decay rate of the integrand as a function of $k$, resulting in non-optimal scaling with precision of the number of queries to the matrices. 
Specifically, to numerically implement~\cref{eqn:LCHS_original}, we truncate the integral to a finite interval $[-K,K]$ and use numerical quadrature. 
As $\frac{1}{\pi(1+k^2)} =\Theta(k^{-2})$ only quadratically decays, we should choose the truncation parameter $K =\Or(1/\epsilon)$ in order to obtain an $\epsilon$-approximation.  
This directly introduces a computational overhead of  $\Or(1/\epsilon)$ in the quantum algorithm, as the Hamiltonian $kL(s)+H(s)$ may have spectral norm as large as $K\norm{L(s)}$ and the complexity of generic Hamiltonian simulation algorithms scales at least linearly in the spectral norm (and thus at least linearly in $1/\epsilon$).

In this study, we substantially broaden the scope of the original LCHS method by introducing a family of LCHS formulae. The best formula in this family has an integrand with an exponential decay rate. Consequently, in comparison to~\cref{eqn:LCHS_original}, the precision dependence of the LCHS algorithm is exponentially improved. For the first time, our new LCHS algorithm achieves simultaneously both optimal state preparation cost and near-optimal scaling\footnote{In this paper, with some slight abuse of terminology, ``near-optimal scaling'' means that the dependence on each individual parameter  matches its query lower bound up to poly-logarithmic factors. Similarly, ``exponential improvement in precision'' refers to the improvement from a polynomial to a poly-logarithmic dependence on $1/\epsilon$.} of the number of queries to matrices as a function of all parameters. 
We further improve our LCHS algorithm with even better precision dependence for ODEs with time-independent coefficient matrix, and design an efficient Gibbs state preparation algorithm based on it. 
We also discuss a hybrid implementation of the improved LCHS formula, which may be suitable for early fault-tolerant quantum devices.

\subsection{Main results}

The main mathematical result in this work is a generalization of the original LCHS formula~\cref{eqn:LCHS_original}. 
Surprisingly, the integrand $\frac{1}{\pi(1+k^2)}$ in~\cref{eqn:LCHS_original} can be replaced by a large family of functions in the form of $\frac{f(k)}{1-ik}$ where $f(k)$, which we call the kernel function, satisfies only mild assumptions. 

\begin{result} [Informal version of~\cref{thm:LCHS_improved}]\label{result:LCHS} 
Assuming that $L(s) \succeq 0$ for all $s\in [0,t]$, and that the function $f(z)$ is analytic on the lower half plane with mild decay and normalization conditions, the following formula holds:
    \begin{equation}\label{eqn:LCHS_intro}
        \mathcal{T} e^{-\int_0^t A(s) \ud s} = \int_{\mathbb{R}} \frac{f(k)}{ 1-ik} \mathcal{T} e^{-i \int_0^t (kL(s)+H(s)) \ud s} \ud k.
    \end{equation}
\end{result}

\Cref{sec:LCHS} discusses~\cref{result:LCHS} in detail. 
Notice that the original LCHS formula~\cref{eqn:LCHS_original} is the special case of the generalized formula with $f(z) = \frac{1}{\pi(1+iz)}$. 
The flexibility in the generalized LCHS formula suggests the possibility of other kernel functions with superior decay properties.
In particular, this paper focuses on the family of kernel functions
\begin{equation}\label{eqn:kernel_intro}
    f(z) = \frac{1}{2\pi e^{-2^\beta} e^{(1+iz)^{\beta}} }, \quad \beta \in (0,1).
\end{equation}
Observe that as $z = k \in \mathbb{R}$ and $k \rightarrow \infty$, this new kernel function decays  at a near-exponential rate of $e^{-c |k|^{\beta}}$. Consequently, we no longer need to truncate the interval at $K = \Or\left(\frac{1}{\epsilon}\right)$, and can instead use the much smaller cutoff $K = \mathcal{O}((\log(1/\epsilon))^{1/\beta})$.\footnote{Throughout the paper \REV{(except in~\cref{sec:adaptive})}, $\beta$ should be viewed as a fixed constant, and the constant factors in the big-O notation may depend on $\beta$. \REV{In~\cref{sec:adaptive} we discuss how to find a good choice of $\beta$ to further reduce the total complexity.}} This leads to an exponential reduction in the Hamiltonian simulation time with respect to $\epsilon$ when compared to the original LCHS formula.

Furthermore, in~\cref{sec:kernel_nogo}, we prove that for kernel functions $f(z)$ which meet the analyticity and decay criteria in \cref{result:LCHS}, even a slight improvement from near exponential decay as specified in \cref{eqn:kernel_intro} to an exponential decay rate, i.e., $e^{-c|k|}$ as $|k| \to \infty$ along the real axis, leads to the kernel function becoming identically zero in the complex plane (\cref{prop:prop_non_existence}). Consequently, $f(z)$ cannot satisfy the normalization condition, and \cref{eqn:LCHS_intro} fails to hold. 
Therefore, the choice of the kernel function given in~\cref{eqn:kernel_intro} as $\beta \rightarrow 1$ is asymptotically near-optimal within the current LCHS framework. 

Based on the improved LCHS formula, we propose an efficient quantum algorithm for solving linear ODEs in the form of~\cref{eqn:ODE}. 
By plugging the improved LCHS formula into the solution representation in~\cref{eqn:ODE_solu}, we can write the solution $u(T)$ as a weighted continuous summation of unitary operators acting on quantum states. 
We then use composite Gaussian quadrature (see~\cref{sec:discertization}) to discretize the integrals, resulting in an approximation of the solution $u(T)$ as a linear combination of Hamiltonian simulation problems. 
To achieve the best asymptotic scaling, we implement each Hamiltonian evolution operator by the truncated Dyson series method~\cite{LowWiebe2019}, and perform the linear combination by the LCU subroutine~\cite{GilyenSuLowEtAl2019} (see~\cref{app:LCU} for an introduction of LCU). 
This improved LCHS algorithm has the following query complexity (see~\cref{sec:quantum_complexity} for details). 

\begin{result}[Informal version of~\cref{thm:complexity_inhomo}]\label{result:complexity}
    The improved LCHS algorithm can solve the linear ODE~\cref{eqn:ODE} at time $T$ with error at most $\epsilon$ using $ \widetilde{\mathcal{O}}\left( \frac{\|u_0\|+\|b\|_{L^1}}{\|u(T)\|} \alpha_A T \left(\log\left(\frac{ 1 }{ \epsilon}\right)\right)^{1+1/\beta}  \right)$ queries to the matrix input oracle\footnote{\REV{Here the matrix input oracle refers to the HAM-T oracle for $A(t)$, whose formal definition is given later in~\cref{eqn:oracle_A}. The HAM-T oracle is a standard matrix input model for time-dependent Hamiltonian simulation algorithms~\cite{LowWiebe2019} and quantum algorithms for differential equations with time-dependent coefficient matrices~\cite{FangLinTong2022,BerryCosta2022}. } } for $A(t)$, and $ \mathcal{O}\left( \frac{\|u_0\|+\|b\|_{L^1}}{\|u(T)\|}  \right)$ queries to the state preparation oracles for $u_0$ and $b(t)$. 
    Here $\alpha_A \geq \max_t \|A(t)\|$, $\|b\|_{L^1} = \int_0^T \|b(s)\| \ud s$, and $\beta \in (0,1)$. 
\end{result}

The improved LCHS method inherits the optimal state preparation cost from the original LCHS method. 
Specifically, we only need one copy of the input state in a single run of the algorithm.  
The factor $\frac{\|u_0\|+\|b\|_{L^1}}{\|u(T)\|}$ is due to the success probability in the post-selection step after LCU. 
The dominant factor in the matrix query complexity arises from the Hamiltonian simulation of $kL(t)+H(t)$, which is executed via the truncated Dyson series method. This results in linear scaling with respect to both $\alpha_A$ and $T$.
The exponent of $\log(1/\epsilon)$ is $1+1/\beta$, which has two main contributing factors.
One factor of $\log(1/\epsilon)$ is tied to the precision requirement in the truncated Dyson series method. 
The additional factor $(\log(1/\epsilon))^{1/\beta}$ is a result of setting $K = \mathcal{O}((\log(1/\epsilon))^{1/\beta})$ to control the truncation error effectively. Consequently, the spectral norm $\norm{kL(t)+H(t)}$ is bounded by $K\norm{L(t)}+\norm{H(t)} = \mathcal{O}((\log(1/\epsilon))^{1/\beta})\norm{L(t)}+\norm{H(t)}$. This bound reflects the multiplicative influence of the factor $\mathcal{O}((\log(1/\epsilon))^{1/\beta})$ on the spectral norm, particularly when the first term is the dominant contributor.

\subsection{Related works}\label{sec:intro_comp}

Berry proposed the first quantum linear ODE algorithm, based on a multi-step method~\cite{Berry2014}. 
The algorithm converts the linear ODE problem into a dilated system of linear equations and then applies a quantum linear system algorithm. 
Specifically, the algorithm first applies standard multi-step time discretization. 
The discretized ODE is formulated as a linear system of equations, whose solution is the entire history of the evolution. 
Then the HHL algorithm~\cite{HarrowHassidimLloyd2009} is applied to obtain a quantum state encoding the evolution trajectory, and \REV{amplitude amplification and post-selection are performed to produce the state encoding only the solution $u(T)$ at the final time. }

Quantum multi-step methods are applicable to inhomogeneous ODEs with time-independent $A$ and $b$ assuming that $A$ can be diagonalized and satisfies certain stability criteria. 
These methods have query complexity that is polynomial in the norm of $A$, evolution time $T$, precision $\epsilon$, and certain norms of the solutions. 
Several subsequent works have provided better asymptotic scaling and broader applicability. 
Most of these improved algorithms rely on the same strategy of converting the ODE problem into a linear system of equations but apply more advanced time discretization approaches and/or quantum linear system algorithms~\cite{Ambainis2012,ChildsKothariSomma2017,SubasiSommaOrsucci2019,LinTong2020,AnLin2022,CostaAnYuvalEtAl2022}. 
Works along this direction include quantum algorithms based on truncated Taylor series~\cite{BerryChildsOstranderEtAl2017,Krovi2022}, spectral methods~\cite{ChildsLiu2020}, and truncated Dyson series~\cite{BerryCosta2022}. 
The most recent linear-system-based algorithm is the truncated Dyson series method\footnote{In this paper, the ``truncated Dyson series method'' may refer to two different methods, depending on the context. For Hamiltonian simulation, it means the algorithm proposed in~\cite{BerryChildsCleveEtAl2015} and subsequently analyzed in~\cite{LowWiebe2019}. For generic linear ODEs, it means the algorithm in~\cite{BerryCosta2022}. }~\cite{BerryCosta2022}, which scales linearly in the norm of $A$ and the evolution time $T$, and poly-logarithmically in the error $\epsilon$.
The truncated Dyson series method is applicable to time-dependent $A(t)$ and $b(t)$ such that the real part $L(t)$ is positive semi-definite. 

From the perspective of computational complexity, the main drawback of the algorithms based on linear system solvers is that they use a large number of queries to the state preparation oracle. 
Specifically, even for the optimal quantum linear system algorithm~\cite{CostaAnYuvalEtAl2022}, the query complexity to the state preparation oracle for $y$ in the system $Mx = y$ is still $\mathcal{O}(\kappa_M \log(1/\epsilon))$, where $\kappa_M$ is the condition number of the matrix $M$. 
When applied to designing quantum ODE algorithms, the matrix of the dilated linear system has condition number linear in $T$ and the spectral norm of $A(t)$ even with high-order time discretization~\cite{ChildsLiu2020,BerryCosta2022}, and the construction of the right-hand side in the linear system involves the initial condition $u_0$ of the ODE as well as the inhomogeneous term $b(t)$. 
As a result, the state preparation cost in these quantum ODE algorithms still depends linearly on $T$ and $\|A(t)\|$, and logarithmically on $1/\epsilon$. 
On the contrary, Hamiltonian simulation algorithms only need a single copy of the initial state. The best known lower bound on the state preparation cost for generic linear ODEs is $\Omega(\|u_0\|/\|u(T)\|)$~\cite{AnLiuWangEtAl2022}, which does not have explicit dependence on $T$, $\|A(t)\|$, or $\epsilon$. 

QSVT~\cite{GilyenSuLowEtAl2019} can implement a matrix function $f(A)$ where the transformation is defined over the singular values of $A$. 
When $A$ is time-independent and either Hermitian or anti-Hermitian, the time-evolution operator $e^{-AT}$ is closely related to the singular value transformation. 
In particular, thanks to the spectral mapping theorem, QSVT can implement $e^{-AT}$ with time-independent Hermitian or anti-Hermitian $A$. Combined with amplitude amplification, the state preparation cost is optimal. 
However, such an argument does not generalize to cases where $A$ has both Hermitian and anti-Hermitian parts or $A(t)$ is time-dependent.

The quantum time-marching method~\cite{FangLinTong2022} is the first quantum ODE algorithm with optimal state preparation cost for general homogeneous linear ODEs. 
This algorithm avoids using the quantum linear system algorithm. 
Instead, it mimics the classical time-marching strategy for ODEs by dividing the entire time interval into short segments and sequentially applying short-time numerical integrators to the initial state. 
Since numerical integrators for general ODEs are not inherently unitary, each utilization of these operators introduces a probability of failure. 
To prevent the success probability from exponentially decaying, the algorithm applies a uniform amplitude amplification procedure~\cite{LowChuang2017a,GilyenSuLowEtAl2019} based on QSVT at each evolution step. 
The quantum time-marching algorithm achieves optimal state preparation cost, but the dependence on $T$ and $\|A(t)\|$ becomes quadratic due to the extra uniform amplitude amplification subroutines. 

The original LCHS method~\cite{AnLiuLin2023} is the first quantum ODE algorithm that achieves optimal state preparation cost for both homogeneous and inhomogeneous linear differential equations. As discussed earlier, the main drawback of the original LCHS is that it is only a first-order method, so the matrix query complexity depends linearly on $1/\epsilon$. LCHS has been applied to address non-unitary quantum dynamics in the context of complex absorbing potentials (CAP). Here the Hermitian component is time-independent and fast-forwardable. By leveraging interaction picture Hamiltonian simulation~\cite{LowWiebe2019} and fast-forwarding techniques, the constraints of first-order accuracy can be surpassed, and the resulting complexity is near-optimal across all parameters for this application~\cite{AnLiuLin2023}.
The original LCHS method is closely related to the  Schr\"odingerization~\cite{JinLiuYu2022} approach, which converts a non-unitary differential equation into a dilated Schr\"odinger equation with an additional momentum dimension, subject to specific initial conditions. This technique has been successfully applied across various studies~\cite{JinLiuLiEtAl2023,HuJinZhang2023,JinLiLiuEtAl2023,JinLiuMa2023,HuJinLiuEtAl2023}. 
Schr\"odingerization is also a first-order method, and its matrix query complexity depends linearly on $1/\epsilon$.

\begin{table}[t]
    \renewcommand{\arraystretch}{2}
    \centering
    \begin{tabular}{c|c|c}\hline\hline
        \textbf{Method} & \textbf{Queries to $A(t)$} & \textbf{Queries to $\ket{u_0}$} \\\hline 
        Spectral method~\cite{ChildsLiu2020} & $\widetilde{\mathcal{O}}\left( \frac{\|u_0\|}{\|u(T)\|} \kappa_V \alpha_A T \text{~poly}\left(\log\left(\frac{1}{\epsilon}\right)\right)\right)$  & $\widetilde{\mathcal{O}}\left( \frac{\|u_0\|}{\|u(T)\|} \kappa_V \alpha_A T \text{~poly}\left(\log\left(\frac{1}{\epsilon}\right)\right)\right)$ \\\hline
        Truncated Dyson~\cite{BerryCosta2022}  & $\widetilde{\mathcal{O}}\left( \frac{\|u_0\|}{\|u(T)\|} \alpha_A T \left(\log\left(\frac{1}{\epsilon}\right)\right)^2\right)$ & $\mathcal{O}\left( \frac{\|u_0\|}{\|u(T)\|} \alpha_A T \log\left(\frac{1}{\epsilon}\right) \right)$ \\\hline
        Time-marching~\cite{FangLinTong2022} & $\widetilde{\mathcal{O}}\left( \frac{\|u_0\|}{\|u(T)\|} \alpha_A^2 T^2 \log\left(\frac{1}{\epsilon}\right)\right)$ & $\mathcal{O}\left( \frac{\|u_0\|}{\|u(T)\|} \right)$  \\\hline
        Original LCHS~\cite{AnLiuLin2023} & $\widetilde{\mathcal{O}}\left( \left(\frac{\|u_0\|}{\|u(T)\|}\right)^2 \alpha_A T /\epsilon \right)$ & $\mathcal{O}\left( \frac{\|u_0\|}{\|u(T)\|} \right)$ \\\hline
        \makecell{Improved LCHS \\(this work, time-dependent)} & $\widetilde{\mathcal{O}}\left( \frac{\|u_0\|}{\|u(T)\|} \alpha_A T \left(\log\left(\frac{1}{\epsilon}\right)\right)^{1+1/\beta} \right)$  & $\mathcal{O}\left( \frac{\|u_0\|}{\|u(T)\|} \right)$  \\\hline
        \makecell{Improved LCHS \\(this work, time-independent)} & $\widetilde{\mathcal{O}}\left( \frac{\|u_0\|}{\|u(T)\|} \alpha_A T \left(\log\left(\frac{1}{\epsilon}\right)\right)^{1/\beta} \right)$  & $\mathcal{O}\left( \frac{\|u_0\|}{\|u(T)\|} \right)$  \\\hline\hline
    \end{tabular}
    \caption{Comparison among improved LCHS and previous methods for homogeneous ODEs.
    All but the last line are complexities for an ODE with time-dependent coefficient matrix $A(t)$, while the last line assumes $A(t) \equiv A$ is time-independent. Here $\alpha_A \geq \|A(t)\|$, $T$ is the evolution time, $\epsilon$ is the error, and $\beta \in (0,1)$ is the parameter in the kernel function of the improved LCHS. 
    All but the spectral method assume the real part of $A(t)$ to be positive semi-definite, while in the spectral method $A(t)$ is assumed to be diagonalizable with matrix $V(t)$ such that $\kappa_V \geq \|V(t)\|\|V(t)^{-1}\|$ and all the eigenvalues of $A(t)$ have non-negative real parts. }
    \label{tab:comparison}
\end{table}
 
Compared to previous quantum ODE algorithms, our improved LCHS algorithm simultaneously achieves optimal state preparation cost and near-optimal dependence of the number of queries to the matrix on all parameters, without relying on fast-forwarding assumptions. 
\cref{tab:comparison} compares our improved LCHS algorithm and previous approaches in the homogeneous case.  
The previous algorithms shown in the table include two linear-system-based algorithms (the spectral method~\cite{ChildsLiu2020} and the truncated Dyson series method~\cite{BerryCosta2022}), the time-marching method~\cite{FangLinTong2022}, and the original LCHS method~\cite{AnLiuLin2023}. 
Compared to the spectral method and the truncated Dyson series method, our improved LCHS achieves lower state preparation cost while at least matching their complexity in queries to the matrix. 
Furthermore, our state preparation cost is $\mathcal{O}(\|u_0\|/\|u(T)\|)$, which is optimal and exactly matches the lower bound $\Omega(\|u_0\|/\|u(T)\|)$~\cite{AnLiuWangEtAl2022}. 
Compared to the time-marching and the original LCHS method, all of these three methods achieve optimal state preparation cost.
However, our improved LCHS method generally has better dependence of the query complexity to the matrix on the norm $\|A(t)\|$, the evolution time $T$, and the inverse precision $1/\epsilon$. 
The only exception is that our approach has an additional term $\left(\log\left(\frac{1}{\epsilon}\right)\right)^{1/\beta}$ in the matrix query complexity compared to that of the time-marching method.

\subsection{Additional results} 

\paragraph{Time-independent simulation:}
\Cref{result:complexity} describes the overall query complexity of our algorithm for general linear ODEs. 
In~\cref{sec:special_applications}, we obtain even better query complexity under some specific circumstances. 
We first consider the time-independent homogeneous case, where the coefficient matrix $A(t) \equiv A$ is time-independent and $b(t) \equiv 0$. 

\begin{result}[Informal version of~\cref{cor:complexity_homo_time_independent}]\label{result:complexity_homo_time_independent}
    Consider the time-independent homogeneous case of~\cref{eqn:ODE} where $A(t) \equiv A$ and $b(t) \equiv 0$. 
    Then the improved LCHS algorithm can prepare the normalized solution at time $T$ with error at most $\epsilon$ using $ \widetilde{\mathcal{O}}\left( \frac{\|u_0\|}{\|u(T)\|} \alpha_A T \left(\log\left(\frac{ 1 }{ \epsilon}\right)\right)^{1/\beta}  \right)$ queries to the block-encoding of $A$ and $ \mathcal{O}\left( \frac{\|u_0\|}{\|u(T)\|}  \right)$ queries to the state preparation oracles for $u_0$, 
    where $\alpha_A \geq \|A\|$ and $\beta \in (0,1)$. 
\end{result}

This improves the dependence on precision from $\widetilde{\mathcal{O}}((\log(1/\epsilon))^{1+1/\beta})$ to $\widetilde{\mathcal{O}}((\log(1/\epsilon))^{1/\beta})$, almost matching the lower bound $\widetilde{\Omega}(\log(1/\epsilon))$~\cite{BerryChildsCleveEtAl2014} as $\beta \rightarrow 1$. 
This improvement results from using a better Hamiltonian simulation algorithm in the time-independent case. 
Specifically, as $L$ and $H$ are time-independent, instead of the truncated Dyson series method, we can use quantum signal processing (QSP)~\cite{LowChuang2017,LowChuang2019} or QSVT~\cite{GilyenSuLowEtAl2019} to simulate $e^{-iT(kL+H)}$, so that the query complexity of this step is only $\mathcal{O}(\|kL+H\| T + \log(1/\epsilon))$. 
Even if $k$ is as large as $\mathcal{O}((\log(1/\epsilon))^{1/\beta})$, the overall matrix query complexity is still only $\mathcal{O}(\|A\| T (\log(1/\epsilon))^{1/\beta} + \log(1/\epsilon))$, which is almost linear in $\log(1/\epsilon)$. 
The key aspect here is the additive scaling of QSP/QSVT between the Hamiltonian norm and the precision. 
A similar improvement also applies to the case with time-dependent inhomogeneous term $b(t)$, as long as the coefficient matrix $A(t)\equiv A$ is time-independent.

\paragraph{Gibbs state preparation:}
Our algorithm for time-independent homogeneous ODEs can be directly applied to preparing a Gibbs (i.e., thermal) state $e^{-\gamma L} / Z_{\gamma}$. 
Here $L \succeq 0$ is the Hamiltonian, $\gamma$ is the inverse temperature of the system, and $Z_{\gamma} = \mathop{\mathrm{Tr}}(e^{-\gamma L}) $ is the partition function. 
The Gibbs state can be obtained as a partial trace of the purified Gibbs state $\sqrt{N/Z_{\gamma}} (I\otimes e^{-\gamma L/2} ) \bigl(\frac{1}{\sqrt{N}} \sum_{j=0}^{N-1} \ket{j}\ket{j} \bigr)$, which can be prepared by constructing a block-encoding of $e^{-\gamma L/2}$ using our LCHS method. 

\begin{result}[Informal version of~\cref{cor:gibbs_state}]\label{result:gibbs}
    Suppose $L \succeq 0$ is a Hamiltonian. 
    Then the improved LCHS algorithm can prepare the purified Gibbs state with error at most $\epsilon$ using $ \widetilde{\mathcal{O}}\left( \sqrt{\frac{N}{Z_{\gamma}}} \gamma \alpha_L  \left(\log\left(\frac{1}{\epsilon}\right)\right)^{1/\beta} \right)$ queries to a block-encoding of $L$. 
    Here $\alpha_L \geq \|L\|$, $\gamma$ is the inverse temperature, $Z_{\gamma}$ is the partition function, and $\beta \in (0,1)$. 
\end{result}

There have been several works on quantum algorithms for Gibbs state preparation. 
Many of them require stronger assumptions and/or stronger query access to the Hamiltonian $L$. 
For example, the early references~\cite{poulin2009sampling,vanApeldoorn2020quantum} assume the Hamiltonian to be strictly positive definite.
References~\cite{ChowdhurySomma2016,GilyenSuLowEtAl2019,ApersChakrabortyNovoEtAl2022} consider a positive semi-definite Hamiltonian $L$, but require access to a block-encoding of the shifted Hamiltonian $I-L/\alpha_L$ or the square root $\sqrt{L}$ of the Hamiltonian. 
Recently,~\cite{HolmesMuraleedharanSommaEtAl2022} proposed a Gibbs state preparation algorithm for a general Hamiltonian $L$. 
However, when applied to positive semi-definite Hamiltonians, the algorithm in~\cite{HolmesMuraleedharanSommaEtAl2022} scales super-linearly in $\alpha_L$ and $\gamma$, and exponentially in $\poly(\log(1/\epsilon))$. 
Compared to these previous methods, LCHS-based Gibbs state preparation works for a positive semi-definite Hamiltonian using only the ability to simulate its dynamics, and achieves almost linear dependence on $\gamma$ and $\alpha_L$ as well as poly-logarithmic dependence on $1/\epsilon$. 
See~\cref{app:Gibbs} for a more detailed discussion. 

\REV{
\paragraph{Adaptive implementation:} So far we have assumed $\beta$ in the kernel function to be a fixed parameter in $(0,1)$. 
Varying $\beta$ has two opposing effects on the efficiency of our algorithm: when $\beta$ gets closer to $1$, the query complexity of our algorithm becomes asymptotically better, but the constant factor blows up. 
In~\cref{sec:adaptive}, we show how to adaptively choose $\beta$ according to the target error $\epsilon$ to further improve the LCHS algorithm. 
Specifically, by choosing $\beta = 1 - \mathcal{O}((\log\log( 1/\epsilon ) )^{-1})$, we can improve the error dependence of the total query complexity from $(\log(1/\epsilon))^{1+1/\beta}$ to $(\log(1/\epsilon))^{2}$, although this increases the state preparation cost by an additional multiplicative factor of $\log\log(1/\epsilon)$. 
\begin{result}[Informal version of~\cref{cor:complexity_adaptive}]\label{result:adaptive}
    By adaptively choosing $\beta$, the improved LCHS algorithm can solve the linear ODE~\cref{eqn:ODE} at time $T$ with error at most $\epsilon$ using $ \widetilde{\mathcal{O}}\left( \frac{\|u_0\|+\|b\|_{L^1}}{\|u(T)\|} \alpha_A T \left(\log\left(\frac{ 1 }{ \epsilon}\right)\right)^2  \right)$ queries to the matrix input oracle for $A(t)$ and $ \widetilde{\mathcal{O}}\left( \frac{\|u_0\|+\|b\|_{L^1}}{\|u(T)\|} \log\log\left(\frac{1}{\epsilon}\right) \right)$ queries to the state preparation oracles for $u_0$ and $b(t)$. 
    Here $\alpha_A \geq \max_t \|A(t)\|$ and $\|b\|_{L^1} = \int_0^T \|b(s)\| \ud s$. 
\end{result}
}

\paragraph{Hybrid algorithm:}
All the algorithms we have discussed are based on a quantum coherent implementation of LCU, requiring multiple additional ancilla qubits and coherently controlled Hamiltonian simulation. 
If we are interested in the observable $u(T)^*O u(T)$ of the unnormalized solution $u(T)$ for a Hermitian matrix $O$, we may design a more near-term hybrid algorithm, as discussed in~\cref{sec:hybrid}. 
We first discretize~\cref{eqn:LCHS_intro}, again using composite Gaussian quadrature, to obtain the linear combination $\mathcal{T} e^{-\int_0^t A(s) \ud s} \approx \sum_j c_j U_j$ for some coefficients $c_j$ and Hamiltonian simulation operators $U_j$. Then 
\begin{equation}
    u(T)^* O u(T) \approx \sum_{l,j} \overline{c}_l c_j \braket{ u_0 |U_l^{\dagger} O U_j| u_0 },  
\end{equation}
where $\overline{z}$ denotes the complex conjugate of a complex number $z$. 
Thus the observable can be estimated by importance sampling. 
Specifically, we classically sample the index $(j,l)$ with probability proportional to the real and imaginary parts of $\overline{c}_lc_j$, quantumly estimate $\braket{ u_0 |U_l^{\dagger} O U_j| u_0 }$ by the Hadamard test for non-unitary matrices~\cite{TongAnWiebe2021} and amplitude estimation~\cite{BrassardHoyerMoscaEtAl2002}, and classically average all the sampled observables. 
This hybrid algorithm requires fewer ancilla qubits compared to the quantum LCU subroutine, and may be more feasible in the early fault-tolerant regime. 
We remark that the importance sampling strategy works in the general setting of LCU beyond LCHS, and we refer to~\cite{LinTong2022,WanBertaCampbell2022,WangMcArdleBerta2023,Chakraborty2023} for more details.

\subsection{Discussions and open questions}

While our enhanced LCHS method attains near-optimal dependence on all parameters, there remains an opportunity to slightly improve the dependence of its matrix query complexity on precision.
Specifically, in the general time-dependent case, our method scales as $\mathcal{O}((\log(1/\epsilon))^{1+1/\beta})$ for a real parameter $\beta \in (0,1)$, so the power of this logarithmic factor is always larger than $2$. 
On the other hand, the best known lower bound is only $\widetilde{\Omega}(\log(1/\epsilon))$ from the Hamiltonian simulation literature~\cite{BerryChildsCleveEtAl2014}. 
It is then a natural open question whether we can further improve the power of $\log(1/\epsilon)$ in the improved LCHS method while maintaining its other advantages.
By~\cref{prop:prop_non_existence}, such an improvement cannot be achieved by solely modifying the kernel function in the LCHS formula. 

To address this, let us first recall the contributing factors of the $\mathcal{O}((\log(1/\epsilon))^{1+1/\beta})$ scaling. The power $1$ is due to the complexity of the truncated Dyson series method, and the power $1/\beta$ (which can arbitrarily close to $1$ asymptotically) comes from the fact that the spectral norm of the Hamiltonian $kL(t)+H(t)$ depends linearly on $k$, and the largest $k$ used in the Gaussian quadrature is $\mathcal{O}((\log(1/\epsilon))^{1/\beta})$. 
Therefore, in order to obtain an $\widetilde{\mathcal{O}}(\log(1/\epsilon))$ scaling, we can try to reduce the overhead brought by one of these two components. 

One potential route is to employ a better quadrature formula. 
Specifically, if we can construct a formula
\begin{equation}
    \int_{\mathbb{R}} \frac{f(k)}{ 1-ik} \mathcal{T} e^{-i \int_0^t (kL(s)+H(s)) \ud s} \ud k \approx \sum_{j} c_j \mathcal{T} e^{-i \int_0^t (k_j L(s)+H(s)) \ud s}
\end{equation}
in which $\max_j |k_j|$ is significantly reduced, then 
the cost of the LCHS algorithm would have improved dependence on precision. 
Specifically, as discussed after the statement of~\cref{result:complexity}, by treating all parameters other than $\epsilon$ as constants, the overall query complexity to the matrix input oracle is $\mathcal{O}((\max_j |k_j|)\log(1/\epsilon))$. 
If we can construct a quadrature formula with $\max_j |k_j| = \mathcal{O}\left( \poly(\log\log(1/\epsilon))\right)$,  the overall complexity becomes $\widetilde{\mathcal{O}}(\log(1/\epsilon))$.  
We are unsure whether such a quadrature formula can be constructed.

Another approach is to improve the Hamiltonian simulation step, which has readily been achieved in the time-independent case $A(t) \equiv A$ as shown in~\cref{result:complexity_homo_time_independent}. 
The key aspect there is the additive scaling between the Hamiltonian norm and the precision in QSP/QSVT for time-independent Hamiltonian simulation. 
However, it is not clear whether a similar improvement is possible in the general time-dependent case. 
Thus, investigating this direction is reduced to a standalone open question: can we also achieve additive scaling in the time-dependent Hamiltonian simulation problem? 

There are a few more natural open questions on mathematical extensions of the LCHS formula~\cref{eqn:LCHS_intro}. 
First, the current LCHS formula only holds when the real part $L(t)$ is positive semi-definite. 
Can we extend it to a more general scenario where $L(t)$ may have negative eigenvalues but the dynamics are still asymptotically stable? 
There are two potential approaches to this. 
One is to employ other stability assumptions on the matrix $A(t)$ and explore whether a similar LCHS formula can be established. 
Note that while there exist other assumptions also ensuring the stability of non-unitary dynamics, no single stability assumption completely encompasses the others (see~\cite{Krovi2022} for a detailed discussion in the time-independent case). 
The second approach is to consider a state-dependent version of the LCHS formula. 
Specifically, we may consider the case where $L(t)$ has negative eigenvalues but the initial solution $u_0$ only lies in the subspaces corresponding to non-negative eigenvalues of $L(t)$. 
Then, following the same proof in this paper, we find that
\begin{equation}\label{eqn:vector_LCHS}
\mathcal{T} e^{-\int_0^t A(s) \ud s} u_0 = \int_{\mathbb{R}} \frac{f(k)}{ 1-ik} \left(\mathcal{T} e^{-i \int_0^t (kL(s)+H(s)) \ud s} u_0\right)\ud k
\end{equation}
still holds. This initial vector-dependent version of the LCHS formula could be applicable in dynamics with less stringent stability conditions. 
Another question is whether we can use LCHS to represent linear eigenvalue transformations other than matrix exponentials, or even nonlinear transformations, with reduced computational complexity. 

\REV{
\paragraph{Note added:}
Since the announcement of the first version of this paper, there have been several works that provide even further improvements to LCHS. 
Here we briefly summarize some of them for interested readers. 
Reference~\cite{LuLiLiuLiu2025} generalizes the LCHS formalism to the case where $A(t)$ is an infinite-dimensional operator. 
References~\cite{PocrnicJohnsonKatabarwaWiebe2025,HuangAn2025} improve the constant factors in the LCHS algorithm, while~\cite{HuangAn2025} also shows a closer relationship between LCHS and the Fourier transform. 
Reference~\cite{Li2025} establishes a universal moment-fulfilling dilation framework that covers both LCHS and Schr\"odingerization. 
Reference~\cite{LowSomma2025} establishes an approximate version of LCHS and designs a quantum linear ODE algorithm with optimal query complexity in all parameters in the time-independent case. 
It also improves the implementation and the constant factors of LCHS. 
}

\subsection{Organization}

The rest of the paper is organized as follows. 
In~\cref{sec:LCHS}, we lay the mathematical foundation of this work by establishing the general LCHS formula, presenting the improved kernel function, and showing its near-optimality. 
Then we discuss discretization of the LCHS formula in~\cref{sec:discertization}, followed by our quantum algorithm based on this discretized LCHS formula and its complexity analysis in~\cref{sec:quantum_complexity}. 
In~\cref{sec:hybrid} we discuss the hybrid implementation of the LCHS formula.

\section*{Acknowledgements}

We thank Dominic Berry for pointing out the state-dependent version of the LCHS formula in \cref{eqn:vector_LCHS} and the complexity improvement for time-independent inhomogeneous equations in \cref{sec:timeindependent_improve}. 
We thank Robin Kothari for suggesting an adaptive way of choosing $\beta$ to further reduce the total query complexity, as described in~\cref{sec:adaptive}. 
We also thank Jin-Peng Liu and Nathan Wiebe for helpful discussions.

DA acknowledges support from the Department of Defense through the Hartree Postdoctoral Fellowship at QuICS. 
AMC acknowledges support from the Department of Energy, Office of Science, Office of Advanced Scientific Computing Research, Accelerated Research in Quantum Computing program.
This work also received support from the National Science Foundation through QLCI grants OMA-2120757 (DA and AMC) and OMA-2016245 (LL). 
LL is a Simons investigator in Mathematics.

\section*{Competing interests} 

The authors have no competing interests to declare that are relevant to the content of this article.

\section*{Data availability statement} 

Codes and data generated during the current study are available at \url{https://github.com/dong-an/LCHS}.

\section{Linear combination of Hamiltonian simulation}\label{sec:LCHS}

A key observation in this work is that the construction in \cref{eqn:LCHS_original} can be generalized to obtain a family of identities for expressing the propagator. 
In particular, consider the formula 
\begin{equation}\label{eqn:LCHS_improved}
    \mathcal{T} e^{-\int_0^t A(s) \ud s} = \int_{\mathbb{R}} \frac{f(k)}{ 1-ik} \mathcal{T} e^{-i \int_0^t (kL(s)+H(s)) \ud s} \ud k
\end{equation}
with a more general function $f(k)$, which we call a kernel function. 
The following theorem can be proven using a matrix version of Cauchy's integral theorem. 

\begin{thm}\label{thm:LCHS_improved}
    Let $f(z)$ be a function of $z \in \mathbb{C}$, such that 
    \begin{enumerate} 
        \item (Analyticity) $f(z)$ is analytic on the lower half plane $\left\{z: \Im(z) < 0 \right\}$ and continuous on $\left\{z: \Im(z) \leq 0 \right\}$, 
        \item (Decay) there exists a parameter $\alpha > 0$ such that $|z|^\alpha |f(z)| \leq \widetilde{C}$ for a constant $\widetilde{C}$ when $\Im(z) \leq 0$, and 
        \item (Normalization) $\int_{\mathbb{R}} \frac{f(k)}{ 1-ik} \ud k = 1$. 
    \end{enumerate}
    Let $A(t) \in \mathbb{C}^{N\times N}$ be decomposed according to~\cref{eqn:A_cartesian_1,eqn:A_cartesian_2}. If $L(s)\succeq 0$ for all $0\le s\le t$, then
    \begin{equation}\label{eqn:LCHS_improved_in_thm}
    \mathcal{T} e^{-\int_0^t A(s) \ud s} = \int_{\mathbb{R}} \frac{f(k)}{ 1-ik} \mathcal{T} e^{-i \int_0^t (kL(s)+H(s)) \ud s} \ud k. 
\end{equation}
\end{thm}

\subsection{Proof of \texorpdfstring{\cref{thm:LCHS_improved}}{Theorem 5}}

We first use Cauchy's integral theorem to prove the following lemma stating that if $L(t)$ is uniformly positive definite in $[0,T]$, then $\mathcal{P} \int_{\mathbb{R}} f(k) \mathcal{T} e^{-i \int_0^t (kL(s)+H(s)) \ud s} \ud k$ is $0$. 
Here $\mathcal{P}$ denotes the Cauchy principal value. 
\REV{Notice that the Cauchy principal value is necessary here, because the kernel function $f(k)$ is only assumed to decay as $1/|k|^{\alpha}$ for $\alpha > 0$ and the integral $\int_{\mathbb{R}} f(k) \mathcal{T} e^{-i \int_0^t (kL(s)+H(s)) \ud s} \ud k$ does not necessarily converge. Instead, the integral $\int_{\mathbb{R}} \frac{f(k)}{1-ik} \mathcal{T} e^{-i \int_0^t (kL(s)+H(s)) \ud s} \ud k$ always converges, so the Cauchy principal value is not needed in the statement of~\cref{thm:LCHS_improved}. }
We remark that, as being proved soon, the Cauchy principal value exists for any $\alpha > 0$ where $\alpha$ is the decay parameter in~\cref{thm:LCHS_improved}, although the integral $\int_{\mathbb{R}} f(k) \mathcal{T} e^{-i \int_0^t (kL(s)+H(s)) \ud s} \ud k$ may diverge for small $\alpha$. \footnote{For example, a necessary condition of the convergence of $\int_{\mathbb{R}} f(k) \mathcal{T} e^{-i \int_0^t (kL(s)+H(s)) \ud s} \ud k$ is the existence of the Fourier transform of $f(k)$. According to the Hausdorff-Young inequality, we need $f\in L^p(\mathbb{R})$ with $1\le p\le 2$, so $\alpha>1/2$. }
    
\begin{lem}\label{lem:cauchy_vanish}
    Let $f(z)$ be a function satisfying the conditions in~\cref{thm:LCHS_improved}, and suppose that $L(t) \succeq \lambda_0 > 0$ for a positive number $\lambda_0$ and all $t \in [0,T]$. 
    Then
    \begin{equation}
        \mathcal{P} \int_{\mathbb{R}} f(k) \mathcal{T} e^{-i \int_0^t (kL(s)+H(s)) \ud s} \ud k = 0. 
    \end{equation}
    Here $\mathcal{P}$ denotes the Cauchy principal value. 
\end{lem}
\begin{proof}
    Using the substitution $\omega = ik$, we have 
    \begin{equation}\label{eqn:LCHS_proof_eq2}
        \mathcal{P} \int_{\mathbb{R}} f(k) \mathcal{T} e^{-i \int_0^t (kL(s)+H(s)) \ud s} \ud k = -i \lim_{R\rightarrow\infty} \int_{-iR}^{iR} f(-i\omega) \mathcal{T} e^{ \int_0^t (-\omega L(s)-iH(s)) \ud s} \ud \omega.  
    \end{equation}
    As shown in~\cref{fig:contour}, let us choose a path $\gamma_C = \left\{ \omega = Re^{i\theta} : \theta \in [-\pi/2,\pi/2] \right\}$ and consider the contour $[-iR,iR] \cup \gamma_C$. 
    According to the analyticity condition, $f(-i\omega)$ is analytic on the right half plane $\left\{ \omega: \Re(\omega) > 0\right\}$ and continuous on the right half plane plus the imaginary axis. By Cauchy's integral theorem, we have 
    \begin{equation}
        \int_{[-iR,iR] \cup \gamma_C}  f(-i\omega) \mathcal{T} e^{ \int_0^t (-\omega L(s)-iH(s)) \ud s} \ud \omega = 0. 
    \end{equation}
    Plugging this back into~\cref{eqn:LCHS_proof_eq2}, we obtain 
    \begin{equation}
        \mathcal{P} \int_{\mathbb{R}} f(k) \mathcal{T} e^{-i \int_0^t (kL(s)+H(s)) \ud s} \ud k = -i \lim_{R\rightarrow\infty} \int_{\gamma_C} f(-i\omega) \mathcal{T} e^{ \int_0^t (-\omega L(s)-iH(s)) \ud s} \ud \omega. 
    \end{equation}
    Changing the variable from $\omega$ to $\theta$ gives 
    \begin{align}
        \mathcal{P} \int_{\mathbb{R}} f(k) \mathcal{T} e^{-i \int_0^t (kL(s)+H(s)) \ud s} \ud k &= \lim_{R\rightarrow\infty} \int_{-\pi/2}^{\pi/2} Re^{i\theta} f(-iRe^{i\theta} ) \mathcal{T} e^{ \int_0^t (-Re^{i\theta} L(s)-iH(s)) \ud s} \ud \theta \\
        & = \lim_{R\rightarrow\infty} \int_{I\cup J} Re^{i\theta} f(-iRe^{i\theta} ) \mathcal{T} e^{ \int_0^t (-Re^{i\theta} L(s)-iH(s)) \ud s} \ud \theta, \label{eqn:LCHS_proof_eq3}
    \end{align}
    where $I = [-\pi/2+\theta_0,\pi/2-\theta_0]$, $J = [-\pi/2,\pi/2] \setminus I$, and 
    \begin{equation}
        \theta_0 = \min \left\{ \frac{1}{R^{1-\alpha/2}}, \frac{\pi}{4} \right\}.
    \end{equation}
    Here $\alpha$ is the parameter in the decay condition. 
    We first bound the integral over $I$ by the inequality 
    \begin{equation}\label{eqn:LCHS_proof_bound_to}
        \left\| \mathcal{T} e^{ \int_0^t (-Re^{i\theta} L(s)-iH(s)) \ud s} \right\| \leq e^{-t\lambda_0 R \cos\theta} \leq e^{-\frac{2}{\pi}t\lambda_0 R \theta_0 }, 
    \end{equation}
    where in the last inequality we use $\cos\theta \geq \cos(\pi/2-\theta_0) = \sin\theta_0 \geq (2/\pi)\theta_0$. 
    Then 
    \begin{equation}\label{eqn:LCHS_proof_bound_I}
            \left\| \int_{I} Re^{i\theta} f(-iRe^{i\theta} ) \mathcal{T} e^{ \int_0^t (-Re^{i\theta} L(s)-iH(s)) \ud s} \ud \theta \right\|
            \leq \int_I R |f(-iRe^{i\theta})| e^{-\frac{2}{\pi}t\lambda_0 R \theta_0 } \ud \theta . 
    \end{equation}
    For sufficiently large $R$, using the decay condition and noticing that $(-i R e^{i\theta})$ is in the lower half plane for $\theta \in [-\pi/2,\pi/2]$, we have 
    \begin{equation}\label{eqn:LCHS_proof_bound_f}
        |f(-iRe^{i\theta})| \leq \frac{\widetilde{C}}{R^\alpha}. 
    \end{equation}
    Then~\cref{eqn:LCHS_proof_bound_I} becomes 
    \begin{equation}
        \left\| \int_{I} Re^{i\theta} f(-iRe^{i\theta} ) \mathcal{T} e^{ \int_0^t (-Re^{i\theta} L(s)-iH(s)) \ud s} \ud \theta \right\|
         \leq \pi \widetilde{C} R^{1-\alpha} e^{-\frac{2}{\pi}t\lambda_0 R \theta_0 } \leq \pi \widetilde{C} R^{1-\alpha} \max\left\{ e^{-\frac{2}{\pi} t\lambda_0 R^{\alpha/2}}, e^{-\frac{1}{2} t\lambda_0 R} \right\}, 
    \end{equation}
    which vanishes as $R \rightarrow \infty$. 
    For the integral over $J$, we may again use the bound of $f$ in~\cref{eqn:LCHS_proof_bound_f} and $\left\| \mathcal{T} e^{ \int_0^t (-Re^{i\theta} L(s)-iH(s)) \ud s} \right\| \leq e^{-t\lambda_0 R \cos\theta} \leq 1$ for any $\theta \in [-\pi/2,\pi/2]$. 
    Then we obtain
    \begin{equation}
        \begin{split}
             \left\|\int_{J} Re^{i\theta} f(-iRe^{i\theta} ) \mathcal{T} e^{ \int_0^t (-Re^{i\theta} L(s)-iH(s)) \ud s} \ud \theta \right\|  \leq \int_J R |f(-iRe^{i\theta})| \ud \theta \leq \widetilde{C} R^{1-\alpha} \cdot 2\theta_0  \leq 2\widetilde{C} R^{-\alpha/2}, 
        \end{split}
    \end{equation}
    which also vanishes as $R \rightarrow \infty$. 
    Therefore,~\cref{eqn:LCHS_proof_eq3} becomes 
    \begin{equation}
        \mathcal{P} \int_{\mathbb{R}} f(k) \mathcal{T} e^{-i \int_0^t (kL(s)+H(s)) \ud s} \ud k = 0. 
    \end{equation}
\end{proof}

\begin{figure}
        \centering
        \includegraphics[width = 0.3\textwidth]{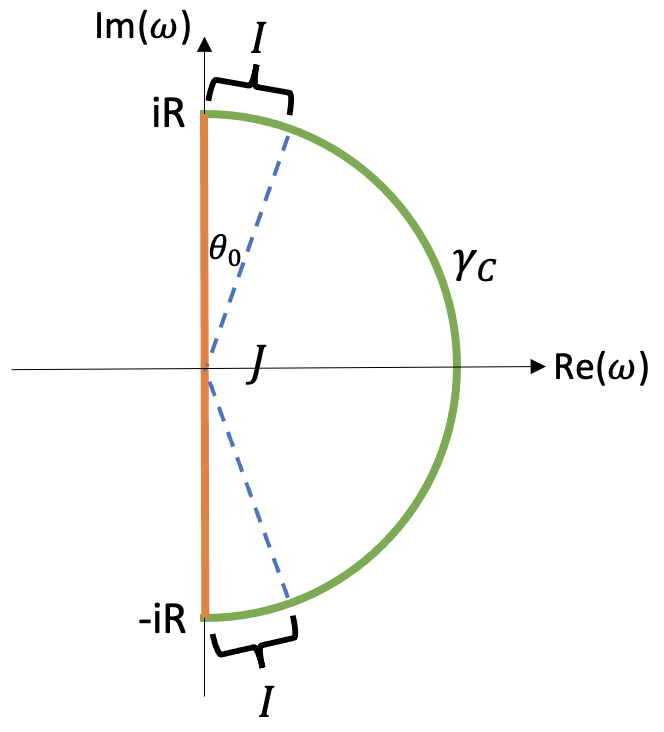}
        \caption{Contour used in the proof of~\cref{lem:cauchy_vanish}.}
        \label{fig:contour}
\end{figure}

We are now ready to prove~\cref{thm:LCHS_improved}. 

\begin{proof}[Proof of~\cref{thm:LCHS_improved}]
    We first prove the theorem for $L(t) \succeq \lambda_0 > 0$ for a positive number $\lambda_0$ and all $t \in [0,T]$. 
    Let $O_L(t) = \mathcal{T} e^{-\int_0^t A(s) \ud s}$ and $O_R(t) = \int_{\mathbb{R}} \frac{f(k)}{ 1-ik} \mathcal{T} e^{-i \int_0^t (kL(s)+H(s)) \ud s} \ud k$. 
    By definition, $O_L(t)$ satisfies the ODE 
    \begin{equation}
        \frac{\ud O_L }{\ud t} = -A(t) O_L(t), \quad O_L(0) = I. 
    \end{equation}
    We will show that $O_R(t)$ satisfies the same ODE. 

    First, for the initial condition, by plugging in $t = 0$ and using the normalization condition, we have 
    \begin{equation}
        O_R(0) = \int_{\mathbb{R}} \frac{f(k)}{ 1-ik} I \ud k = I. 
    \end{equation}
    Now, we choose a fixed $\delta > 0$ and consider $t \in [\delta,T]$. 
    Differentiating $O_R(t)$ with respect to $t$ yields 
    \begin{align}
            \frac{\ud O_R(t)}{\ud t} & = \mathcal{P} \int_{\mathbb{R}} \frac{f(k)}{1-ik} (-ikL(t) - iH(t)) \mathcal{T} e^{-i \int_0^t (kL(s)+H(s)) \ud s} \ud k \\
            & = \mathcal{P} \int_{\mathbb{R}} \frac{f(k)}{1-ik} \left[(-L(t) - iH(t)) + (1-ik)L(t) \right]\mathcal{T} e^{-i \int_0^t (kL(s)+H(s)) \ud s} \ud k \\
            & = -A(t) O_R(t) + L(t) \left(\mathcal{P} \int_{\mathbb{R}} f(k) \mathcal{T} e^{-i \int_0^t (kL(s)+H(s)) \ud s} \ud k \right). \label{eqn:LCHS_proof_eq1}
    \end{align}
    Here we consider the Cauchy principal values of the integrals, as the integrals may not converge absolutely. 
    Additionally, in the first line we exchange the order of the limit (i.e., the Cauchy principal value) and differentiation. 
    This is rigorous as the right-hand side of the first line uniformly converges in $t\in[\delta,T]$. 

    According to~\cref{lem:cauchy_vanish}, the second term on the right-hand side of~\cref{eqn:LCHS_proof_eq1} vanishes. 
    So we obtain
    \begin{equation}\label{eqn:LCHS_proof_eq4}
        \frac{\ud O_R(t)}{\ud t} = -A(t) O_R(t)
    \end{equation}
    for all $t \in [\delta,T]$. 
    As $\delta$ can be an arbitrary positive number,~\cref{eqn:LCHS_proof_eq4} holds for all $t \in (0,T]$. 
    This proves that $O_R(t)$ satisfies the same ODE as $O_L(t)$ with the same initial condition, so $O_L(t) = O_R(t)$ by the uniqueness of the solution. 

    We have assumed $L(t)\succeq \lambda_0 > 0$ for an arbitrary positive number $\lambda_0$ and proved~\cref{eqn:LCHS_improved_in_thm}. 
    Notice that both the left- and the right-hand sides of~\cref{eqn:LCHS_improved_in_thm} are continuous with respect to $L(t)$, so we may take the limit $\lambda_0 \rightarrow 0$ and complete the proof. 
\end{proof}

\subsection{Kernel functions}

There are infinitely many possible kernel functions that satisfy the conditions in~\cref{thm:LCHS_improved}.  
Here we present the one with the best convergence performance; more examples are discussed in~\cref{app:kernel}. 
Note that in the original LCHS formula~\cite{AnLiuLin2023}, $f(z)$ is chosen to be $1/(\pi(1+ik))$. 
Instead, we use 
\begin{equation}\label{eqn:kernel_exp}
    f(z) = \frac{1}{ C_{\beta} e^{(1+iz)^{\beta}} }, 
\end{equation}
where $0 < \beta < 1$, and the normalization factor $C_{\beta}$ can be computed using the residue theorem as 
\begin{equation}
C_\beta =
\int_{\mathbb{R}} \frac{1}{(1-ik)e^{(1+ik)^{\beta}}} \ud k = 
i\int_{\mathbb{R}} \frac{e^{-(1+ik)^{\beta}}}{k+i} \ud k=2\pi e^{-2^\beta}.
\end{equation}
The $\beta$-power of a complex number $Re^{i\theta}$ with $\theta \in (-\pi,\pi]$ is defined to be its principal value $R^{\beta} e^{i\beta\theta}$. 

We first show that~\cref{eqn:kernel_exp} satisfies the conditions in~\cref{thm:LCHS_improved}. 
Consider $z = x+iy$ with $y \leq 0$. 
Then $1+iz = 1-y+ix = Re^{i\theta}$, where $R = \sqrt{(1-y)^2+x^2}$ and the angle $\theta$ is within $(-\pi/2,\pi/2)$ as $1-y \geq 1$.
Hence for any $0 < \beta < 1$, the functions $(1+iz)^{\beta}$ and $1/\exp((1+iz)^{\beta})$ are analytic on the lower half plane. This shows the analyticity condition  in the lower half plane.
To prove the decay condition in the lower half plane, observe that for any $\alpha > 0$, 
\begin{equation}
    |z|^{\alpha} |f(z)| = \frac{|z|^{\alpha}}{C_{\beta}|e^{(1+iz)^{\beta}}|} = \frac{|z|^{\alpha}}{C_{\beta}|e^{R^{\beta}e^{i\beta\theta}}|} = \frac{|z|^{\alpha}}{C_{\beta} e^{R^{\beta} \cos(\beta\theta)}}. 
\end{equation}
Notice that $\cos(\beta\theta) > \cos(\beta\pi/2) $ as $\beta\theta \in (-\beta\pi/2,\beta\pi/2)$ for $0 < \beta < 1$, and $R = \sqrt{(1-y)^2 + x^2} = \sqrt{1-2y+x^2+y^2} > |z|$. 
Then, 
\begin{equation}\label{eqn:bound_f_abs}
    |z|^{\alpha} |f(z)| \leq \frac{|z|^{\alpha}}{ C_{\beta} e^{|z|^{\beta} \cos(\beta\pi/2)} }, 
\end{equation}
which vanishes (and is of course bounded) as $|z| \rightarrow \infty$. 
The normalization condition is directly satisfied by the definition of $C_{\beta}$. 

Compared to the Cauchy distribution used in the original LCHS formula~\cref{eqn:LCHS_original}, our improved choice in~\cref{eqn:kernel_exp} decays exponentially faster, so the improved LCHS can exhibit exponentially better asymptotic convergence (i.e., in numerical discretization, we may truncate the integral at an exponentially smaller threshold. See~\cref{lem:truncation} for more details).   
Notice that there is an extra $\beta \in (0,1)$ parameter in the kernel function~\cref{eqn:kernel_exp}. 
When $\beta \in (0,1)$, larger $\beta$ implies faster asymptotic convergence performance when $z = k \in \mathbb{R}$ and $|k| \rightarrow \infty$, as $|f(k)| \leq e^{-\cos(\beta\pi/2) |k|^{\beta}}$.  
However, $\beta$ cannot be exactly equal to $1$, because $\beta = 1$ gives the kernel function $\frac{1}{C_1 e^{1+ik}}$, which does not decay as $|k| \rightarrow \infty$. 

\begin{figure}[t]
    \centering
    \includegraphics[width=0.45\textwidth]{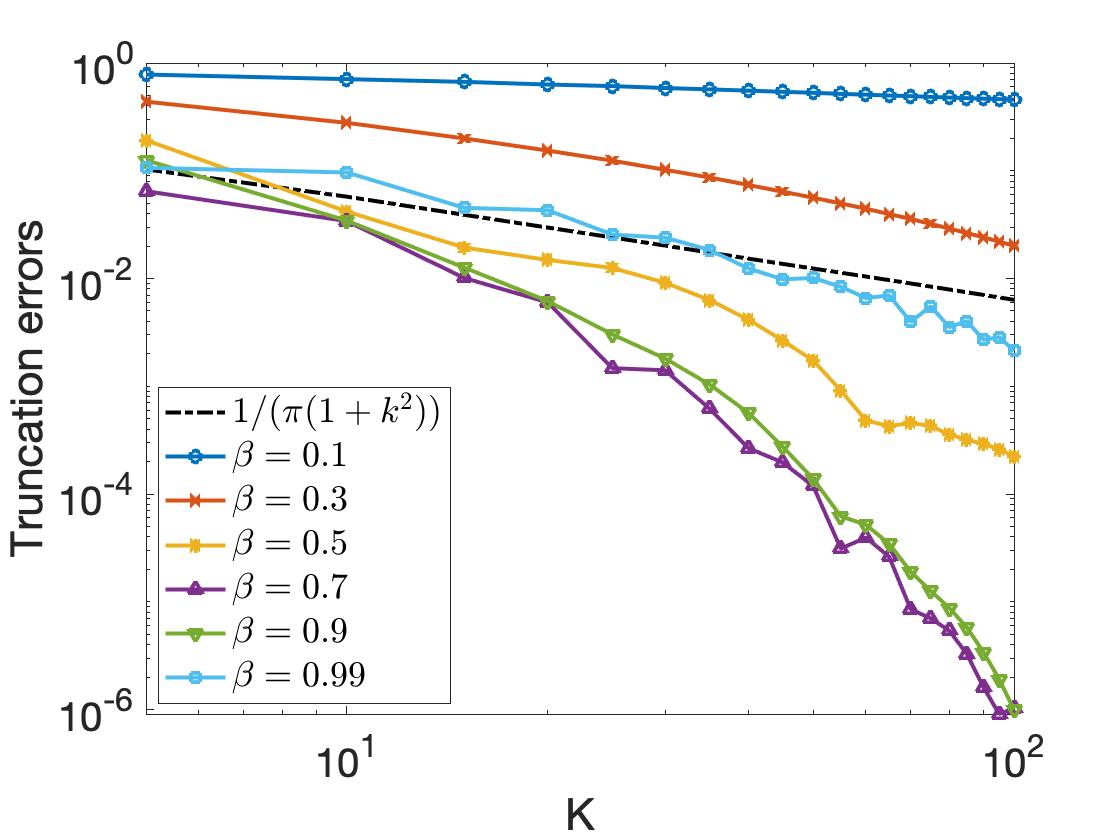}
    \includegraphics[width=0.45\textwidth]{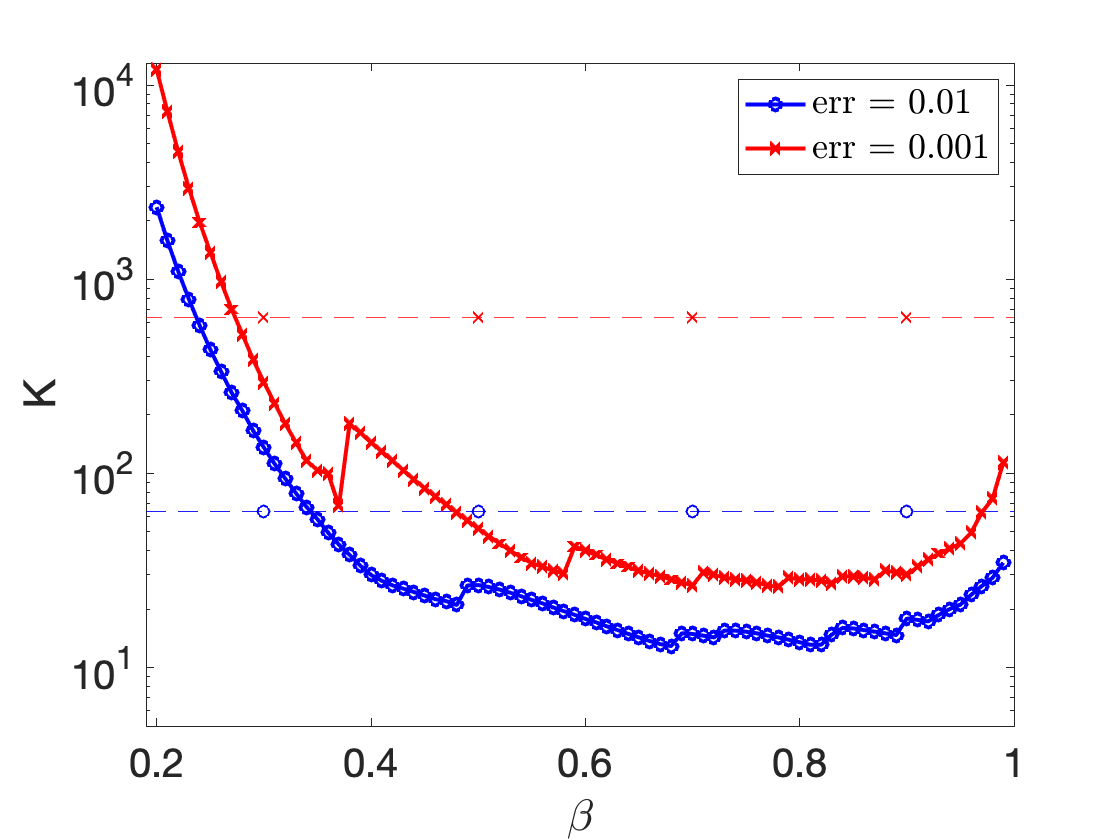}
    \caption{Left: Truncation errors with respect to $K$ when approximating the function $e^{-(L+iH)}$ for some random $L,H$, using either the original LCHS method or the kernel functions as in~\cref{eqn:kernel_exp} with different values of $\beta$. 
    Right: Smallest $K$ that suffices to reduce the truncation errors below $0.01/0.001$ with respect to $\beta$ in~\cref{eqn:kernel_exp}. The dashed lines indicate the corresponding values of $K$ in the original LCHS method.}
    \label{fig:kernel_errors}
\end{figure}

The poor behavior of the limiting case $\beta \rightarrow 1$ suggests that in a nonasymptotic setting, we should choose a value of $\beta$ away from both $0$ and $1$. 
In~\cref{fig:kernel_errors} (left), we numerically evaluate the truncation error as a function of the truncation parameter $K$ (i.e., the integral in~\cref{eqn:LCHS_improved} is truncated on $[-K,K]$) for different kernel functions. 
For this numerical test, we approximate the function $e^{-(L+iH)}$ where $L$ and $H$ are randomly generated $8 \times 8$ Hermitian matrices such that $\|L\|=\|H\| = 1$ and the smallest eigenvalue of $L$ is $0$.  
\cref{fig:kernel_errors} (left) suggests that, in addition to its asymptotic advantage, the function in~\cref{eqn:kernel_exp} also has faster numerical convergence for certain regime of the parameter $\beta$ compared to the Cauchy distribution. 
Furthermore, in~\cref{fig:kernel_errors} (right), we present, under different choices of $\beta$, the smallest numerical values of $K$ to reduce the truncation errors below $0.01$ or $0.001$. 
The dashed lines show the corresponding values of $K$ to achieve the same accuracy in the original LCHS method. 
From this plot, we observe that the Cauchy distribution is outperformed by the kernel function in~\cref{eqn:kernel_exp} for $\beta \in [0.35,0.99]$ (resp.~$\beta \in [0.28,0.99]$) if the target precision is $0.01$ (resp.~$0.001$). 
A practical choice of $\beta$ for the best numerical performance is within the interval $[0.7,0.8]$. 
More numerical illustrations of this family of kernel functions can be found in~\cref{app:kernel}. 

As suggested in~\cref{thm:LCHS_improved}, there are infinitely many kernel functions $f(z)$ that satisfy the LCHS formula~\cref{eqn:LCHS_improved}. 
All valid kernel functions lead to the same Fourier transform on the positive real axis. 
To see this in action, consider a special one-dimensional case of the LCHS formula~\cref{eqn:LCHS_improved}, where $L(s) \equiv x \geq 0$ for a positive real number $x$ and $H(s) \equiv 0$. Then the LCHS formula becomes (by further choosing $t = 1$) 
\begin{equation}
    e^{-x} = \int_{\mathbb{R}} \frac{f(k)}{ 1-ik} e^{-ikx} \ud k, \quad x \geq 0, 
\end{equation}
which implies that the Fourier transform of $\frac{f(k)}{ 1-ik}$ is exactly $e^{-x}$ when $x \geq 0$. 
\cref{fig:kernel_FT} shows the Fourier transform of the function $f(k)/(1-ik)$ where $f(k)$ is either the function in the original LCHS formula or the improved one defined in~\cref{eqn:kernel_exp}. 
As shown in~\cref{fig:kernel_FT}, all the transformed functions are equal to $e^{-x}$ when $x \geq 0$, while they differ considerably on the negative real axis. 
Note that the Cauchy distribution used in the original LCHS corresponds to the inverse Fourier transform of the symmetric function $e^{-|x|}$. 
By exploring the freedom of how the kernel function affects the behavior along the negative real axis, the asymptotic decay rate of the integrand can be significantly improved.

\begin{figure}[t]
    \centering
    \includegraphics[width=0.45\textwidth]{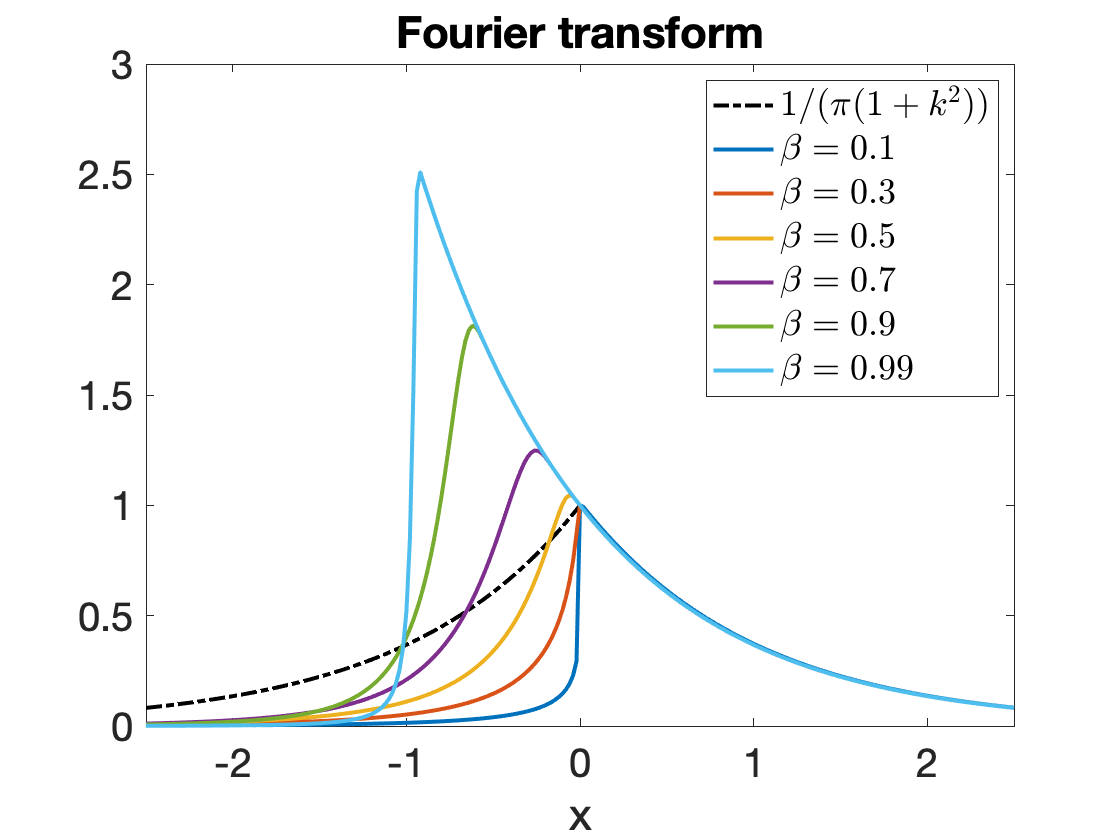}
    \caption{Fourier transform of the function $f(k)/(1-ik)$ where $f(k)$ is the kernel function in the original LCHS method or the one defined in~\cref{eqn:kernel_exp}. Note that the values of all functions agree on $[0,\infty)$ and are equal to $e^{-x}$.}
    \label{fig:kernel_FT}
\end{figure}

\subsection{Near optimality of the kernel function}\label{sec:kernel_nogo}

Notice that when $z = k \in \mathbb{R}$, the function $f(k)$ constructed in~\cref{eqn:kernel_exp} decays as $e^{-|k|^{\beta}\cos(\beta\pi/2)}$. 
To bound the truncation error by $\epsilon$, the integration range should be truncated to $[-K,K]$ with $K=\Or(\log(1/\epsilon))^{1/\beta}$ for $\beta < 1$ (see \cref{lem:truncation} below for a precise statement). 
One may wonder if there exists an even better kernel function with faster decay along the real axis, e.g., $e^{-c|k|}$ as $|k| \rightarrow \infty$. 
If so, we could improve the truncated integration range to $K=\Or(\log(1/\epsilon))$.
Unfortunately, such a function does not exist, so the current LCHS method cannot be significantly improved. 
We show this in the following proposition. 

\begin{prop}\label{prop:prop_non_existence}
    Let $f(z)$ be a function of $z \in \mathbb{C}$ such that 
    \begin{enumerate}
        \item (Analyticity) $f(z)$ is analytic on the lower half plane $\left\{z: \Im(z) < 0 \right\}$ and continuous on $\left\{z: \Im(z) \leq 0 \right\}$, 
        \item (Boundedness) there exists a constant $\widetilde{C}$ such that $|f(z)| \leq \widetilde{C}$ on $\left\{z:\Im(z) \leq 0\right\}$, and
        \item (Exponential decay) there exist constants $c,\widetilde{c} > 0$ such that for any $z = k \in \mathbb{R}$, we have $|f(k)| \leq \widetilde{c} e^{-c|k|}$. 
    \end{enumerate}
    Then $f(z) = 0$ for all $z\in\left\{z:\Im(z) \leq 0\right\}$ (including all $z \in \mathbb{R}$). 
\end{prop}

Notice that the analyticity condition in~\cref{prop:prop_non_existence} is exactly the same as that in~\cref{thm:LCHS_improved}. The boundedness condition in~\cref{prop:prop_non_existence} is  weaker than the decay condition in~\cref{thm:LCHS_improved}. 
Therefore, \cref{prop:prop_non_existence} shows that even slightly faster decay of the kernel function along the real axis than~\cref{eqn:kernel_exp} will result in a function equal to $0$ everywhere, which violates the normalization condition and cannot serve as a satisfying kernel function. 
In this sense, the choice of kernels in~\cref{eqn:kernel_exp} is near-optimal within the current LCHS framework.

To prove~\cref{prop:prop_non_existence}, we use the Phragm\'en–Lindel\"of principle to control the decay of the function in a larger domain. 
The Phragm\'en–Lindel\"of principle is a generalization of the maximum modulus principle to an unbounded domain. 
We first state a specific version of this principle that we use. 

\begin{lem}[Phragm\'en–Lindel\"of principle, {\cite[Corollary 4.4]{Conway1994}}]\label{lem:PhragmenLindelof}
    Let $g(z)$ be a function that is analytic on $\left\{z: \Im(z) < 0 \right\}$ and continuous on $\left\{z: \Im(z) \leq 0 \right\}$. 
    Suppose that $|g(z)| \leq 1$ for all $z \in \mathbb{R}$, and for every $0 < \delta < \delta_0$ there exists $\gamma > 0$ (which may depend on $\delta$) such that $|g(z)| \leq \gamma e^{\delta |z|}$ for sufficiently large $|z|$. 
    Then $|g(z)| \leq 1$ on $\left\{z: \Im(z) \leq 0 \right\}$. 
\end{lem}

Now we are ready to prove~\cref{prop:prop_non_existence}. 

\begin{proof}[Proof of~\cref{prop:prop_non_existence}]
    We first use the Phragm\'en–Lindel\"of principle to show that $f(z)$ also decays exponentially on the entire lower half plane. 
    For any $\xi \in [0,2\pi]$ and $\rho \in (0,1)$, let us define 
    \begin{equation}
        g(z) = \frac{1}{\widetilde{c} e^c +\widetilde{C}}e^{c z^{\rho}e^{i\xi}  } f(z). 
    \end{equation}
    Here $c,\widetilde{c},\widetilde{C}$ are the constants specified in~\cref{prop:prop_non_existence}. 
    Let us verify that $g(z)$ satisfies the assumptions in~\cref{lem:PhragmenLindelof}. 
    We first consider $z = k \in \mathbb{R}$.  
    If $k > 0$, then we have 
    \begin{equation}
        |g(k)| \leq \frac{1}{\widetilde{c} e^c +\widetilde{C}} e^{c k^{\rho} \cos\xi } |f(k)| \leq \frac{1}{e^c} e^{c k^{\rho}  \cos\xi }  e^{-ck} \leq 1. 
    \end{equation}
    If $k < 0$, then 
    \REV{\begin{equation}
        e^{c k^\rho e^{i\xi} } = e^{c|k|^\rho  e^{i\xi} e^{-i\rho \pi}} = e^{c|k|^\rho e^{i(\xi - \rho \pi)}}, 
    \end{equation}}
    and 
    \REV{\begin{equation}
        |g(k)| \leq \frac{1}{\widetilde{c}e^c +\widetilde{C}} e^{c |k|^{\rho} \cos(\xi-\rho \pi) } |f(k)| \leq \frac{1}{e^c} e^{c |k|^{\rho} \cos(\xi-\rho \pi) } e^{-c|k|} \leq 1, 
    \end{equation}}
    so we always have $|g(k)| \leq 1$ for $k \in \mathbb{R}$. 
    Now let $z = Re^{i\theta}$ for $\theta \in (-\pi,0)$. 
    Then 
    \begin{equation}
        e^{cz^{\rho} e^{i\xi} } = e^{c R^{\rho} e^{i(\xi+\rho\theta)} }, 
    \end{equation}
    so 
    \begin{equation}
        |g(z)| = \frac{1}{\widetilde{c}e^c+\widetilde{C}} e^{c R^{\rho} \cos(\xi+\rho\theta) } |f(z)| \leq e^{c R^{\rho} \cos(\xi+\rho\theta) } \leq e^{c R^{\rho} }. 
    \end{equation}
    For any $R > 0$ and $0 < \delta < 1$, notice that $c R^{\rho} - \delta R $ as a function of $R$ attains its maximum at $R = (c\rho/\delta)^{1/(1-\rho)}$. Therefore 
    \begin{equation}
        cR^{\rho} \leq c(c\rho/\delta)^{\rho/(1-\rho)} - \delta (c\rho/\delta)^{1/(1-\rho)} + \delta R, 
    \end{equation}
    and thus 
    \begin{equation}
        |g(z)| \leq e^{c(c\rho/\delta)^{\rho/(1-\rho)} - \delta (c\rho/\delta)^{1/(1-\rho)}} e^{\delta |z|}. 
    \end{equation}
    Therefore, by~\cref{lem:PhragmenLindelof}, we have $|g(z)| \leq 1$ on the entire lower plane, and thus 
    \begin{equation}\label{eqn:proof_non_existence_eq1}
        |f(z)| \leq (\widetilde{c}e^c+\widetilde{C}) e^{ - c |z|^{\rho} \cos(\xi+\rho\theta) }. 
    \end{equation}
    As~\cref{eqn:proof_non_existence_eq1} holds for all $\xi$, for any fixed $z$ and $\rho$, we may accordingly choose $\xi$ such that $\cos(\xi+\rho\theta) = 1$. 
    Then 
    \begin{equation}
         |f(z)| \leq (\widetilde{c}e^c+\widetilde{C}) e^{ - c |z|^{\rho} }. 
    \end{equation}
    Taking the limit $\rho \rightarrow 1$, we obtain 
    \begin{equation}\label{eqn:proof_non_existence_eq2}
         |f(z)| \leq (\widetilde{c}e^c+\widetilde{C}) e^{ - c |z| }. 
    \end{equation}

    Now consider the Fourier transform of $f(z)$ extended to a neighborhood of the real axis.
    Specifically, for $-c/2 < \Im(w) < c/2$, as depicted in~\cref{fig:contour2} (left), we define 
    \begin{equation}
        F(w) = \int_{\mathbb{R}} f(k) e^{-iwk} \ud k. 
    \end{equation}
    Notice that this is a well-defined function by the dominated convergence theorem, using the fact that
    \begin{equation}
        |f(k) e^{-iwk}| \leq \widetilde{c}e^{-c |k|} e^{k \Im(w) } \leq \widetilde{c}e^{-c |k|/2}.
    \end{equation}
    Furthermore, $F(w)$ is analytic over the strip $\left\{w : -c/2 < \Im(w) < c/2 \right\}$. 
    This can be proved by definition. 
    For any fixed $w_0$, consider a sequence $\left\{w_n\right\} \in B_{r}(w_0)$ with $w_n \rightarrow w_0$, where $r < c/4$.  
    Let 
    \begin{equation}
        h_n(k) = \frac{e^{-iw_nk} - e^{-iw_0k}}{w_n-w_0}. 
    \end{equation}
    As $e^{-iwk}$ is analytic, we have $h_n(k) \rightarrow -ik e^{-iw_0 k}$. 
    Furthermore, 
    \begin{align}
        |h_n(k)| &\leq \sup_{w \in B_{r}(w_0)} |\mathrm{d}(e^{-iwk})/\mathrm{d} w| = \sup_{w \in B_{r}(w_0)} |-ik e^{-iwk}| \\
        & = |k| \sup_{\xi \in [0,2\pi]} \left|e^{-ik (w_0+re^{i\xi}) }\right| \\
        & = |k|  e^{ k \Im w_0 } \sup_{\xi \in [0,2\pi]} e^{ k r \sin\xi } \\
        & \leq |k| e^{3c|k|/4 }.  
    \end{align}
    Then $|f(k) h_n(k)| \leq \widetilde{c} |k| e^{-c|k|/4} $, and the integral of its upper bound converges. 
    By the dominated convergence theorem again, we have 
    \begin{equation}
        \lim_{w_n \rightarrow w_0} \frac{F(w_n)-F(w_0)}{w_n-w_0} = \int_{\mathbb{R}} f(k) \lim_{w_n \rightarrow w_0} \frac{e^{-iw_n k}-e^{-i w_0 k}}{w_n-w_0} \ud k = \int_{\mathbb{R}} f(k)(-i)k e^{-iw_0k} \ud k, 
    \end{equation}
    which proves that $F(w)$ is analytic over the strip $\left\{w : -c/2 < \Im(w) < c/2 \right\}$. 

    \begin{figure}
        \centering
        \includegraphics[width = 0.45\textwidth]{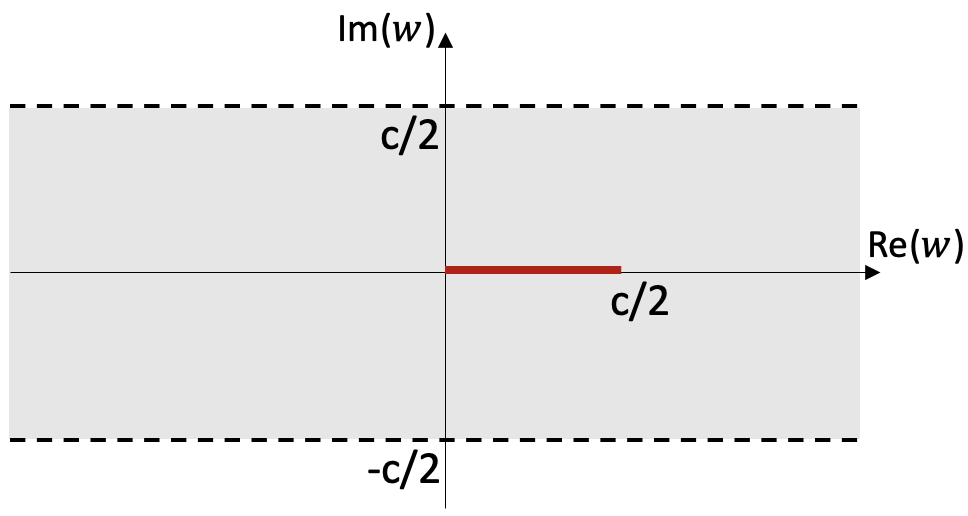}
        \includegraphics[width = 0.45\textwidth]{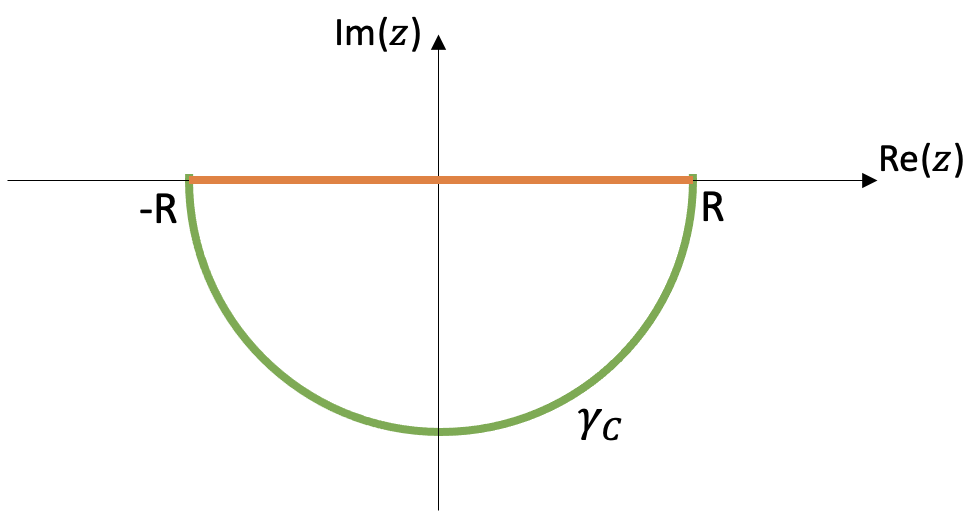}
        \caption{The strip and the contour used in the proof of~\cref{prop:prop_non_existence}. }
        \label{fig:contour2}
    \end{figure}

    Now we show that $F(w)$ is always $0$ over this strip. 
    To show this, we first fix a $w \in (0,c/2)$ to be real (see~\cref{fig:contour2} (left)) and consider the function $f_{w}(z) = f(z) e^{-iwz}$. 
    Then $f_w(z)$ is analytic on $\left\{z : \Im(z) < 0\right\}$ and continuous on $\left\{z : \Im(z) \leq 0\right\}$. 
    Let us consider the contour $[-R,R]\cup \gamma_C$ as shown in~\cref{fig:contour2} (right), where $\gamma_C = \left\{ Re^{i\theta}, \theta \in [-\pi,0] \right\}$. 
    Then, by Cauchy's integral theorem, we have 
    \begin{equation}
        F(w) = \int_{\mathbb{R}} f(k) e^{-iwk} \ud k = \lim_{R \rightarrow \infty} \int_{\gamma_C} f(z) e^{-iwz} \ud z = \lim_{R \rightarrow \infty} \int_{-\pi}^{0} f(Re^{i\theta}) e^{-i w Re^{i\theta}} iRe^{i\theta}\ud \theta. 
    \end{equation}
    Using~\cref{eqn:proof_non_existence_eq2}, we have 
    \begin{equation}
        \left| \int_{-\pi}^{0} f(Re^{i\theta}) e^{-i w Re^{i\theta}} iRe^{i\theta} \ud \theta \right| \leq \pi R \sup_{\theta\in[-\pi,0]}  \left(e^{ w R \sin\theta}\left|f(Re^{i\theta})\right| \right) \leq (\widetilde{c}e^c+\widetilde{C})  \pi R e^{ - c R }, 
    \end{equation}
    which vanishes as $R \rightarrow \infty$. 
    Therefore $F(w) = 0$ for all $w \in (0,c/2)$. This immediately implies that $F(w) = 0$ must hold for all $w \in \left\{ -c/2 < \Im(w) < c/2 \right\}$ according to the identity theorem in complex analysis. 

    We have shown that the Fourier transform of $f(k)$ is identically $0$. 
    By taking the inverse Fourier transform, we obtain that 
    \begin{equation}\label{eqn:proof_non_existence_eq3}
        f(k) = 0, \quad \forall k \in \mathbb{R}. 
    \end{equation} 
    For any $z$ with $\Im(z) < 0$, we may consider again the domain enclosed by the contour in~\cref{fig:contour2} (right) with a sufficiently large $R$. 
    The maximum modulus principle and~\cref{eqn:proof_non_existence_eq2,eqn:proof_non_existence_eq3} imply that
    \begin{equation}
        |f(z)| \leq (\widetilde{c}e^c+\widetilde{C}) e^{ - c R }, \quad \forall |z|<R.  
    \end{equation}
    Taking $R\rightarrow \infty$ yields $f(z) = 0$ for all $z \in \left\{ \Im(z) \leq 0 \right\}$ and completes the proof. 
\end{proof}

\section{Numerical discretization}\label{sec:discertization}

We have shown that the time-evolution operator $\mathcal{T} e^{-\int_0^t A(s) \ud s}$ of the ODE~\cref{eqn:ODE} can be represented as the weighted integral of a set of unitary operators. 
In order to design numerical algorithms based on this, we need to discretize the integral. 
In this section, we discuss approaches to this discretization and provide error bounds. 
We first discuss the discretization of the time-evolution operator, which can readily be used in algorithms for homogeneous ODEs. 
Then we discuss the discretization for the inhomogeneous term. 
For notational simplicity, from now on we only consider the solution at the final time $T$. 

\subsection{Homogeneous term}\label{sec:discertization_homo}

We discretize~\cref{eqn:LCHS_improved} into a discrete sum of unitaries with a specific choice of kernel in~\cref{eqn:kernel_exp}. 
We first truncate to a finite interval $[-K,K]$, and then use the composite Gaussian quadrature (see~\cref{app:quadrature} for a review) to discretize the interval. 
Specifically, let 
\begin{equation}\label{eqn:quadrature}
    g(k) = \frac{1}{C_{\beta} (1-ik) e^{(1+ik)^{\beta}} }, \quad U(T,k) = \mathcal{T} e^{-i \int_0^T (kL(s)+H(s)) \ud s}. 
\end{equation}
We write 
\begin{equation}\label{eqn:LCHS_LCU_composite}
\begin{split}
    \mathcal{T} e^{-\int_0^T A(s) \ud s} = \int_{\mathbb{R}} g(k) U(T,k) \ud k &\approx \int_{-K}^{K} g(k) U(T,k) \ud k \\
    & = \sum_{m = -K/h_1}^{K/h_1-1} \int_{mh_1}^{(m+1)h_1} g(k) U(T,k) \ud k \approx \sum_{m = -K/h_1}^{K/h_1-1} \sum_{q=0}^{Q-1} c_{q,m} U(T,k_{q,m}). 
\end{split}
\end{equation}
Here $h_1$ is the step size used in the composite quadrature rule. 
Without loss of generality, we may choose $K$ and $h_1$ so that $K/h_1$ is an integer. 
On each interval $[mh_1,(m+1)h_1]$, we use Gaussian quadrature with $Q$ nodes. 
Here the $k_{q,m}$'s are the Gaussian nodes, $c_{q,m} = w_q g(k_{q,m})$, and the $w_q$'s are the Gaussian weights (which do not depend on the choice of $m$). 

We now bound the errors in each approximation and estimate how the parameters scale. 

\begin{lem}\label{lem:truncation}
    \begin{enumerate}
        \item The truncation error can be bounded as 
        \begin{equation}
        \left\| \int_{\mathbb{R}} g(k) U(T,k) \ud k - \int_{-K}^{K} g(k) U(T,k) \ud k  \right\| \leq \frac{2^{\left\lceil 1/\beta \right\rceil+1} \left\lceil 1/\beta \right\rceil !}{C_{\beta} \left(\cos(\beta\pi/2)\right)^{\left\lceil 1/\beta \right\rceil} } \frac{1}{ K } e^{-\frac{1}{2}K^{\beta} \cos(\beta\pi/2) }.~  
    \end{equation}
    \item For $\epsilon > 0$, in order to bound the error by $\epsilon$, it suffices to choose 
    \begin{equation}
        K = \mathcal{O}\left( \left(\log\left(\frac{1}{\epsilon}\right)\right)^{1/\beta} \right). 
    \end{equation}
    \end{enumerate}
\end{lem}

\begin{lem}\label{lem:quadrature}
    \REV{Suppose that $T \max_t \|L(t)\| \geq 32/e$. Then}
    \begin{enumerate}
        \item The quadrature error can be bounded as 
        \begin{equation}
            \left\| \int_{-K}^{K} g(k) U(T,k) \ud k - \sum_{m = -K/h_1}^{K/h_1-1} \sum_{q=0}^{Q-1} c_{q,m} U(T,k_{q,m}) \right\| \leq \frac{8 }{3C_{\beta}} K h_1^{2Q} \left(\frac{eT\max_t \|L(t)\|}{2}\right)^{2Q}.
        \end{equation}
        \item For $\epsilon > 0$, in order to bound the error by $\epsilon$, it suffices to choose 
        \begin{equation}
            h_1 = \frac{1}{eT \max_t\|L(t)\|}, \quad Q = \left\lceil \frac{1}{\log 4} \log\left( \frac{8}{3C_{\beta}} \frac{K}{\epsilon} \right) \right\rceil = \mathcal{O}\left( \log\left( \frac{1}{\epsilon} \right) \right),
        \end{equation}
        and the overall number of unitaries in the summation formula is 
        \begin{equation}
            M = \frac{2KQ}{h_1} = \mathcal{O}\left( T \max_t \|L(t)\| \left(\log\left(\frac{1}{\epsilon}\right)\right)^{1+1/\beta} \right). 
        \end{equation}
    \end{enumerate}
\end{lem}

Proofs can be found in~\cref{app:truncation_error} and~\cref{app:quadrature_error}. 
Here $K$ scales as $(\log(1/\epsilon))^{1/\beta}$ since the kernel function~\cref{eqn:kernel_exp} decays as $e^{-c|k|^{\beta}}$ and we need to bound the tail integral by $\epsilon$. 
The high-level reason why $M$ scales as $\mathcal{O}((\log(1/\epsilon))^{1+1/\beta})$ is as follows. 
Since the integrand is analytic in a strip in the complex plane along the real axis, we expect composite Gaussian quadrature to converge exponentially. 
This means that the number of quadrature points in each interval $[mh_1,(m+1)h_1]$ is $\mathcal{O}(\log(1/\epsilon))$.
Multiplying this by the truncation range $\mathcal{O}((\log(1/\epsilon))^{1/\beta})$ yields the claimed scaling. 

We also estimate the sum of all the coefficients $c_{q,m}$, as this affects the success probability of our quantum algorithm. 
Intuitively, this sum should be $\mathcal{O}(1)$ as it can be viewed as the composite Gaussian quadrature formula of an integral of $\mathcal{O}(1)$. 
We state this result as follows, and provide a proof in~\cref{app:coefficient_1norm}. 

\begin{lem}\label{lem:coefficient_1norm}
    Let $c_{q,m}$ be the coefficients defined in the discrete LCHS formula~\cref{eqn:LCHS_LCU_composite}. 
    \REV{Suppose that $T \max_t \|L(t)\| \geq 32/e$. }
    Then 
    \begin{equation}
        \sum_{q,m} |c_{q,m}| = \mathcal{O}(1). 
    \end{equation}
\end{lem}

To simplify the notation, we rewrite~\cref{eqn:LCHS_LCU_composite} as 
\begin{equation}\label{eqn:LCHS_LCU}
    \mathcal{T} e^{-\int_0^T A(s) \ud s} \approx \sum_{j=0}^{M-1} c_j U(T,k_j). 
\end{equation}
Here $M$ is the overall number of the nodes, as estimated in~\cref{lem:quadrature}.

\subsection{Inhomogeneous term} \label{sec:discertization_inhomo}

We now move on to the inhomogeneous term $\int_0^T \mathcal{T}e^{-\int_s^T A(s')\ud s'} b(s) \ud s$. 
By the LCHS formula, we have 
\begin{equation}
    \int_0^T \mathcal{T}e^{-\int_s^T A(s')\ud s'} b(s) \ud s = \int_0^T \int_{\mathbb{R}} \frac{f(k)}{ 1-ik} U(T,s,k)  b(s) \ud k \ud s, 
\end{equation}
where 
\begin{equation}
    U(T,s,k) = \mathcal{T} e^{-i \int_s^T (kL(s')+H(s')) \ud s'}. 
\end{equation}
We first discretize the variable $k$ using the same approach as in~\cref{sec:discertization_homo} and obtain 
\begin{equation}
    \int_0^T \mathcal{T}e^{-\int_s^T A(s')\ud s'} b(s) \ud s \approx \int_0^T \sum_{m_1 = -K/h_1}^{K/h_1-1} \sum_{q_1=0}^{Q_1-1} c_{q_1,m_1} U(T,s,k_{q_1,m_1})  b(s)  \ud s. 
\end{equation}
Then we use composite Gaussian quadrature again to discretize the time variable and obtain 
\begin{equation}\label{eqn:discretization_inhomo}
    \int_0^T \mathcal{T}e^{-\int_s^T A(s')\ud s'} b(s) \ud s \approx \sum^{T/h_2-1}_{m_2 = 0} \sum_{q_2=0}^{Q_2-1} \sum_{m_1 = -K/h_1}^{K/h_1-1} \sum_{q_1=0}^{Q_1-1} c'_{q_2,m_2} c_{q_1,m_1} U(T,s_{q_2,m_2},k_{q_1,m_1})  \ket{b(s_{q_2,m_2})}. 
\end{equation}
Here $h_2$ is the step size for the discretization of the time integral, $Q_2$ is the number of Gaussian quadrature nodes, the $s_{q_2,m_2}$s are the nodes, $c'_{q_2,m_2} = w'_{q_2} \|b(s_{q_2,m_2})\|$, and the $w'_{q_2}$s are the Gaussian weights. 
The $\ket{b}$ with the ket notation implies that it is the normalized vector $b/\|b\|$.

To bound the discretization error by $\epsilon$, we choose sufficiently large $Q_1,Q_2,K$ and sufficiently small $h_1,h_2$. 
The following lemma gives suitable choices. 
Its proof can be found in~\cref{app:discretization_error_inhomo}.

\begin{lem}\label{lem:quadrature_error_inhomo}
    Consider the discretization in~\cref{eqn:discretization_inhomo}. 
 Define $\Lambda = \sup_{p \geq 0, t \in [0,T]} \|A^{(p)}\|^{1/(p+1)} $ and $\Xi = \sup_{ p\geq 0, t \in [0,T] } \|b^{(p)}\|^{1/(p+1)} $, where the superscript $(p)$ refers to the $p$th-order time derivative. 
    \REV{Suppose that $T \max_t \|L(t)\| \geq 32/e$. Then}
    for any $\epsilon > 0$, in order to bound the approximation error of~\cref{eqn:discretization_inhomo} by $\epsilon$, it suffices to choose 
    \begin{equation}
        K = \mathcal{O}\left( \left(\log\left(1+\frac{\|b\|_{L^1}}{\epsilon}\right)\right)^{1/\beta} \right), \quad h_1 = \frac{1}{eT \max_t\|L(t)\|}, \quad h_2 = \frac{1}{eK(\Lambda+\Xi)}, 
    \end{equation}
    \begin{equation}
        Q_1 = \mathcal{O}\left( \log\left( 1+\frac{\|b\|_{L^1}}{\epsilon} \right) \right), \quad Q_2 = \mathcal{O}\left( \log\left(\frac{T(\Lambda+\Xi)}{\epsilon}\right) \right). 
    \end{equation}
\end{lem}

Following similar arguments as for~\cref{lem:coefficient_1norm}, we may also bound the $1$-norm of $c'$ as follows. 

\begin{lem}\label{lem:coefficient_1norm_inhomo}
    Let $c'_{q,m}$ be the coefficients defined in~\cref{eqn:discretization_inhomo}. 
    \REV{Suppose that $T \max_t \|L(t)\| \geq 32/e$.} 
    Then 
    \begin{equation}
        \sum_{q,m} |c'_{q,m}| = \mathcal{O}(\|b\|_{L^1}). 
    \end{equation}
\end{lem}

To simplify the notation, from now on we rewrite the summation in~\cref{eqn:discretization_inhomo} as 
\begin{equation}\label{eqn:discretization_inhomo_2}
    \int_0^T \mathcal{T}e^{-\int_s^T A(s')\ud s'} b(s) \ud s \approx \sum_{j'=0}^{M'-1} \sum_{j=0}^{M-1} c'_{j} c_{j} U(T,s_{j'},k_{j}) \ket{b(s_{j'})},
\end{equation}
where 
\begin{equation}
    M' = \frac{TQ_2}{h_2} = \mathcal{O}\left( T K(\Lambda+\Xi) \log\left(\frac{T(\Lambda+\Xi)}{\epsilon}\right)  \right) = \widetilde{\mathcal{O}}\left( T (\Lambda+\Xi) \left(\log\left(1+\frac{\|b\|_{L^1}}{\epsilon}\right)\right)^{1/\beta} \log\left(\frac{1}{\epsilon}\right)   \right). 
\end{equation}

\section{Quantum implementation and complexity} \label{sec:quantum_complexity}

Now we discuss the quantum implementation of the LCHS method for a general inhomogeneous ODE as in~\cref{eqn:ODE}. 
The algorithm is based on the representation of the solution in~\cref{eqn:ODE_solu} and the discretization of the integrals in~\cref{eqn:LCHS_LCU,eqn:discretization_inhomo_2}, which together result in the approximation 
\begin{equation}\label{eqn:ODE_solu_discrete}
    u(T) \approx \sum_{j=0}^{M-1} c_j U(T,k_j) \ket{u_0} +  \sum_{j'=0}^{M'-1} \sum_{j=0}^{M-1} c'_{j} c_{j} U(T,s_{j'},k_{j})  \ket{b(s_{j'})}. 
\end{equation}
Here $U(T,s,k) = \mathcal{T} e^{-i \int_s^T (kL(s')+H(s')) \ud s'}$ and $U(T,k)$ is shorthand for $U(T,0,k)$. 
Notice that, once we have the circuit for the unitaries $U(T,s,k)$, each summation can be implemented by the linear combination of unitaries (LCU) technique~\cite{ChildsWiebe2012}, and the sum of the two summations can be computed by another LCU. 
To implement each $U(T,s,k)$ with the lowest asymptotic complexity, in this work we use the truncated Dyson series method~\cite{LowWiebe2019}. 
We remark that actually we may choose any time-dependent Hamiltonian simulation algorithm to implement $U(T,s,k)$, such as the simplest product formula~\cite{WiebeBerryHoyerEtAl2010}, which has worse asymptotic scaling but works under weaker oracle assumptions and may have a simpler quantum circuit. 

In this section, we first discuss the oracle assumptions for our algorithm, then describe the algorithm itself, and present a complexity analysis in the general time-independent inhomogeneous scenario. 
We also discuss applications of our algorithm to some specific cases, where we may achieve even better scalings.

\subsection{Oracles}\label{sec:oracles}

Suppose that $[0,T]$ is divided into short time intervals $[qh,(q+1)h]$,
where $h$ is a specified time step in the truncated Dyson series method and $q$ is an integer. 

For the matrices, we assume the HAM-T structure as in~\cite{LowWiebe2019}, which can be viewed as a simultaneous block-encoding of the matrices evaluated at different times. \REV{For certain matrices, such an oracle can be efficiently constructed by known block-encoding techniques~\cite{GilyenSuLowEtAl2019,CladerDalzellStamatopoulosEtAl2022,NguyenKianiLloyd2022,camps2023explicit,sunderhauf2023blockencoding} given time-dependent versions of standard access models to the matrix elements $A_{i,j}(t)$.}
In particular, suppose that we are given oracles $\text{HAM-T}_{A,q}$ such that 
\begin{equation}\label{eqn:oracle_A}
    (\bra{0}_a \otimes I) \text{HAM-T}_{A,q} (\ket{0}_a \otimes I) = \sum_{l=0}^{M_{\text{D}}-1} \ket{l}\bra{l} \otimes \frac{A(qh+lh/M_{\text{D}})}{\alpha_A}, 
\end{equation}
where $M_{\text{D}}$ is the number of the steps used in the short-time truncated Dyson series method, and $\alpha_A$ is the block-encoding factor such that $\alpha_A \geq \max_t\|A(t)\|$. 
Since the unitary $U(T,s,k)$ corresponds to simulation of the Hamiltonian $kL(t) + H(t)$, the Hamiltonian simulation algorithm indeed uses oracles encoding $L(t)$ and $H(t)$ separately, namely
\begin{align}\label{eqn:oracle_L}
    (\bra{0}_a \otimes I) \text{HAM-T}_{L,q} (\ket{0}_a \otimes I) &= \sum_{l=0}^{M_{\text{D}}-1} \ket{l}\bra{l} \otimes \frac{L(qh+lh/M_{\text{D}})}{\alpha_L}, \\
\label{eqn:oracle_H}
    (\bra{0}_a \otimes I) \text{HAM-T}_{H,q} (\ket{0}_a \otimes I) &= \sum_{l=0}^{M_{\text{D}}-1} \ket{l}\bra{l} \otimes \frac{H(qh+lh/M_{\text{D}})}{\alpha_H}. 
\end{align}
Here $\alpha_L$ and $\alpha_H$ are the block-encoding factors such that $\alpha_L \geq \max_t\|L(t)\|$, $\alpha_H \geq \max_t\|H(t)\|$. 
Notice that $\text{HAM-T}_{L,q}$ and $\text{HAM-T}_{H,q}$ can be directly constructed using LCU with $\mathcal{O}(1)$ applications of $\text{HAM-T}_{A,q}$, a controlled rotation gate, and an additional ancilla qubit. The resulting block-encoding factors satisfy $\alpha_L, \alpha_H = \mathcal{O}(\alpha_A)$. 
Alternatively, we may directly assume access to $\text{HAM-T}_{L,q}$ and $\text{HAM-T}_{H,q}$, which may result in better block-encoding factors in some applications. 

For the coefficients $c_j$ in the quadrature formula, we assume access to a pair of state preparation oracles $(O_{c,l}, O_{c,r})$ acting as 
\begin{align}
    O_{c,l}&: \ket{0} \rightarrow \frac{1}{\sqrt{\|c\|_1}} \sum_{j=0}^{M-1} \overline{\sqrt{c_j}} \ket{j}, \\
    O_{c,r}&: \ket{0} \rightarrow \frac{1}{\sqrt{\|c\|_1}} \sum_{j=0}^{M-1} \sqrt{c_j} \ket{j}. 
\end{align}
Here for a complex number $z = Re^{i\theta}$ with $\theta \in (-\pi,\pi]$ denoting its principal argument, we define $\sqrt{z} = \sqrt{R} e^{i\theta/2}$ and its conjugate $\overline{z} = Re^{-i\theta}$. 
Notice that $M = \mathcal{O}\left(\alpha_L T \left(\log\left(\frac{1}{\epsilon}\right)\right)^{1+1/\beta}\right)$ 
from~\cref{lem:quadrature}, so the gate complexity of constructing $O_{c,l}$ and $O_{c,r}$ is $\mathcal{O}\left( \alpha_L T \left(\log\left(\frac{1}{\epsilon}\right)\right)^{1+1/\beta} \right)$. 
Furthermore, we assume an oracle to encode the quadrature nodes $k_j$ in binary, acting as 
\begin{equation}
    O_k: \ket{j}\ket{0} \rightarrow \ket{j}\ket{k_j}, 
\end{equation}
which can be efficiently constructed via classical arithmetic. 
For the initial state $\ket{u_0}$, we assume access to the state preparation oracle acting as
\begin{equation}
    O_u: \ket{0} \rightarrow \ket{u_0}. 
\end{equation}

Related to the inhomogeneous term, we additionally assume access to a state preparation oracle for $b(s)$ that acts as 
\begin{equation}
    O_b: \ket{j'}\ket{0} \rightarrow \ket{j'} \ket{b(s_{j'})}. 
\end{equation}
For the coefficients $c'$, we assume access to state preparation oracle pair $(O_{c',l},O_{c',r})$ acting as
\begin{align}
    O_{c',l}&: \ket{0} \rightarrow \frac{1}{\sqrt{\|c'\|_1}}  \sum_{j'=0}^{M'-1} \overline{\sqrt{c'_{j'}}} \ket{j'}, \\
    O_{c',r}&: \ket{0} \rightarrow \frac{1}{\sqrt{\|c'\|_1}}  \sum_{j'=0}^{M'-1} \sqrt{c'_{j'}} \ket{j'}. 
\end{align}
Notice that, similarly to $O_{c,l}$ and $O_{c,r}$, the gate complexity of constructing $O_{c',l}$ and $O_{c',r}$ is at most  $\mathcal{O}(M') = \widetilde{\mathcal{O}}\left( T (\Lambda+\Xi) \left(\log\left(1+\frac{\|b\|_{L^1}}{\epsilon}\right)\right)^{1/\beta} \log\left(\frac{1}{\epsilon}\right) \right)$. 

\subsection{Algorithm}\label{sec:algorithm}

The main idea of the algorithm is to implement~\cref{eqn:ODE_solu_discrete} via several applications of LCU. 
We first discuss how to implement the homogeneous part, and then discuss the implementation of the full solution. 

\subsubsection{Homogeneous term}\label{sec:algorithm_homo}

The key to implementing the LCU method is to construct the select oracle of $U(T,k_j)$. 
We first show how to do this. 
Let us start with a state 
\begin{equation}\label{eqn:select_oracle_ini_state}
    \ket{j}\ket{l} \ket{0}_k \ket{0}_R \ket{0}_a \ket{\psi}. 
\end{equation}
Here $\ket{\psi}$ is the input state in the system register, and $\ket{j}$ and $\ket{l}$ indicate the indices of the unitary and the time step, respectively. 
There are three ancilla registers: $\ket{0}_k$ encodes $k_j$, $\ket{0}_R$ encodes a choice of rotation, and $\ket{0}_a$ is the ancilla register for the block-encoding of $L$ and $H$.  
We first apply $O_k$ and obtain 
\begin{equation}
    \ket{j}\ket{l} \ket{k_j}_k \ket{0}_R \ket{0}_a \ket{\psi}. 
\end{equation}
Now we construct the HAM-T oracle for $kL+H$ by another layer of LCU. 
We apply the controlled rotation gate
\begin{equation}
    \text{c-R}: \ket{k}\ket{0} \mapsto \ket{k} \left( \frac{\sqrt{\alpha_L k}}{\sqrt{\alpha_L |k| + \alpha_H} } \ket{0} + \frac{\sqrt{\alpha_H}}{\sqrt{\alpha_L |k| + \alpha_H} } \ket{1} \right)
\end{equation}
to obtain 
\begin{equation}
    \ket{j}\ket{l} \ket{k_j}_k \frac{\sqrt{\alpha_L k}}{\sqrt{\alpha_L |k| + \alpha_H} } \ket{0}_R \ket{0}_a \ket{\psi} + \ket{j}\ket{l} \ket{k_j}_k \frac{\sqrt{\alpha_H}}{\sqrt{\alpha_L |k| + \alpha_H} }  \ket{1}_R \ket{0}_a \ket{\psi}. 
\end{equation}
\REV{We recall from the discussion in \cref{sec:oracles} that $\sqrt{k}$ is defined using the principal argument, so that $\sqrt{k}$ is well-defined for all $k \in \mathbb{R}$.}
Now we apply controlled versions of $\text{HAM-T}_{L,q}$ and $\text{HAM-T}_{H,q}$, where the ancilla state $\ket{0}_R$ indicates that we apply $\text{HAM-T}_{L,q}$ and $\ket{1}_R$ indicates that we apply $\text{HAM-T}_{H,q}$. 
Then we obtain
\begin{align}
    &\ket{j}\ket{l} \ket{k_j}_k \frac{\sqrt{\alpha_L k_j}}{\sqrt{\alpha_L |k_j| + \alpha_H} } \ket{0}_R \ket{0}_a \frac{L(qh+lh/M_{\text{D}})}{\alpha_L}\ket{\psi} \nonumber\\ +\,& \ket{j}\ket{l} \ket{k_j}_k \frac{\sqrt{\alpha_H}}{\sqrt{\alpha_L |k_j| + \alpha_H} }  \ket{1}_R \ket{0}_a \frac{H(qh+lh/M_{\text{D}})}{\alpha_H}\ket{\psi} + \ket{\perp_a}. 
\end{align}
Here $\ket{\perp_a}$ represents an unnormalized part of the state that is orthogonal to $\ket{0}_a$ in the ancilla register labelled by $a$. 
We then apply \REV{$\text{c-}\text{R}^{\top}$ where $\cdot^{\top}$ denotes the matrix transpose,} giving
\begin{align}
    &\ket{j}\ket{l} \ket{k_j}_k  \ket{0}_R \ket{0}_a \left(\frac{\alpha_L k_j}{\alpha_L |k_j| + \alpha_H }\frac{L(qh+lh/M_{\text{D}})}{\alpha_L} + \frac{\alpha_H}{\alpha_L |k_j| + \alpha_H } \frac{H(qh+lh/M_{\text{D}})}{\alpha_H} \right)\ket{\psi} + \ket{\perp_{R,a}} \\
    =\,& \ket{j}\ket{l} \ket{k_j}_k  \ket{0}_R \ket{0}_a \frac{k_j L(qh+lh/M_{\text{D}}) + H(qh+lh/M_{\text{D}})}{\alpha_L |k_j| + \alpha_H } \ket{\psi} + \ket{\perp_{R,a}}.
\end{align}
This already gives a block-encoding of $k_jL+H$, but the block-encoding factor varies for different $k_j$. 
We achieve a uniform factor (namely, the worst-case $\alpha_L K + \alpha_H$) by appending another ancilla qubit $\ket{0}_{R'}$ and performing the controlled rotation
\begin{equation}
    \ket{k}\ket{0} \mapsto \ket{k} \left( \frac{\alpha_L |k| + \alpha_H}{\alpha_L K + \alpha_H} \ket{0} + \sqrt{1-\left| \frac{\alpha_L |k| + \alpha_H}{\alpha_L K + \alpha_H} \right|^2}\ket{1} \right). 
\end{equation}
Then we obtain the state 
\begin{equation}
    \ket{j}\ket{l} \ket{k_j}_k  \ket{0}_R \ket{0}_{R'} \ket{0}_a \frac{k_j L(qh+lh/M_{\text{D}}) + H(qh+lh/M_{\text{D}})}{\alpha_L K + \alpha_H } \ket{\psi} + \ket{\perp_{R,R',a}}. 
\end{equation}
Finally, applying $O_k^{\dagger}$ gives the state 
\begin{equation}\label{eqn:select_oracle_final_state}
    \ket{j}\ket{l} \ket{0}_k  \ket{0}_R \ket{0}_{R'} \ket{0}_a \frac{k_j L(qh+lh/M_{\text{D}}) + H(qh+lh/M_{\text{D}})}{\alpha_L K + \alpha_H } \ket{\psi} + \ket{\perp_{R,R',a}}. 
\end{equation}
By the definition of block-encoding, the sequence of operations mapping~\cref{eqn:select_oracle_ini_state} to~\cref{eqn:select_oracle_final_state} implements the HAM-T oracle of $kL+H$, $\text{HAM-T}_{kL+H,q}$, such that 
\begin{equation}\label{eqn:algorithm_hamt_kLH}
    (\bra{0}_{a'}\otimes I) \text{HAM-T}_{kL+H,q} (\ket{0}_{a'} \otimes I) = \sum_{j=0}^{M-1}\sum_{l=0}^{M_{\text{D}}-1} \ket{j}\bra{j} \otimes \ket{l}\bra{l} \otimes \frac{k_j L(qh+lh/M_{\text{D}}) + H(qh+lh/M_{\text{D}})}{\alpha_L K + \alpha_H }. 
\end{equation}
Notice that the number of the ancilla qubits increases by $\mathcal{O}(\log K) + 2$. 

By~\cite[Corollary 4]{LowWiebe2019}, we can implement 
\begin{equation}\label{eqn:algorithm_sel_oracle}
    \text{SEL} = \sum_{j=0}^{M-1} \ket{j}\bra{j} \otimes W_j
\end{equation}
using the truncated Dyson series method.
Here $W_j$ is a block-encoding of some $V_j \approx U(T,k_j)$ with $\|V_j\| \leq 1$.
Finally, the homogeneous time-evolution operator $\mathcal{T}e^{-\int_0^T A(s) \ud s} \approx \sum_{j=0}^{M-1} c_j U(T,k_j)$ can be block-encoded by $ (O_{c,l}^{\dagger} \otimes I) \text{SEL} (O_{c,r}\otimes I)$. 
Applying this circuit to the input state $\ket{0}\ket{u_0}$ gives 
\begin{equation}\label{eqn:algorithm_general_homo_part}
    (O_{c,l}^{\dagger} \otimes I) \text{SEL} (O_{c,r}\otimes I) \ket{0}\ket{u_0} = \frac{1}{\|c\|_1} \ket{0} \left(\sum_{j=0}^{M-1} c_j V_j\right)\ket{u_0} + \ket{\perp}, 
\end{equation}
in which the part corresponding to $\ket{0}$ approximates the solution of the homogeneous ODE.

\subsubsection{Full solution}\label{sec:algorithm_inhomo}

To obtain the full solution of the general inhomogeneous ODE including the term $b(t)$, we prepare the homogeneous and inhomogeneous parts separately and linearly combine them. 

We have already discussed the implementation of the homogeneous term. 
For the inhomogeneous part, we implement $U(T,s_{j'},k_j)$ for different $s_{j'}$ and perform LCU for both $j$ and $j'$. 
To deal with different time periods, we may write 
\begin{equation}
    U(T,s,k) = \mathcal{T}e^{-i\int_s^T (kL(s')+H(s')) \ud s'} = \mathcal{T}e^{-i\int_0^T (k 1_{s'\geq s} L(s') + 1_{s'\geq s} H(s')) \ud s'}. 
\end{equation}
Notice that HAM-T oracles for $1_{s'\geq s} L(s')$ and $1_{s'\geq s} H(s')$ can be directly constructed using one application of HAM-T oracles of $L$ and $H$, respectively, and a compare operation that returns, on an additional ancilla qubit, $0$ if $s'\geq s$ and $1$ if $s' < s$. 
Then, using the same approach as for the select oracle in~\cref{sec:algorithm_homo}, we may construct
\begin{equation}\label{eqn:algorithm_inhomo_sel_HS}
    \text{SEL}' = \sum_{j'=0}^{M'-1} \sum_{j=0}^{M-1} \ket{j'} \bra{j'} \otimes \ket{j} \bra{j} \otimes W_{j,j'}, 
\end{equation}
where $W_{j,j'}$ is the block-encoding of $V_{j,j'}$ and $V_{j,j'} \approx U(T,s_{j'},k_j)$. 
Notice that we may write the oracle $O_b$ as 
\begin{equation}
    O_b = \sum_{j'=0}^{M'-1} \ket{j'}\bra{j'} \otimes O_{b(s_{j'})}, 
\end{equation}
where $O_{b(s_{j'})}$ is a unitary mapping $\ket{0}$ to $\ket{b(s_{j'})}$. 
Then by first applying $O_{b}$ and then applying $\text{SEL}'$, we obtain the select oracle 
\begin{equation}
    \widetilde{\text{SEL}}' = \sum_{j'=0}^{M'-1} \sum_{j=0}^{M-1} \ket{j'} \bra{j'} \otimes \ket{j} \bra{j} \otimes \widetilde{W}_{j,j'}, 
\end{equation}
where $\widetilde{W}_{j,j'}$ is a block-encoding of  $V_{j,j'}O_{b(s_{j'})}$. 
Then the inhomogeneous term can be implemented by the standard LCU subroutine as 
\begin{equation}
    (O_{c',l}^{\dagger} \otimes O_{c,l}^{\dagger} \otimes I) \widetilde{\text{SEL}}' (O_{c',r} \otimes O_{c,r} \otimes I). 
\end{equation}
Applying this operator on the zero state yields 
\begin{equation}\label{eqn:algorithm_general_inhomo_part}
    \frac{1}{\|c\|_1\|c'\|_1} \ket{0} \left( \sum_{j'=0}^{M'-1} \sum_{j=0}^{M-1} c'_{j'}c_j V_{j,j'} \ket{b(s_{j'})} \right) + \ket{\perp}. 
\end{equation}

Finally we combine~\cref{eqn:algorithm_general_homo_part} and~\cref{eqn:algorithm_general_inhomo_part} by another LCU. 
Specifically, we use one more ancilla qubit, apply a single-qubit acting as 
\begin{equation}
    R: \ket{0} \mapsto  \frac{1}{\sqrt{\|u_0\|+ \|c'\|_1}}\left(\sqrt{\|u_0\|} \ket{0} + \sqrt{\|c'\|_1} \ket{1} \right), 
\end{equation}
prepare~\cref{eqn:algorithm_general_homo_part} controlled by this new ancilla qubit if it is $0$ and~\cref{eqn:algorithm_general_inhomo_part} if it is $1$, and finally apply $R^{\dagger}$ on the ancilla qubit. 
Combining all the ancilla registers as a single register, the output quantum state is
\begin{equation}\label{eqn:algorithm_numerical_solu}
    \frac{1}{\|c\|_1 (\|u_0\|+ \|c'\|_1) } \ket{0} v + \ket{\perp} 
\end{equation}
where 
\begin{equation}
    v = \sum_{j=0}^{M-1} c_jV_j \|u_0\| \ket{u_0} + \sum_{j'=0}^{M'-1}\sum_{j=0}^{M-1} c'_{j'}c_j V_{j,j'} \ket{b(s_{j'})}
\end{equation}
is an approximation of the exact solution $u(T)$ and might be an unnormalized vector. 
Postselecting the ancilla registers on $0$ yields the desired state.

\subsection{Complexity analysis}

We now present a complexity analysis for a general inhomogeneous ODE. 

\begin{thm}\label{thm:complexity_inhomo}
    Consider the inhomogeneous ODE system in~\cref{eqn:ODE}. 
    Suppose that $L(t)$ is positive semi-definite on $[0,T]$, and we are given the oracles described in~\cref{sec:oracles}. 
    Let $\|A(t)\| \leq \alpha_A$ and define $\Lambda = \sup_{p \geq 0, t \in [0,T]} \|A^{(p)}\|^{1/(p+1)} $ and $\Xi = \sup_{ p\geq 0, t \in [0,T] } \|b^{(p)}\|^{1/(p+1)} $, where the superscript $(p)$ indicates the $p$th-order time derivative. 
    Then we can prepare an $\epsilon$-approximation of the normalized solution $\ket{u(T)}$ with $\Omega(1)$ probability and a flag indicating success, by choosing 
    \begin{equation}
         M = \mathcal{O}\left( \alpha_A T\left(\log\left(\frac{\|u_0\|+\|b\|_{L^1}}{\|u(T)\| \epsilon}\right)\right)^{1+1/\beta} \right), \quad M' = \widetilde{\mathcal{O}}\left( T (\Lambda+\Xi) \left(\log\left(\frac{1+\|b\|_{L^1}}{\|u(T)\|\epsilon}\right)\right)^{1+1/\beta} \right), 
    \end{equation}
    using
        \begin{equation}
            \widetilde{\mathcal{O}}\left( \frac{\|u_0\|+\|b\|_{L^1}}{\|u(T)\|} \alpha_A T \left(\log\left(\frac{ 1 }{ \epsilon}\right)\right)^{1+1/\beta}  \right)
        \end{equation}
        queries to the $\text{HAM-T}$ oracle and 
            \begin{equation}
            \mathcal{O}\left( \frac{\|u_0\|+\|b\|_{L^1}}{\|u(T)\|} \right)
        \end{equation}
        queries to the state preparation oracles $O_u$ and $O_{b}$. 
        \REV{Here $\beta \in (0,1)$ is a tunable parameter in the kernel function in~\cref{eqn:kernel_exp}. }
\end{thm}
\begin{proof}
    According to the discussions in~\cref{sec:algorithm_homo}, the $\text{HAM-T}_{kL+H,q}$ oracle can be implemented with $\mathcal{O}(1)$ queries to $\text{HAM-T}_{L,q}$, and $\text{HAM-T}_{H,q}$. 
    By~\cite[Corollary 4]{LowWiebe2019}, for any $\epsilon_1 > 0$, we can implement 
    \begin{equation}
        \text{SEL} = \sum_{j=0}^{M-1} \ket{j}\bra{j} \otimes W_j, 
    \end{equation}
    where $W_j$ block-encodes $V_j$ with $\|V_j - U(T,k_j)\| \leq \epsilon_1$, by using $\text{HAM-T}_{L,q}$, and $\text{HAM-T}_{H,q}$ 
    \begin{equation}
        \mathcal{O}\left( (\alpha_LK+\alpha_H) T \frac{\log((\alpha_LK+\alpha_H) T/\epsilon_1)}{\log\log((\alpha_LK+\alpha_H) T/\epsilon_1)} \right) = \widetilde{\mathcal{O}} \left( \alpha_A K T \log(1/\epsilon_1) \right)
    \end{equation}
    times. 
    Then, $(O_{c,l}^{\dagger} \otimes I) \text{SEL} (O_{c,r}\otimes I)$ is a block-encoding of $\frac{1}{\|c\|_1}\sum_{j=0}^{M-1} c_j V_j$, and thus 
    \begin{equation}\label{eqn:complexity_proof_eq1}
        (O_{c,l}^{\dagger} \otimes I) \text{SEL} (O_{c,r}\otimes I) (O_u\otimes I)\ket{0}\ket{0} = \frac{1}{\|c\|_1} \ket{0} \left(\sum_{j=0}^{M-1} c_j V_j\right)\ket{u_0} + \ket{\perp}. 
    \end{equation}
    This step only needs $1$ query to $O_{c,l},O_{c,r},O_u$, and $\text{SEL}$.  
    
    By the same argument, we can prepare a single copy of~\cref{eqn:algorithm_general_inhomo_part} using $\mathcal{O}(1)$ queries to $O_{c,l},O_{c,r},O_{c',l},O_{c',r},O_b$ and
    \begin{equation}
        \mathcal{O}\left( (\alpha_LK+\alpha_H) T \frac{\log((\alpha_LK+\alpha_H) T/\epsilon_2)}{\log\log((\alpha_LK+\alpha_H) T/\epsilon_2)} \right) = \widetilde{\mathcal{O}} \left( \alpha_A K T \log(1/\epsilon_2) \right)
    \end{equation}
    queries to $\text{HAM-T}_L$ and $\text{HAM-T}_H$.
    Here $\epsilon_2$ is an upper bound on $\|V_{j,j'} - U(T,s_{j'},k_j)\|$. 
    Therefore, we can prepare a single copy of~\cref{eqn:algorithm_numerical_solu} by querying $O_{c,l},O_{c,r},O_{c',l},O_{c',r},O_u,O_b$ $\mathcal{O}(1)$ times, and querying $\text{HAM-T}_L$ and $\text{HAM-T}_H$
    \begin{equation}
        \widetilde{\mathcal{O}}\left( \alpha_A K T \log\left(\frac{1}{\min \left\{\epsilon_1,\epsilon_2\right\}}\right) \right)
    \end{equation}
    times.

    We now determine the choice of $\epsilon_1$ and $\epsilon_2$. 
    The error in the unnormalized solution can be bounded as 
    \begin{align}
            \left\| v - u(T)\right\|
            & \leq \left\| \sum_{j} c_jV_j u_0 -   \mathcal{T}e^{-\int_0^T A(s) \ud s} u_0 \right\| + \left\| \sum_{j,j'} c_{j'}c_j V_{j,j'} \ket{b(s_{j'})} -   \int_0^T \mathcal{T}e^{-\int_s^T A(s')\ud s'} b(s) \right\| \\
            & \leq \|c\|_1 \|u_0\| \epsilon_1 + \|u_0\| \left\| \sum_{j} c_j U(T,0,k_j) -   \mathcal{T}e^{-\int_0^T A(s) \ud s} \right\| \\
            & \quad + \|c\|_1\|c'\|_1 \epsilon_2 + \left\|  \sum_{j,j'} c_{j'}c_j U(T,s_{j'},k_j) \ket{b(s_{j'})} -   \int_0^T \mathcal{T}e^{-\int_s^T A(s')\ud s'} b(s) \right\|. 
    \end{align}
    According to~\cref{lem:truncation},~\cref{lem:quadrature}, and~\cref{lem:quadrature_error_inhomo}, we can bound the homogeneous (resp.~inhomogeneous) integral discretization error by $\epsilon_3/\|u_0\|$ (resp.~$\epsilon_3$) by choosing 
    \begin{equation}\label{eqn:proof_complexity_K}
        K = \mathcal{O}\left( \left(\log\left(\frac{\|u_0\|+\|b\|_{L^1}}{\epsilon_3}\right)\right)^{1/\beta} \right), \quad M = \mathcal{O}\left( \alpha_A T\left(\log\left(\frac{\|u_0\|+\|b\|_{L^1}}{\epsilon_3}\right)\right)^{1+1/\beta} \right), 
    \end{equation}
    and 
    \begin{equation}
        M' = \widetilde{\mathcal{O}}\left( T (\Lambda+\Xi) \left(\log\left(\frac{1+\|b\|_{L^1}}{\epsilon_3}\right)\right)^{1+1/\beta}  \right). 
    \end{equation}
    Then we have 
    \begin{equation}\label{eqn:proof_complexity_inhomo_eq1}
        \left\| v - u(T)\right\| \leq \|c\|_1 \|u_0\| \epsilon_1 +  \|c\|_1\|c'\|_1 \epsilon_2 + 2\epsilon_3, 
    \end{equation}
    and the error in the quantum state can be bounded as 
    \begin{equation}
        \left\| \ket{v} - \ket{u(T)}\right\| \leq \frac{2}{\|u(T)\|} \left\| v - u(T)\right\| \leq \frac{2\|c\|_1 \|u_0\|}{\|u(T)\|} \epsilon_1 +  \frac{2\|c\|_1\|c'\|_1}{\|u(T)\|} \epsilon_2 + \frac{4}{\|u(T)\|}\epsilon_3. 
    \end{equation}
    To bound this error by $\epsilon$, we can choose 
    \begin{equation}\label{eqn:proof_complexity_eps}
        \epsilon_1 = \frac{\|u(T)\|}{8\|c\|_1 \|u_0\|}\epsilon, \quad \epsilon_2 = \frac{\|u(T)\|}{8\|c\|_1\|c'\|_1}\epsilon, \quad \epsilon_3 = \frac{\|u(T)\|}{8}\epsilon. 
    \end{equation}
    These choices, together with~\cref{lem:coefficient_1norm} and~\cref{lem:coefficient_1norm_inhomo}, imply the choice of $M$ and $M'$ as stated in the theorem. A single run of the algorithm queries $\text{HAM-T}_L$ and $\text{HAM-T}_H$
    \begin{align}
        \widetilde{\mathcal{O}}\left( \alpha_A K T \log\left(\frac{1}{\min \left\{\epsilon_1,\epsilon_2\right\}}\right) \right) & = \widetilde{\mathcal{O}}\left( \alpha_A K T \log\left(\frac{ \|c\|_1(\|u_0\| + \|c'\|_1)}{ \|u(T)\| \epsilon}\right) \right) \label{eqn:proof_total_complexity_with_beta}\\
        & = \widetilde{\mathcal{O}}\left( \alpha_A T \left(\log\left(\frac{ \|u_0\| + \|b\|_{L^1}}{ \|u(T)\| \epsilon}\right)\right)^{1+1/\beta} \right)
    \end{align}
    times.

    As in~\cref{eqn:algorithm_numerical_solu}, the rescaling factor in front of the ``correct'' subspace is $\|v\|/(\|c\|_1(\|u_0\|+\|c'\|_1))$. 
    With amplitude amplification, the number of repetitions of the algorithm is $\mathcal{O}(\|c\|_1(\|u_0\|+\|c'\|_1)/\|v\|)$. 
    By~\cref{eqn:proof_complexity_inhomo_eq1} and the triangle inequality, we have 
    \begin{equation}
        \|v\| \geq \|u(T)\| - (\|c\|_1 \|u_0\| \epsilon_1 +  \|c\|_1\|c'\|_1 \epsilon_2 + 2\epsilon_3) = \|u(T)\| (1-\epsilon/2).  
    \end{equation}
    Therefore the number of repetitions is 
    \begin{equation}\label{eqn:proof_complexity_repeat}
        \mathcal{O}\left( \frac{\|c\|_1(\|u_0\|+\|c'\|_1)}{\|v\|} \right) = \mathcal{O}\left( \frac{\|u_0\|+\|b\|_{L^1}}{\|u(T)\|} \right). 
    \end{equation}
    Multiplying this factor by the complexities of a single run of the algorithm yields the claimed overall complexities. 
\end{proof}

\subsection{Special cases and applications}\label{sec:special_applications}

\subsubsection{Homogeneous case}

\cref{thm:complexity_inhomo} gives the query complexity of our LCHS algorithm in the most general inhomogeneous case. 
However, the homogeneous case ($b(t) \equiv 0$) is of particular interest. For example, this case can model the dynamics of so-called ``non-Hermitian Hamiltonians''~\cite{MatsumotoKawabataAshidaEtAl2020,OkumaKawabataShiozakiEtAl2020,BergholtzBudichKunst2021,ChenSongLado2023}.
For clarity, we explicitly write down the complexity of our algorithm applied to homogeneous ODEs in the following corollary, which can be directly obtained from~\cref{thm:complexity_inhomo} by discarding all the components of the inhomogeneous term and letting $\|b\|_{L^1} = 0$. 

\begin{cor}\label{cor:complexity_homo}
    Consider the homogeneous ODE system in~\cref{eqn:ODE} with $b(t) \equiv 0$. 
    Suppose that $L(t)$ is positive semi-definite, and we are given the oracles described in~\cref{sec:oracles}. 
    Let $\alpha_A \geq \|A(t)\|$ be the block-encoding factor of $A(t)$. 
    Then we can prepare an $\epsilon$-approximation of the normalized solution $\ket{u(T)}$ with $\Omega(1)$ probability and a flag indicating success, by choosing $ M = \mathcal{O}\left( \alpha_A T\left(\log\left(\frac{\|u_0\|}{\|u(T)\|\epsilon}\right)\right)^{1+1/\beta} \right)$, using
        \begin{equation}
            \widetilde{\mathcal{O}}\left( \frac{\|u_0\|}{\|u(T)\|} \alpha_A T \left(\log\left(\frac{1}{\epsilon}\right)\right)^{1+1/\beta} \right)
        \end{equation}
        queries to the $\text{HAM-T}$ oracle and
        \begin{equation}
            \mathcal{O}\left( \frac{\|u_0\|}{ \|u(T)\|}  \right)
        \end{equation}
        queries to the state preparation oracle $O_u$.
        \REV{Here $\beta \in (0,1)$ is a tunable parameter in the kernel function in~\cref{eqn:kernel_exp}. }
\end{cor}

\REV{Although the explicit time scaling of our algorithm is almost linear as shown in~\cref{cor:complexity_homo}, the actual time scaling also depends on the factor $\|u_0\|/\|u(T)\|$. It is possible for $\|u_0\|/\|u(T)\|$ to be independent of $T$, for example in the Hamiltonian simulation case where $L(t) \equiv 0$. On the other hand, if the matrix $L(t)$ is strictly positive definite, $\|u_0\|/\|u(T)\|$ is exponential in $T$, so the overall complexity is also exponential in $T$. However, such exponential scaling is unavoidable, because a lower bound $\Omega(\|u(0)\|/\|u(T)\|)$ has been proved for generic quantum ODE algorithms~\cite{FangLinTong2022,AnLiuWangEtAl2022}. }

\subsubsection{Time-independent case}\label{sec:timeindependent_improve}

Now we consider the time-independent case where $A(t) \equiv A$ is time-independent. 
We first discuss the homogeneous case where $b(t) \equiv 0$. 
Although~\cref{cor:complexity_homo} also applies to this special case, it is possible to design an algorithm with better asymptotic scaling by taking advantage of better time-independent Hamiltonian simulation algorithms such as QSP and QSVT. 

Specifically, in the time-independent case, we still use the algorithm presented in~\cref{sec:algorithm_homo}, but we can simplify the matrix oracles and use QSVT as the Hamiltonian simulation subroutine. 
First, the HAM-T oracles reduce to the standard block-encoding $\text{BE}_A$ of $A$ such that 
\begin{equation}
    \left( \bra{0}_a \otimes I \right) \text{BE}_A  \left( \ket{0}_a \otimes I \right) = \frac{A}{\alpha_A}. 
\end{equation}
Then, following~\cref{sec:algorithm_homo} up to~\cref{eqn:algorithm_hamt_kLH}, we again construct the block-encoding of $kL+H$ such that 
\begin{equation}
    (\bra{0}_{a'}\otimes I) \text{HAM-T}_{kL+H} (\ket{0}_{a'} \otimes I) = \sum_{j=0}^{M-1}\ket{j}\bra{j} \otimes \frac{k_j L + H}{\alpha_L K + \alpha_H }. 
\end{equation}
Now, we use QSVT~\cite[Corollary 60]{GilyenSuLowEtAl2019} to obtain the select oracle in~\cref{eqn:algorithm_sel_oracle}, and the homogeneous time-evolution operator can be block-encoded by LCU. 

The overall query complexity is as follows. 

\begin{cor}\label{cor:complexity_homo_time_independent}
    Consider the homogeneous ODE system in~\cref{eqn:ODE} with $b(t) \equiv 0$ and time-independent $A(t) \equiv A$.  
    Under the assumptions in~\cref{cor:complexity_homo}, we can prepare an $\epsilon$-approximation of the normalized solution $\ket{u(T)}$ with $\Omega(1)$ probability and a flag indicating success, using
        \begin{equation}
            \widetilde{\mathcal{O}}\left( \frac{\|u_0\|}{\|u(T)\|} \alpha_A T \left(\log\left(\frac{1}{\epsilon}\right)\right)^{1/\beta} \right)
        \end{equation}
        queries to the block-encoding of $A$, and
        \begin{equation}
            \mathcal{O}\left( \frac{\|u_0\|}{ \|u(T)\|}  \right)
        \end{equation}
        queries to the state preparation oracle $O_u$. 
        \REV{Here $\beta \in (0,1)$ is a tunable parameter in the kernel function in~\cref{eqn:kernel_exp}. }
\end{cor}
\begin{proof}
    This corollary can be proved in the same way as~\cref{thm:complexity_inhomo}. 
    In each run of LCU, we use $1$ query to $O_u$, and by~\cite[Corollary 60]{GilyenSuLowEtAl2019},
     \begin{equation}
        \mathcal{O}\left( (\alpha_LK+\alpha_H) T + \log(1/\epsilon_1)\right) = \mathcal{O} \left( \alpha_A K T + \log(1/\epsilon_1) \right)
    \end{equation}
    queries to the block-encoding of $A$. 
    Here the error tolerance $\epsilon_1$ of Hamiltonian simulation can be chosen as $\epsilon_1 = \frac{\|u(T)\|}{8\|c\|_1 \|u_0\|}\epsilon$ according to~\cref{eqn:proof_complexity_eps}. 
    Furthermore, the truncation threshold is $K = \mathcal{O}\left( \left(\log\left(\frac{\|u_0\|}{\|u(T)\|\epsilon}\right)\right)^{1/\beta} \right)$ according to~\cref{eqn:proof_complexity_K} and~\cref{eqn:proof_complexity_eps}. 
    Therefore, in each run of LCU, the query complexity to the block-encoding of $A$ is 
    \begin{equation}
        \mathcal{O} \left( \alpha_A T  \left(\log\left(\frac{\|u_0\|}{\|u(T)\|\epsilon}\right)\right)^{1/\beta} + \log\left(\frac{\|c\|_1 \|u_0\|}{\|u(T)\|\epsilon}\right) \right) = \mathcal{O} \left( \alpha_A T  \left(\log\left(\frac{\|u_0\|}{\|u(T)\|\epsilon}\right)\right)^{1/\beta} \right). 
    \end{equation}
    The number of repetitions until success is $\mathcal{O}(\|u_0\|/\|u(T)\|)$ according to~\cref{eqn:proof_complexity_repeat}, which contributes to another multiplicative factor and completes the proof. 
\end{proof}

For the inhomogeneous case, we must also prepare the state $\int_0^T e^{-A(T-s)} b(s) \ud s$. 
We follow the algorithm in~\cref{sec:algorithm_inhomo}, with the only difference being that we use QSVT for Hamiltonian simulation to construct the operator in~\cref{eqn:algorithm_inhomo_sel_HS}, where each unitary operator $W_{j,j'}$ is now the block-encoding of $e^{-i (T-s_{j'}) (k_j L + H) }$. 
To deal with different evolution times, we start with the $\text{HAM-T}_{kL+H}$ satisfying 
\begin{equation}
    (\bra{0}_{a}\otimes I) \text{HAM-T}_{kL+H} (\ket{0}_{a} \otimes I) = \sum_{j'=0}^{M'-1}\sum_{j=0}^{M-1}\ket{j'}\bra{j'} \otimes \ket{j}\bra{j} \otimes \frac{k_j L + H}{\alpha_L K + \alpha_H }, 
\end{equation}
and apply, on an additional qubit, a controlled rotation gate that maps the state $\ket{j'}\ket{0}$ to the state $\ket{j'}\left((1-s_{j'}/T)\ket{0} + \sqrt{1-(1-s_{j'}/T)^2}\ket{1}\right)$. 
This gives rise to the HAM-T oracle of the rescaled Hamiltonian $(1-s/T)(kL+H)$, with one additional ancilla qubit, satisfying 
\begin{equation}
    (\bra{0}_{a'}\otimes I) \text{HAM-T}_{(1-s/T)(kL+H)} (\ket{0}_{a'} \otimes I) = \sum_{j'=0}^{M'-1}\sum_{j=0}^{M-1}\ket{j'}\bra{j'} \otimes \ket{j}\bra{j} \otimes \frac{T-s_{j'}}{T(\alpha_L K + \alpha_H) } (k_j L + H). 
\end{equation}
Notice that we need to simulate $\frac{T-s_{j'}}{T(\alpha_L K + \alpha_H) } (k_j L + H)$ up to time $T(\alpha_L K + \alpha_H)$, which is independent of both $j$ and $j'$. 
Thus we may use QSVT with the phase factors associated with the polynomial approximation of $e^{-iT(\alpha_L K + \alpha_H)x}$ to implement the operator in~\cref{eqn:algorithm_inhomo_sel_HS}. 
The overall complexity is still dominated by the Hamiltonian simulation step, which takes advantage of the additive scaling of QSVT and scales as $\widetilde{\mathcal{O}}\left( \left(\log\left(\frac{1}{\epsilon}\right)\right)^{1/\beta} \right)$ in precision.

\subsubsection{Gibbs state preparation}

The goal of Gibbs state preparation is to prepare a density matrix $\frac{1}{Z_{\gamma}} e^{-\gamma L}$. 
Here $L \succeq 0$ is the Hamiltonian, $\gamma$ is the inverse temperature, and $Z_{\gamma} = \mathop{\mathrm{Tr}}(e^{-\gamma L})$ is the partition function. 
The Gibbs state can be obtained by first preparing the purified Gibbs state 
\begin{equation}\label{eqn:purified_gibbs_state}
    \ket{\psi} = \sqrt{\frac{N}{Z_{\gamma}}} \left(I \otimes e^{-\gamma L/2}\right) \left( \frac{1}{\sqrt{N}} \sum_{j=0}^{N-1} \ket{j}\ket{j} \right), 
\end{equation}
and then tracing out the first register. 
Therefore, we may prepare the Gibbs state by constructing the block-encoding of $e^{-\gamma L/2}$ via our LCHS method and applying it to the maximally entangled state.  

The query complexity of preparing the purified Gibbs state is as follows, as a direct consequence of~\cref{cor:complexity_homo_time_independent} (observing that the norm of the unnormalized solution $\left(I \otimes e^{-\gamma L/2}\right) \left( \frac{1}{\sqrt{N}} \sum_{j=0}^{N-1} \ket{j}\ket{j} \right)$ is $\sqrt{Z_{\gamma}/N}$). 

\begin{cor}\label{cor:gibbs_state}
    Suppose that $L \succeq 0$ is a Hamiltonian, and we are given the access to a block-encoding of $L$ with block-encoding factor $\alpha_L \geq \|L\|$. 
    Then we can prepare an $\epsilon$-approximation of the purified Gibbs state~\cref{eqn:purified_gibbs_state} with $\Omega(1)$ probability and a flag indicating success, using 
    \begin{equation}
            \widetilde{\mathcal{O}}\left( \sqrt{\frac{N}{Z_{\gamma}}} \gamma \alpha_L  \left(\log\left(\frac{1}{\epsilon}\right)\right)^{1/\beta} \right)
    \end{equation}
    queries to the block-encoding of $L$. 
    \REV{Here $\beta \in (0,1)$ is a tunable parameter in the kernel function in~\cref{eqn:kernel_exp}, and $\gamma$ is the inverse temperature. }
\end{cor}

\REV{
\subsubsection{Adaptive implementation}\label{sec:adaptive}
So far we have regarded $\beta$ as a fixed parameter in $(0,1)$. 
Now we consider the case where we can choose $\beta$ adaptively, i.e., as a function of the evolution time $T$, target error $\epsilon$, and other parameters of the problem. 
Then we can slightly improve the error dependence of the total query complexity, though the number of queries to the state preparation oracle becomes higher. 
}

\REV{
\begin{cor}\label{cor:complexity_adaptive}
    Consider the inhomogeneous ODE system in~\cref{eqn:ODE}. 
    Suppose that $L(t)$ is positive semi-definite on $[0,T]$, and we are given the oracles described in~\cref{sec:oracles}. 
    Let $\alpha_A \geq \|A(t)\|$ be the block-encoding factor of $A(t)$.  
    Then, by choosing
    \begin{equation}
         \beta = 1 - \mathcal{O}\left( \frac{1}{\log\log\left( \frac{\|u_0\|+\|b\|_{L^1}}{\|u(T)\|\epsilon} \right)} \right), 
    \end{equation}
    we can prepare an $\epsilon$-approximation of the normalized solution $\ket{u(T)}$ with $\Omega(1)$ probability and a flag indicating success using
        \begin{equation}
            \widetilde{\mathcal{O}}\left( \frac{\|u_0\|+\|b\|_{L^1}}{\|u(T)\|} \alpha_A T \left(\log\left(\frac{ 1 }{ \epsilon}\right)\right)^2  \right)
        \end{equation}
        queries to the $\text{HAM-T}$ oracle and 
            \begin{equation}
            \widetilde{\mathcal{O}}\left( \frac{\|u_0\|+\|b\|_{L^1}}{\|u(T)\|} \log\log\left( \frac{1}{\epsilon} \right) \right)
        \end{equation}
        queries to the state preparation oracles $O_u$ and $O_{b}$. 
\end{cor}
}
\begin{proof}
    \REV{The idea is to track the $\beta$ dependence in the constant factor and find an optimal value to minimize the total query complexity. 
    In this proof, we use the big-O notation without hiding any extra $\beta$ dependence.}

    \REV{
    According to~\cref{eqn:proof_total_complexity_with_beta} and~\cref{eqn:proof_complexity_repeat} in the proof of~\cref{thm:complexity_inhomo}, in each run we query the HAM-T oracle $\widetilde{\mathcal{O}}\left( \alpha_A K T \log\left(\frac{ \|c\|_1(\|u_0\| + \|c'\|_1)}{ \|u(T)\| \epsilon}\right) \right)$ times, and the number of repetitions is $\mathcal{O}\left( \frac{\|c\|_1(\|u_0\|+\|c'\|_1)}{\|u(T)\|} \right)$ with amplitude amplification. 
    As a result, the total query complexity of our algorithm is 
    \begin{equation}
        \widetilde{\mathcal{O}}\left( \frac{\|c\|_1(\|u_0\|+\|c'\|_1)}{\|u(T)\|} \alpha_A K T \log\left(\frac{ \|c\|_1(\|u_0\| + \|c'\|_1)}{ \|u(T)\| \epsilon}\right) \right). 
    \end{equation}
    We only need to track the explicit $\beta$ dependence in the parameters $K$, $\|c\|_1$ and $\|c'\|_1$. 
    }

    \REV{According to the proof of~\cref{thm:complexity_inhomo}, the parameter $K$ is chosen such that the quadrature error related to the variable $k$ is bounded by $\mathcal{O}\left( \frac{\|u(T)\|\epsilon}{\|u_0\|+\|b\|_{L^1}} \right)$. 
    Using the first part of~\cref{lem:quadrature}, the choice of $K$ should satisfy 
    \begin{equation}
        \frac{2^{\left\lceil 1/\beta \right\rceil+1} \left\lceil 1/\beta \right\rceil !}{C_{\beta} \left(\cos(\beta\pi/2)\right)^{\left\lceil 1/\beta \right\rceil} } \frac{1}{ K } e^{-\frac{1}{2}K^{\beta} \cos(\beta\pi/2) } \leq \mathcal{O}\left( \frac{\|u(T)\|\epsilon}{\|u_0\|+\|b\|_{L^1}} \right). 
    \end{equation}
    Supposing $\beta \geq 1/2$, it suffices to choose 
    \begin{align}
        K &= \mathcal{O}\left( \frac{1}{(\cos(\beta\pi/2))^{1/\beta}} \left(\log\left(\frac{1}{\cos(\beta\pi/2)}\frac{\|u_0\|+\|b\|_{L^1}}{\|u(T)\|\epsilon}\right)\right)^{1/\beta} \right) \\
        & \leq \mathcal{O}\left( \frac{1}{(1-\beta)^{1/\beta}} \left(\log\left(\frac{1}{1-\beta}\frac{\|u_0\|+\|b\|_{L^1}}{\|u(T)\|\epsilon}\right)\right)^{1/\beta} \right). 
    \end{align}
    Furthermore, according to~\cref{lem:coefficient_1norm} (specifically,~\cref{eqn:bound_int_gabs} in its proof) and~\cref{lem:coefficient_1norm_inhomo}, $\|c\|_1 = \mathcal{O}\left(\frac{1}{1-\beta}\right)$ and $\|c'\|_1 = \mathcal{O}(\|b\|_{L^1})$. 
    Thus the total query complexity is 
    \begin{equation}\label{eqn:proof_adaptive_eq1}
        \widetilde{\mathcal{O}}\left( \frac{1}{(1-\beta)^{1+1/\beta}} \frac{\|u_0\|+\|b\|_{L^1}}{\|u(T)\|} \alpha_A T \left(\log\left(\frac{1}{1-\beta}\frac{\|u_0\|+\|b\|_{L^1}}{\|u(T)\|\epsilon}\right)\right)^{1+1/\beta} \right). 
    \end{equation}
    }
    
    \REV{
    Now we choose 
    \begin{equation}
        \beta = 1-x, \quad x = \mathcal{O}\left( \frac{1}{\log\log\left( \frac{\|u_0\|+\|b\|_{L^1}}{\|u(T)\|\epsilon} \right)} \right). 
    \end{equation}
    Then we have 
    \begin{equation}
        \frac{1}{(1-\beta)^{1+1/\beta}} = \left(\frac{1}{x}\right)^{1+\frac{1}{1-x}} \leq \frac{1}{x^2} \left(\frac{1}{x}\right)^{\mathcal{O}(x)} = \mathcal{O}\left( \frac{1}{x^2} \right) = \mathcal{O}\left( \left(\log\log\left( \frac{\|u_0\|+\|b\|_{L^1}}{\|u(T)\|\epsilon} \right)\right)^2 \right), 
    \end{equation}
    and (letting $\mathcal{L} = \log\left( \frac{\|u_0\|+\|b\|_{L^1}}{\|u(T)\|\epsilon} \right)$)
    \begin{align}
        \left(\log\left(\frac{1}{1-\beta}\frac{\|u_0\|+\|b\|_{L^1}}{\|u(T)\|\epsilon}\right)\right)^{1+1/\beta}
        & = \mathcal{O}\left( \left(\log\left( \frac{\|u_0\|+\|b\|_{L^1}}{\|u(T)\|\epsilon}\right)\right)^{2+\mathcal{O}(x)} \right) \\
        & = \mathcal{O}\left( \left(\log\left( \frac{\|u_0\|+\|b\|_{L^1}}{\|u(T)\|\epsilon}\right)\right)^{2} \mathcal{L}^{\mathcal{O}(1/\log\mathcal{L})} \right) \\
        & = \mathcal{O}\left( \left(\log\left( \frac{\|u_0\|+\|b\|_{L^1}}{\|u(T)\|\epsilon}\right)\right)^{2} \right).  
    \end{align}
    Plugging these two estimates back into~\cref{eqn:proof_adaptive_eq1} yields the claimed total query complexity. 
    }
    
    \REV{For the state preparation cost, notice that in each run of the algorithm we only use $\mathcal{O}(1)$ queries, so the scaling is simply the number of repetitions, which is 
    \begin{equation}
        \mathcal{O}\left( \frac{\|c\|_1(\|u_0\|+\|c'\|_1)}{\|u(T)\|} \right) = \mathcal{O}\left( \frac{1}{1-\beta} \frac{\|u_0\|+\|b\|_{L^1}}{\|u(T)\|} \right) = \mathcal{O}\left( \frac{\|u_0\|+\|b\|_{L^1}}{\|u(T)\|} \log\log\left( \frac{\|u_0\|+\|b\|_{L^1}}{\|u(T)\|\epsilon} \right) \right). 
    \end{equation}
    }
\end{proof}

\section{Hybrid implementation}\label{sec:hybrid}

We have discussed how to use the LCHS formula for general inhomogeneous ODEs with a quantum LCU subroutine. 
\REV{The algorithm achieves near-optimal total query complexity with respect to all parameters and optimal state preparation cost. 
However, to implement the algorithm, we use multiple additional ancilla qubits and coherently controlled Hamiltonian simulation in the construction of the select oracle. 
Though feasible on an ideal quantum computer, both of these ingredients pose computational challenges in the early fault-tolerant regime in which we have few logical qubits and limited capability to handle complicated control structures.}

\REV{For an early fault-tolerant implementation, as discussed in~\cite{AnLiuLin2023}, we may also consider a hybrid quantum-classical algorithm based on importance sampling.}
With the improved LCHS, we can follow the same hybrid approach as discussed in~\cite{AnLiuLin2023}; the only nuance is that now the coefficients $c_j$'s are not real positive, so we need to further split them into real and imaginary parts, estimate two summations separately, and take the sign of the real and imaginary parts into consideration. 
For completeness, in this section we present this hybrid algorithm and describe its sample complexity. 

For simplicity, here we only consider the homogeneous case and suppose that $\|u_0\| = 1$. 
\REV{Unlike the quantum algorithm, the goal in the hybrid algorithm is to estimate the (possibly unnormalized) observable $u(T)^{\dagger}Ou(T)$, where $O$ is a known Hermitian matrix. 
We assume access to a block-encoding $U_O$ of $O$ with block-encoding factor $\alpha_O \geq \|O\|$. }
By~\cref{eqn:LCHS_LCU_composite}, we can approximate the solution as 
\begin{equation}
    u(T) = \mathcal{T} e^{-\int_0^T A(s) \ud s} \ket{u_0} \approx \sum_{j=0}^{M-1} c_j U(T,k_j) \ket{u_0}. 
\end{equation}
The observable $u(T)^{\dagger} O u(T)$ can be written as 
\begin{align}
    u(T)^{\dagger} O u(T) &\approx \sum_{l=0}^{M-1}\sum_{j=0}^{M-1} \overline{c}_l c_j \braket{ u_0 |U(T,k_l)^{\dagger} O U(T,k_j)| u_0 } \\
    & = \sum_{j,l} \abs{\Re(\overline{c}_l c_j)} \sgn(\Re(\overline{c}_l c_j)) \braket{ u_0 |U(T,k_l)^{\dagger} O U(T,k_j)| u_0 } \nonumber \\
    & \quad + i \sum_{j,l} \abs{\Im(\overline{c}_l c_j) }  \sgn(\Im(\overline{c}_l c_j)) \braket{ u_0 |U(T,k_l)^{\dagger} O U(T,k_j)| u_0 }. 
\end{align}
Here $\sgn$ refers to the sign function of a real number. 

\REV{The idea of the hybrid algorithm is to estimate the observable $\braket{ u_0 |U(T,k_l)^{\dagger} O U(T,k_j)| u_0 }$ independently of the index on a quantum computer, and then linearly combine them on a classical computer.
The key steps of the algorithm are as follows: 
\begin{enumerate}
    \item Classically sample $(j,l)$ with probability $\propto \abs{\Re(\overline{c}_l c_j)}/\Gamma$ and $(j',l')$ with probability $\propto \abs{\Im(\overline{c}_{l'} c_{j'})}/\Gamma'$, where $\Gamma = \sum_{j,l} \abs{\Re(\overline{c}_l c_j)}$ and $\Gamma' = \sum_{j',l'} \abs{\Im(\overline{c}_{l'} c_{j'}) } $. 
    \item For sampled indices, estimate $o_{j,l} = \braket{ u_0 |U(T,k_l)^{\dagger} O U(T,k_j)| u_0 }$ and $o_{j',l'} = \braket{ u_0 |U(T,k_{l'})^{\dagger} O U(T,k_{j'})| u_0 }$ on a quantum computer, by the Hadamard test for non-unitary matrices~\cite{TongAnWiebe2021} and amplitude estimation~\cite{BrassardHoyerMoscaEtAl2002}. 
    \item Classically compute $\sigma$ and $\sigma'$, which are the averages of sampled $\sgn(\Re(\overline{c}_l c_j)) o_{j,l}$ and $ \sgn(\Im(\overline{c}_l c_j)) o_{j',l'}$, respectively. 
    \item Output $\Gamma \sigma + i \Gamma' \sigma'$ as an estimate of $u(T)^{\dagger} O u(T)$. 
\end{enumerate}
}

\REV{Compared with the quantum algorithm, this hybrid implementation has the advantage of only using one additional ancilla qubit in each quantum circuit and avoiding $k$-controlled implementation of $U(T,k)$. However, its overall query complexity becomes worse.}
The sample complexity can be characterized as follows (rephrasing~\cite[Theorem 9]{AnLiuLin2023}). 

\begin{prop}\label{prop:LCHS_hybrid_complexity}
    Suppose that $O_{u}$ is a state preparation oracle for $\ket{u_0}$, $U_O$ is a block-encoding of $O$ with block-encoding factor $\alpha_O \geq \norm{O}$, and $\widetilde{U}(T,k)$ is a quantum circuit that approximates $U(T,k)$ with error bounded by $\epsilon_{\text{HS}}$ for any $k$. 
    Then $u(t)^{*} O u(t)$ can be estimated with error at most $\epsilon$ and probability at least $1-\delta$ by choosing $\epsilon_{\text{HS}} = \mathcal{O}(\epsilon/\norm{O})$. Furthermore, 
    \begin{enumerate}
        \item the number of the  samples is
        \begin{equation}
            \mathcal{O}\left( \frac{ \norm{O}^2 }{\epsilon^2} \log\left(\frac{1}{\delta}\right) \right), 
        \end{equation}
        and
        \item each circuit with sampling value $(k,k')$ uses   
        \begin{equation}
            \mathcal{O}\left(\frac{\alpha_O}{\epsilon} \log\left(\frac{\alpha_O}{\epsilon}\right)\log\left(\frac{\norm{O} }{\delta \epsilon} \right)\right) 
        \end{equation}
        queries to $O_{u}$ and controlled versions of $U_O$ and $\widetilde{U}(T,k)$. 
    \end{enumerate}
\end{prop}

\bibliographystyle{unsrt}
\bibliography{LCHS}

\clearpage
\appendix

%

\section{Preliminaries}\label{app:prelim}

\subsection{Stability conditions}\label{app:stability}

Throughout this work, we make the assumption that the real part of $A(t)$, defined as $L(t) = (A(t) + A(t)^{\dagger})/2$, is positive semi-definite. 
This condition implies that the solution is non-increasing. 
Specifically, when $A(t)=A$ is time-independent, the norm of the propagator $\norm{e^{-At}}$ can be bounded by $\norm{e^{-Lt}}$.
\begin{lem}[{\cite[Theorem IX.3.1]{Bhatia1997}}]
\label{lem:expA_bound}
Assume $A=L+iH$ with $L = \frac{A+A^{\dagger}}{2}, \quad H = \frac{A-A^{\dagger}}{2i}$. Then for any $t\in \mathbb{R}$,
\begin{equation}
\norm{e^{-At}} \le \norm{e^{-Lt}}.
\end{equation}
\end{lem}

\begin{proof}
For some $r>0$, we use the Trotter product formula to give
\begin{align}
\norm{e^{-At}}=&\norm{(e^{-A t/r})^r}\\
=&\norm{\left(e^{-Lt/r}e^{-iHt/r}+\Or(t^2/r^2)\right)^r}\\
\le & \norm{\left(e^{-Lt/r}e^{-i Ht/r}\right)^r}+\Or(t^2/r)\\
\le & \norm{e^{-Lt/r}}^r \norm{e^{-iHt/r}}^r+\Or(t^2/r)\\
=& \norm{e^{-Lt}}+\Or(t^2/r).
\end{align}
Here we have used that $\norm{e^{-iHt/r}}=1$, and $\norm{e^{-Lt/r}}^r=\norm{\left(e^{-Lt/r}\right)^r}=\norm{e^{-Lt}}$. The lemma follows by taking $r\to \infty$.
\end{proof}

In particular, when $L\succeq 0$ and $t>0$, $\norm{e^{-At}}\le 1$. This means that when $b=0$, the norm of solution $\norm{u(t)}=\norm{e^{-At}u_0}$ is non-increasing. This result can be generalized to the time-dependent setting.

\begin{lem} \label{lem:expAt_bound}
Let $A(t)$ be decomposed according to~\cref{eqn:A_cartesian_1,eqn:A_cartesian_2}. If $L(s)\succeq 0$ for all $0\le s\le t$, then
\begin{equation}
\norm{\mathcal{T}e^{-\int_0^t A(s) \ud s}}\le 1.
\end{equation}
\end{lem}

\begin{proof}
The proof is very similar to that of \cref{lem:expA_bound}. 
We uniformly partition the interval $[0,t]$ into $0=t_0<t_1<\cdots<t_r=t$ so that $t_{j+1}-t_j=t/r$ for $j=0,\ldots,r-1$, and use the time-dependent Trotter product formula to give
\begin{align}
\norm{\mathcal{T}e^{-\int_0^t A(s) \ud s}}&=\norm{\prod_{j=0}^{r-1}\mathcal{T}e^{-\int_{t_j}^{t_{j+1}} A(s) \ud s}}\\
&=\norm{\prod_{j=0}^{r-1}\left(e^{-\int_{t_j}^{t_{j+1}} L(s) \ud s}e^{-i\int_{t_j}^{t_{j+1}} H(s) \ud s}+\Or(t^2/r^2)\right)}\\
&\le \prod_{j=0}^{r-1}\norm{e^{-\int_{t_j}^{t_{j+1}} L(s) \ud s}}\norm{e^{-i\int_{t_j}^{t_{j+1}} H(s) \ud s}}+\Or(t^2/r)\\
&= \prod_{j=0}^{r-1}\norm{e^{-\int_{t_j}^{t_{j+1}} L(s) \ud s}}+\Or(t^2/r)\\
&\le 1 +\Or(t^2/r).
\end{align}
The lemma follows by taking $r\to \infty$.
\end{proof}

\subsection{Block-encoding}\label{app:block_encoding}

Block-encoding is a framework of representing a general matrix on quantum devices. 
We first give its formal definition. 
\begin{defn}[Block-encoding]
    Let $M$ be a $2^{n}\times 2^n$ matrix. 
    Then a $2^{n+a}\times 2^{n+a}$ dimensional unitary $U_M$ is called the block-encoding of $M$, if $M/\alpha_M = \bra{0}^{\otimes a} U_M \ket{0}^{\otimes a}$. 
    Here $\alpha_M$ is called the block-encoding factor and should satisfy $\alpha_M \geq \|M\|$. 
\end{defn}
Intuitively, the block-encoding $U_M$ of $M$ encodes the rescaled matrix $M/\alpha_M$ in its upper left block as 
\begin{equation}
    U_M = \left( \begin{array}{cc}
        \frac{M}{\alpha_M} & * \\
        * & *
    \end{array} \right). 
\end{equation}
Block-encoding is a powerful model for quantum algorithm design as it enables us to perform matrix operations beyond unitary matrices. 
While it remains an open question how to construct the concrete circuit for the block-encoding of an arbitrarily given matrix, there have been substantial efforts on implementing the block-encoding of specific matrices of practical interest~\cite{GilyenSuLowEtAl2019,CladerDalzellStamatopoulosEtAl2022,NguyenKianiLloyd2022,camps2023explicit,sunderhauf2023blockencoding}. 
In this work, we simply assume the query access to the block-encoding of certain matrices as our input model.

\subsection{Linear combination of unitaries}\label{app:LCU}

LCU is a powerful quantum primitive which, as its name suggests, implements a block-encoding of linear combination of unitary operators. 
Here we follow~\cite{GilyenSuLowEtAl2019} and present an improved version of LCU that is applicable to possibly complex coefficients. 

Specifically, the goal of LCU is to implement the operator $\sum_{j=0}^{J-1} c_j U_j$, where $U_j$'s are unitary operators and $c_j$'s are complex numbers. 
For the coefficients, we assume access to a pair of state preparation oracles $(O_{c,l}, O_{c,r})$ acting as 
\begin{align}
    O_{c,l}&: \ket{0} \rightarrow \frac{1}{\sqrt{\|c\|_1}} \sum_{j=0}^{J-1} \overline{\sqrt{c_j}} \ket{j}, \\
    O_{c,r}&: \ket{0} \rightarrow \frac{1}{\sqrt{\|c\|_1}} \sum_{j=0}^{J-1} \sqrt{c_j} \ket{j}. 
\end{align}
Here $\sqrt{z}$ refers to the principal value of the square root of $z$, $\overline{z}$ denotes the conjugate of $z$, and $\|c\|_1$ denotes the $1$-norm of the vector $c = (c_0,\cdots,c_{J-1})$. 
For the unitaries, we assume so-called the \emph{select oracle} $\mathop{\mathrm{SEL}}$ as 
\begin{equation}
    \mathop{\mathrm{SEL}} = \sum_{j=0}^{J-1} \ket{j}\bra{j} \otimes U_j. 
\end{equation}
As shown in~\cite[Lemma 52]{GilyenSuLowEtAl2019}, we may implement $\sum_{j=0}^{J-1} c_j U_j$ as follows. 
\begin{lem}[LCU]
    Let $W =  (O_{c,l}^{\dagger} \otimes I) \mathop{\mathrm{SEL}} (O_{c,r}\otimes I)$. Then 
    \begin{equation}
        W \ket{0}\ket{\psi} = \frac{1}{\|c\|_1} \ket{0} \left(\sum_{j=0}^{J-1} c_j U_j\right) \ket{\psi} + \ket{\perp}. 
    \end{equation}
    Here $\ket{\perp}$ is an unnormalized vector such that $(\ket{0}\bra{0} \otimes I)\ket{\perp} = 0$. 
\end{lem}
According to the definition of block-encoding, $W$ is indeed a block-encoding of $\sum_{j=0}^{J-1} c_j U_j$ with block-encoding factor $\|c\|_1$. 
The complexity of LCU is mainly affected by two factors. 
First, constructing the operator $W$ requires one query to each oracle of $O_{c,l},O_{c,r}$, and $\mathop{\mathrm{SEL}}$, and constructing the select oracle usually dominates the computational cost. 
Second, to extract the information of $\sum_{j=0}^{J-1} c_j U_j$, we need to post-select the ancilla register on $0$. 
For a constant-level success probability, we need to run $\mathcal{O}(\|c\|_1)$ rounds of the amplitude amplification, and thus the operator $W$ needs to be implemented for $\mathcal{O}(\|c\|_1)$ times.

\subsection{Numerical quadrature}\label{app:quadrature}

Numerical quadrature is used to numerically approximate the integral $\int_a^b g(x) \ud x$. 
In general, it takes the form as a weighted summation of the function evaluated at a finite set of nodes $x_j\in [a,b]$, i.e., 
\begin{equation}
    \int_a^b g(x) \approx \sum_{j=0}^{m-1} w_j g(x_j). 
\end{equation}
The Riemann sum, one of the simplest quadrature rules, calculates the average of a function evaluated at equidistant nodes. To achieve better convergence in this work, we employ the composite Gaussian quadrature for discretizing the integrals.

We first state the Gaussian quadrature for integral over $[-1,1]$. 
Let $P_m(x)$ denote the Legendre polynomial of degree $m$. 
In Gaussian quadrature with $m$ nodes, we choose the nodes $x_j$ to be the roots of $P_m(x)$, and the weights $w_j = \frac{2}{(1-x_j^2) (P_m'(x_j))^2 }$. 
Gaussian quadrature becomes exact if $g(x)$ is a polynomial of degree at most $2m-1$, and thus it exhibits high-order convergence if the integrand can be well approximated by polynomials. 

For integral over a long interval $[a,b]$, we may use the composite rule that we divide the interval $[a,b]$ into $K$ segments $[a+(k-1)h, a+kh]$, where $h = \frac{b-a}{K}$ is the step size. 
We use the Gaussian quadrature to approximate the integral on each short interval $[a+(k-1)h, a+kh]$, and sum them up to approximate the integral on the entire interval $[a,b]$. 
Notice that the Gaussian quadrature for each short interval $[a+(k-1)h, a+kh]$ can be constructed from that on $[-1,1]$ simply by the change of variable, resulting in the change of weights $w_j \rightarrow \frac{h}{2}w_j$ and nodes $x_j \rightarrow \frac{h}{2} x_j + a + \frac{2k-1}{2}h$.

\section{Related works on Gibbs state preparation}\label{app:Gibbs}

In this section we summarize some previous quantum algorithms for Gibbs state preparation and compare them with our LCHS algorithm. 
A summary is presented in~\cref{tab:Gibbs}. 
\REV{We note that the approaches in~\cite{BerryChildsOstranderEtAl2017,Krovi2022} (and other quantum ODE solvers based on linear systems of equations) can also solve the Gibbs state preparation problem, by treating~\cref{eqn:purified_gibbs_state} as an ODE problem. 
However, they would have worse query complexity to the state preparation oracle as in the general case that we have discussed in the main text, so we exclude them from the comparison here. 
Additionally, the linear-system-based approaches first prepare a history state that encodes unwanted intermediate states, while LCHS is a more direct approach that avoids this history state.}

\begin{table}[t]
    \renewcommand{\arraystretch}{2}
    \centering
    \scalebox{0.95}{
    \begin{tabular}{c|c|c|c}\hline\hline
        \textbf{Method/Work} & \textbf{Assumption} & \textbf{Query access} & \textbf{Complexity} \\\hline
        Phase estimation~\cite{poulin2009sampling} 
        & $L \succ 0$ 
        & Hamiltonian simulation of $L$
        & $\widetilde{\mathcal{O}}\left( \sqrt{\frac{N}{Z_{\gamma}}} \gamma \alpha_L / \epsilon \right)$ \\\hline
        QSVT~\cite{GilyenSuLowEtAl2019} 
        & $L \succeq 0$ 
        & Block-encoding of $I-L/\alpha_L$
        & $\widetilde{\mathcal{O}}\left( \sqrt{\frac{N}{Z_{\gamma}}} \sqrt{\gamma \alpha_L} \left(\log\left(\frac{1}{\epsilon}\right)\right)^2  \right)$ \\\hline
        QSVT~\cite{GilyenSuLowEtAl2019,AnLiuWangEtAl2022} 
        & $L \succeq 0$
        & Block-encoding of $\sqrt{L}$
        & $\widetilde{\mathcal{O}}\left( \sqrt{\frac{N}{Z_{\gamma}}} \sqrt{\gamma \alpha_L \log\left(\frac{1}{\epsilon}\right)}  \right)$ \\\hline
        LCU~\cite{vanApeldoorn2020quantum} 
        & $L \succeq I$ 
        & Hamiltonian simulation of $L$
        & $\widetilde{\mathcal{O}}\left( \sqrt{\frac{N}{Z_{\gamma}}} \gamma \alpha_L \log\left(\frac{1}{\epsilon}\right) \right)$  \\\hline
LCU~\cite{ChowdhurySomma2016,ApersChakrabortyNovoEtAl2022}
        & $L \succeq 0$
        & Hamiltonian simulation of $\sqrt{L}$
        & $\widetilde{\mathcal{O}}\left( \sqrt{\frac{N}{Z_{\gamma}}} \sqrt{\gamma \alpha_L \log\left(\frac{1}{\epsilon}\right)}  \right)$ \\\hline
        \multirow{2}{4em}{LCU~\cite{HolmesMuraleedharanSommaEtAl2022}} 
        & general 
        & Hamiltonian simulation of $L$
        & $\widetilde{\mathcal{O}}\left( e^{\left( \frac{3}{\epsilon}+\frac{1}{2}\right) \gamma \alpha_L} \right)$ \\
        \cline{2-4} 
         & $L \succeq 0$ 
         & Hamiltonian simulation of $L$
         & $\widetilde{\mathcal{O}}\left( \sqrt{\frac{N}{Z_{\gamma}}}  (\gamma\alpha_L)^{3/2} e^{\sqrt{\log(1/\epsilon)}} \right)$ \\\hline 
        This work (\cref{cor:gibbs_state}) 
        & $L \succeq 0$ 
        & Hamiltonian simulation of $L$
        & $\widetilde{\mathcal{O}}\left( \sqrt{\frac{N}{Z_{\gamma}}} \gamma \alpha_L \left(\log\left(\frac{1}{\epsilon}\right)\right)^{1/\beta} \right)$ \\\hline\hline 
    \end{tabular}}
    \caption{Comparison among the improved LCHS and previous methods for Gibbs state preparation. Here $L$ is the Hamiltonian, $\alpha_L \geq \|L\|$ is the normalization factor of $L$, $\gamma$ is the inverse temperature, $N$ is the dimension of $L$, $Z_{\gamma}$ is the partition function, and $\epsilon$ is the allowed error. 
    Note that since~\cite{vanApeldoorn2020quantum} assumes $L \succeq I$, the post-selection factor $\sqrt{N/Z_{\gamma}}$ scales as $\mathcal{O}(\sqrt{N} e^{\gamma/2} )$ in the best case. 
    }
    \label{tab:Gibbs}
\end{table}

References~\cite{poulin2009sampling,vanApeldoorn2020quantum} propose quantum algorithms for Gibbs state preparation by applying the operator $e^{-\gamma L/2}$ to the maximally entangled state. 
The exponential operator is implemented by quantum phase estimation in~\cite{poulin2009sampling}, and by LCU in~\cite{vanApeldoorn2020quantum}. 
Both works require the Hamiltonian $L$ to be positive definite and efficiently simulated, and scale linearly in the inverse time $\gamma$ and the spectral norm $\alpha_L$ of the Hamiltonian. 
As a comparison, our algorithm works under a weaker assumption that $L$ only needs to be positive semi-definite and achieves comparable complexity. 

After these two early works, there have been several improved Gibbs state preparation algorithms that achieve sublinear scalings in $\alpha_L$ and/or $\gamma$. 
For positive semi-definite Hamiltonians, reference~\cite{GilyenSuLowEtAl2019} gives an algorithm with quadratic improvement in terms of $\gamma$ and $\alpha_L$. 
The algorithm assumes the block-encoding of $I-L/\alpha_L$, and implements $e^{-\gamma L} = e^{-\gamma\alpha_L (I-(I-L/\alpha_L))}$ using QSVT for a polynomial that approximates $e^{- \gamma\alpha_L (1-x)}$. 
The quadratic improvement arises because there is such a polynomial with degree only $\mathcal{O}(\sqrt{\gamma\alpha_L})$. 
Such a quadratic improvement is also possible when we have query access to the square root $\sqrt{L}$ of the Hamiltonian. 
In this case, we need to implement the function $e^{-\gamma x^2}$. 
This can be done by applying the Hubbard-Stratonovich transformation to write it as a linear combination of unitaries~\cite{ChowdhurySomma2016,ApersChakrabortyNovoEtAl2022}, or by directly approximating it by an $\mathcal{O}(\sqrt{\gamma})$-degree polynomial and applying QSVT~\cite{GilyenSuLowEtAl2019,AnLiuWangEtAl2022}. 
Compared to these algorithms, our LCHS algorithm has worse scaling with $\gamma$ and $\alpha_L$, but can be more versatile as we only require black-box access to the Hamiltonian simulation of $L$ itself. 

A recent work~\cite{HolmesMuraleedharanSommaEtAl2022} proposes a new Gibbs state preparation algorithm inspired by fluctuation theorems.  
This algorithm works well when the Hamiltonian can be decomposed as $L = L_0+L_1$ such that $\gamma\|L_1\|$ is moderate and the goal is to prepare the Gibbs state of $L$ from that of $L_0$. 
The algorithm also works for a general Hamiltonian $L$, although then the overall complexity scales badly as $\widetilde{\mathcal{O}}\left( e^{\left( \frac{3}{\epsilon}+\frac{1}{2}\right) \gamma \alpha_L} \right)$~\cite[Corollary 1.7]{HolmesMuraleedharanSommaEtAl2022}. 
We remark that when applied to positive semi-definite Hamiltonians, the complexity of the algorithm in~\cite{HolmesMuraleedharanSommaEtAl2022} is significantly better than its worst-case analysis. 
To see this, we restate the key technical result~\cite[Lemma 3.4]{HolmesMuraleedharanSommaEtAl2022}: letting $\Delta = \sqrt{\log(6/\epsilon)}$, $z = \gamma (w_{\max}-w_l) + 2\Delta^2$, and $\delta = 2\pi/z$, then for any $w_{\max} I \succeq L \succeq w_l I$, we have $\left\|e^{-\gamma L/2} - \sum_{j=-J}^J c_j e^{i j \delta\gamma L/2}\right\| \leq \frac{\epsilon}{3} e^{-\gamma w_l/2}$, where the coefficients satisfy $\|c\|_1 = \mathcal{O}(e^{\Delta} e^{-\gamma w_l/2})$ and the highest order is $J = \mathcal{O}(z^{3/2})$. 
In the positive semi-definite case, we may choose the parameter $w_{\max} = \alpha_L $ and $w_l = 0$. 
The overall query complexity of using LCU to implement $\sum_{j=-J}^J c_j e^{i j \delta\gamma L/2}$ becomes 
\begin{equation}
    \mathcal{O}\left( \|c\|_1 J \delta \gamma \|L\| \right) = \mathcal{O}\left(e^{\Delta} z^{3/2} z^{-1} \gamma \|L\| \right) =  \mathcal{O}\left(e^{\sqrt{\log(1/\epsilon)}} \gamma^{3/2} \|L\|^{3/2} \right). 
\end{equation}
Therefore, our LCHS algorithm has better dependence on all the parameters for positive semi-definite Hamiltonians. 

Finally, we describe two more related works that are not included in~\cref{tab:Gibbs} since their setup is somewhat different. 
Reference~\cite{TongAnWiebe2021} designs a Gibbs state preparation algorithm based on fast inversion. 
It considers a positive definite Hamiltonian $L = L_0 + L_1$ such that $\|L_1\| \ll \|L_0\|$ and $L_0^{-1}$ can be efficiently constructed. 
The algorithm removes the explicit dependence on the spectral norm $\|L\|$ and only depends linearly on $\|L_1\|$. 
Another reference~\cite{ChowdhurySommaSubasi2021} shows how to approximate the matrix function $e^{-\gamma L}$ for any $\|L\| \leq 1$ using an $\mathcal{O}(\gamma+\log(1/\epsilon))$-degree polynomial. 
However, such a polynomial cannot be efficiently implemented by QSVT, because even if we only care about the regime where $L\succeq 0$, the normalization factor of the matrix polynomial constructed by QSVT depends on the worst case over the entire interval $[-1,1]$, and in this example we encounter an exponentially large $e^{\gamma}$ normalization factor due to the bad scaling of the function $e^{-\gamma x}$ at $x = -1$.

\section{Other kernel functions}\label{app:kernel}

Here we present a few other examples of kernel functions $f(z)$ that satisfy the assumptions in~\cref{thm:LCHS_improved}. 

We have discussed that 
\begin{equation}
    f(z) = \frac{1}{\pi(1+iz)}
\end{equation}
gives the original LCHS formula in~\cref{eqn:LCHS_original}. 
It satisfies the assumptions in~\cref{thm:LCHS_improved} with $\alpha = 1$. 
An immediate generalization is to consider a high-order polynomial of $z$ on the denominator, resulting in the kernel function
\begin{equation}
    f(z) = \frac{2^{p-1}}{\pi(1+iz)^p}
\end{equation}
for a positive integer $p$. 
Noting that $1/(1+iz)^p = 1/e^{p\log(1+iz)}$, a slightly better choice is to replace $p\log(1+iz)$ by $(\log(1+iz))^p$. 
This gives the kernel function 
\begin{equation}
    f(z) = \frac{1}{2\pi e^{-(\log 2)^p} e^{(\log(1+iz))^p}}, 
\end{equation}
which decays as $1/\exp(\poly(\log(|z|)))$. 
To get even better convergence, we can try to replace the power of $\log(1+iz)$ by a power of $1+iz$, but the exponent $\beta$ must be smaller than $1$ according to the no-go result in~\cref{prop:prop_non_existence}. 
This gives the function 
\begin{equation}
    f(z) = \frac{1}{ 2\pi e^{-2^{\beta}} e^{(1+iz)^{\beta}} }, \quad 0 < \beta < 1, \label{eqn:kernel_exp_restated}
\end{equation}
which is exactly what we use in the main text. 

Now we provide more numerical illustrations of the kernel function in~\cref{eqn:kernel_exp_restated}. 
\Cref{fig:kernel} shows $\frac{f(k)}{1-ik}$ for real $k$ with different choices of $\beta$. 
From~\cref{fig:kernel}, we see that, as $\beta$ approaches $1$, the function decays slowly within the intermediate region $k\in(-30,30)$, even though the asymptotic decay rate increases with $\beta$. 
One way to understand this is that in the limit of $\beta=1$, the function $\frac{f(k)}{1-ik}$ becomes $\frac{1}{C_1(1-ik)e^{1+ik}}$, which only decays linearly for large $k$. 
Numerically, the function with a value of $\beta$ away from $0$ and $1$ has much better decay properties compared to the Cauchy distribution $1/(\pi(1+k^2))$. 
A practical choice of $\beta$ may be from $0.7$ to $0.8$, as suggested by our numerical tests in~\cref{fig:kernel_errors,fig:kernel}.

\begin{figure}
    \centering
    \includegraphics[width = 0.45\textwidth]{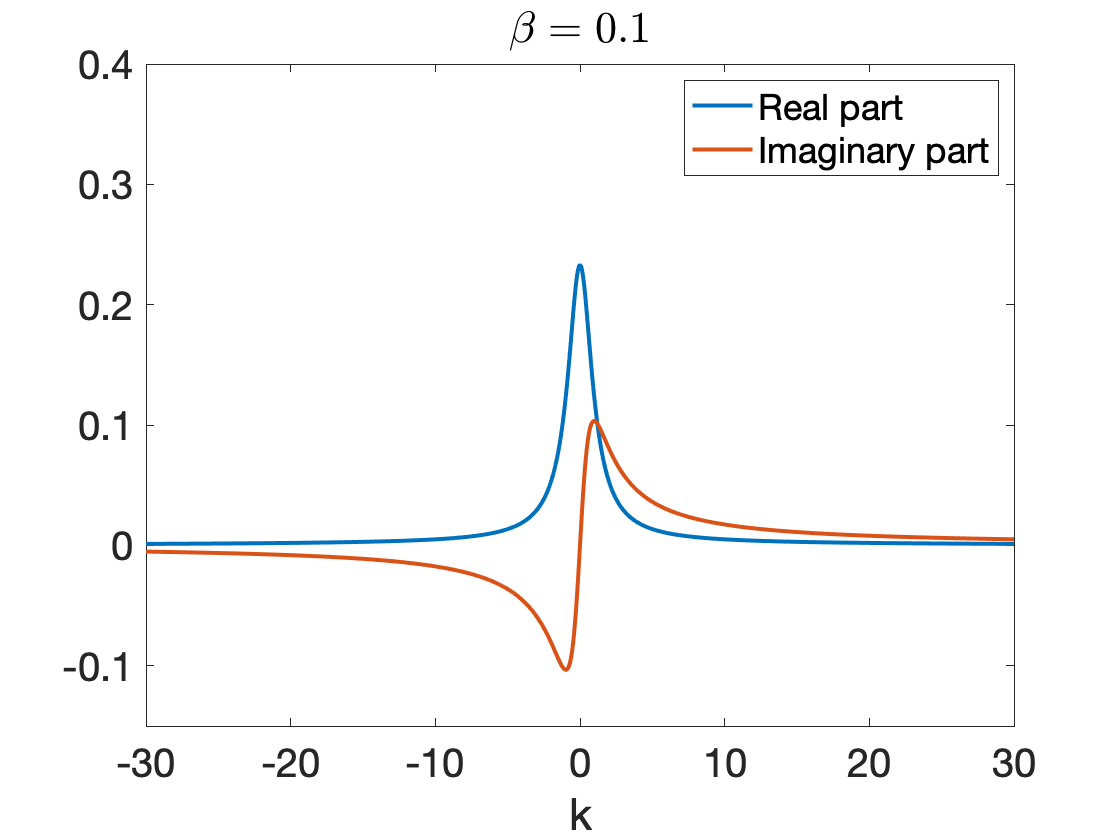} 
    \includegraphics[width = 0.45\textwidth]{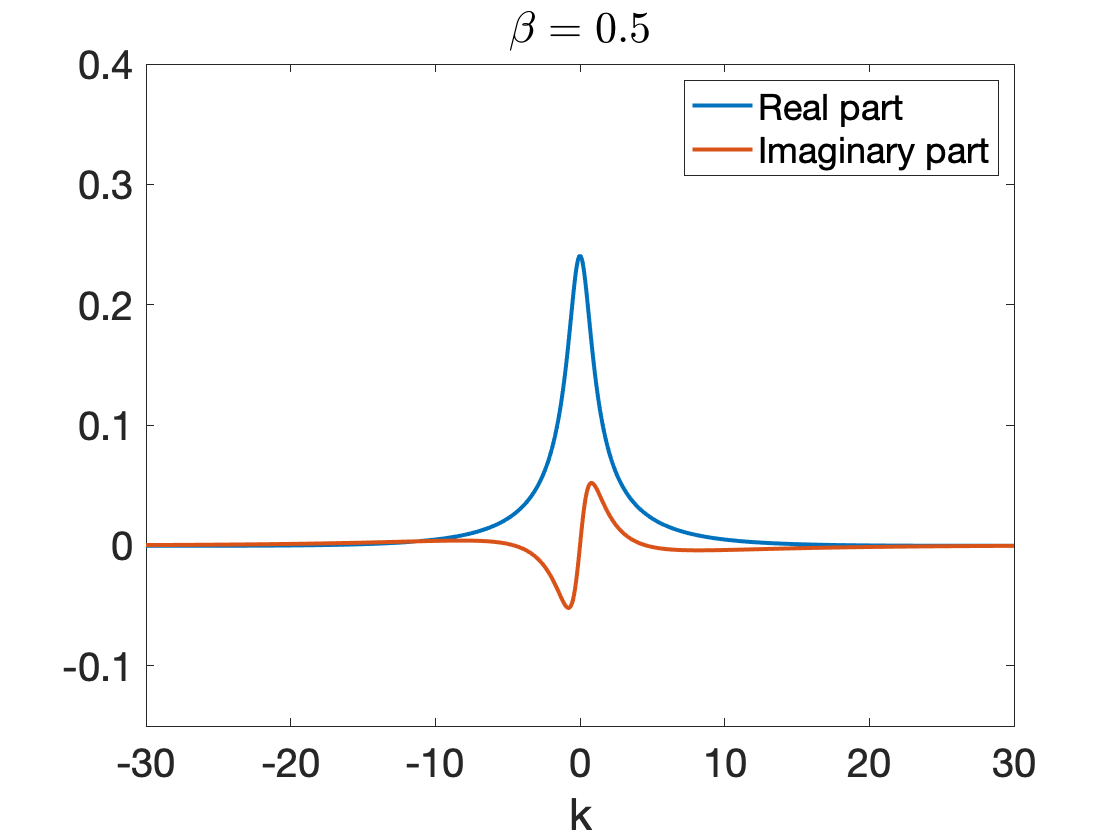} \\
    \includegraphics[width = 0.45\textwidth]{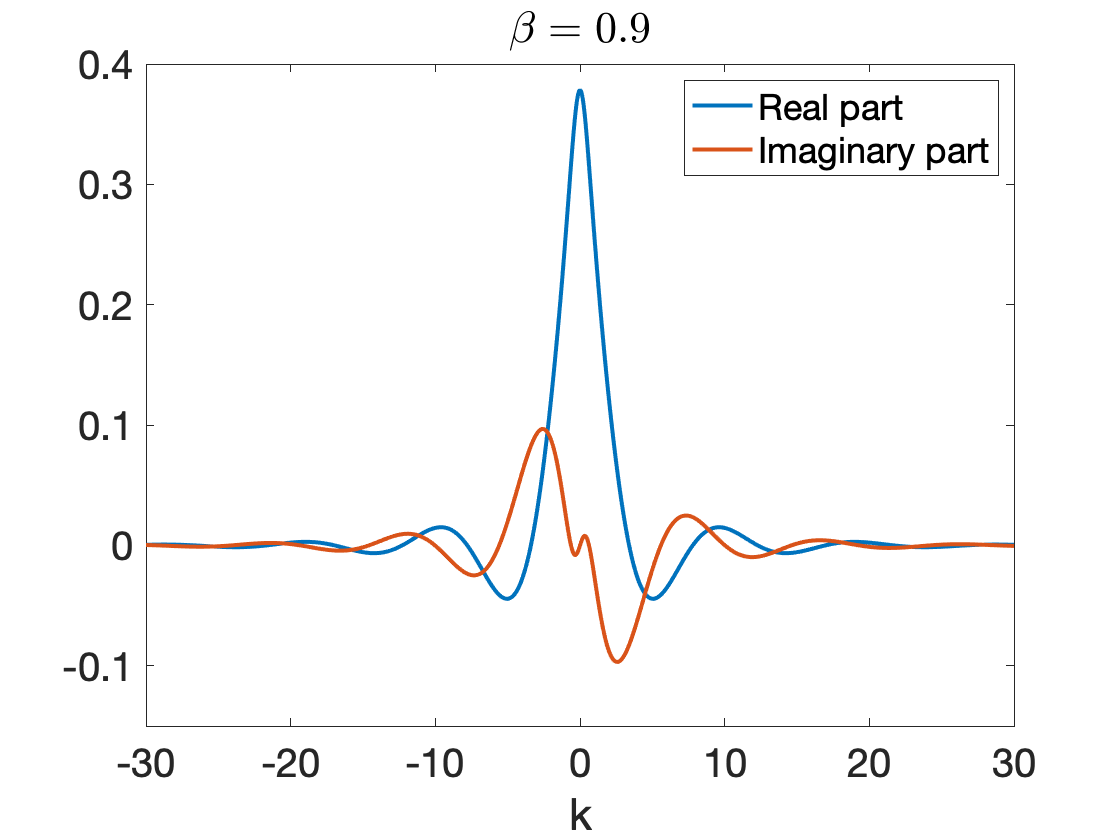} 
    \includegraphics[width = 0.45\textwidth]{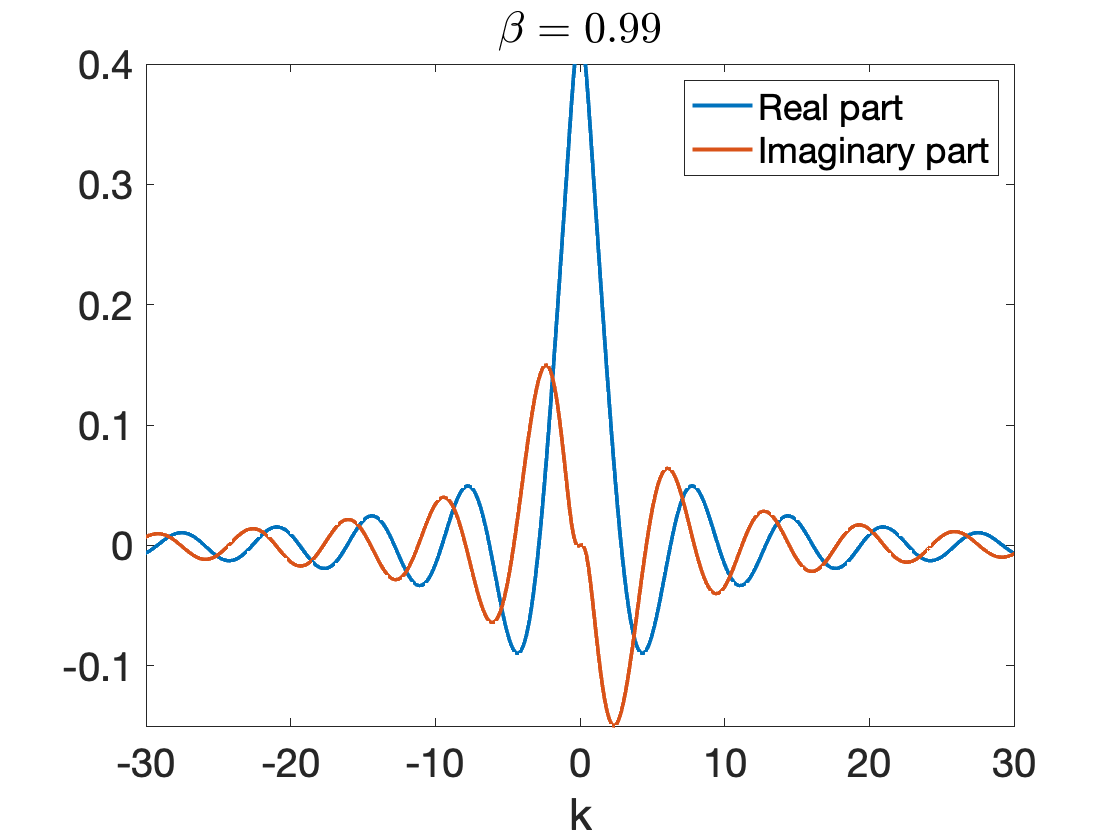} \\
    \includegraphics[width = 0.45\textwidth]{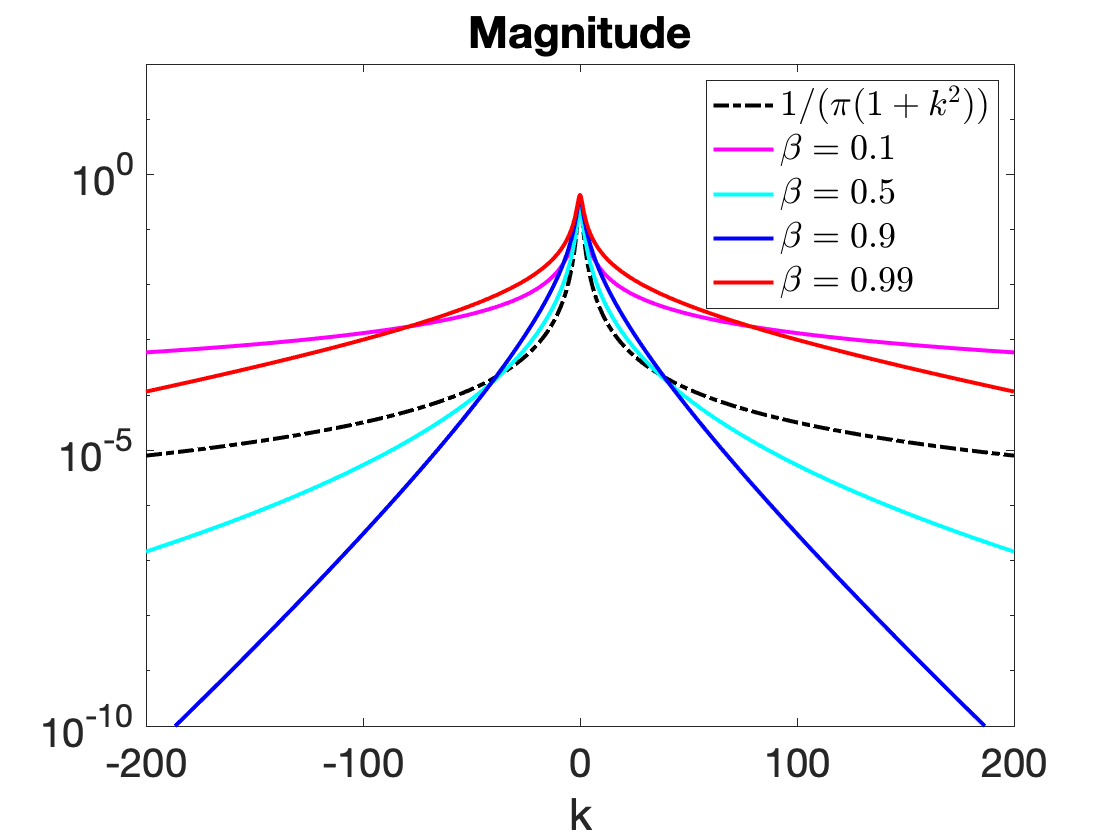} 
    \caption{Plots of $f(k)/(1-ik)$ where $f(k)$ is the kernel function defined in~\cref{eqn:kernel_exp_restated} with $\beta = 0.1, 0.5, 0.9, 0.99$. Top four: real parts and imaginary parts of the functions. Bottom: magnitudes of the functions, compared to the Cauchy distribution used in the original LCHS formula.}
    \label{fig:kernel}
\end{figure}

\section{Proof of~\texorpdfstring{\cref{lem:truncation}}{Lemma 10}}
\label{app:truncation_error}

\begin{proof}[Proof of~\cref{lem:truncation}]
    We start with 
    \begin{equation}\label{eqn:truncation_error_proof_eq1}
        \left\| \int_{\mathbb{R}} g(k) U(T,k) \ud k - \int_{-K}^{K} g(k) U(T,k) \ud k  \right\| \leq \int_{-\infty}^{-K} |g(k)| \ud k + \int_{K}^{\infty} |g(k)| \ud k. 
    \end{equation}
    We now estimate $|g(k)|$. 
    Since $(1+ik)^{\beta} = (\sqrt{k^2+1} e^{i\arctan k})^{\beta} = (k^2+1)^{\beta/2} (\cos(\beta\arctan k) + i \sin(\beta\arctan k))$, we have 
    \begin{equation}
        |g(k)| = \frac{1}{C_{\beta} \sqrt{k^2+1} e^{(k^2+1)^{\beta/2}\cos(\beta\arctan k) }  } \leq \frac{1}{C_{\beta} |k| e^{|k|^{\beta}\cos(\beta\pi/2)} }.  
    \end{equation}
    Let $B = \left\lceil 1/\beta \right\rceil$. 
    For $|k| \geq 1$, by only taking the $B$th-order term in the Taylor series of an exponential function, we have 
    \begin{equation}
        e^{\frac{1}{2}|k|^{\beta} \cos(\beta\pi/2) } \geq \frac{1}{B!} \frac{1}{2^B} |k|^{B\beta} \left(\cos(\beta\pi/2)\right)^B \geq \frac{1}{B!} \frac{1}{2^B} |k| \left(\cos(\beta\pi/2)\right)^B. 
    \end{equation}
    Then, 
    \begin{equation}
        |g(k)| \leq \frac{1}{C_{\beta}|k|} \frac{2^B B!}{\left(\cos(\beta\pi/2)\right)^B |k|} e^{-\frac{1}{2}|k|^{\beta} \cos(\beta\pi/2) } \leq \frac{2^B B!}{C_{\beta} \left(\cos(\beta\pi/2)\right)^B } \frac{1}{ |k|^2} e^{-\frac{1}{2} K^{\beta} \cos(\beta\pi/2) }. 
    \end{equation}
    Notice that in the last inequality, we keep the variable $k$ in the denominator, and replace the one in the exponential function by the truncation parameter $K$. 
    By taking the integral with respect to $k$ in~\cref{eqn:truncation_error_proof_eq1}, we obtain 
    \begin{equation}
        \left\| \int_{\mathbb{R}} g(k) U(T,k) \ud k - \int_{-K}^{K} g(k) U(T,k) \ud k  \right\| \leq \frac{2^{B+1} B !}{C_{\beta} \left(\cos(\beta\pi/2)\right)^{B} } \frac{1}{ K } e^{-\frac{1}{2}K^{\beta} \cos(\beta\pi/2) }. 
    \end{equation}

    To estimate the scaling of $K$, we absorb $\beta$-dependent factors into the big-$\mathcal{O}$ notation and bound $1/K$ by $1$ to further bound the error by $\mathcal{O}\left( e^{-\frac{1}{2}K^{\beta} \cos(\beta\pi/2) } \right) $. 
    By further bounding this by $\mathcal{O}(\epsilon)$ and solving the inequality, we may choose $K = \mathcal{O}\left( (\log(1/\epsilon))^{1/\beta}\right)$. 
\end{proof}

\section{Proof of~\texorpdfstring{\cref{lem:quadrature}}{Lemma 11}}
\label{app:quadrature_error}

\begin{proof}[Proof of~\cref{lem:quadrature}]
    For each interval $[mh_1,(m+1)h_1]$, by~\cite[p.~98]{DavisRabinowitz1984}, we can bound the quadrature error by 
    \begin{equation}\label{eqn:quadrature_error_proof_eq1}
        \left\| \int_{mh_1}^{(m+1)h_1} g(k) U(T,k) \ud k -  \sum_{q=0}^{Q-1} c_{q,m} U(T,k_{q,m}) \right\| \leq \frac{(Q!)^4 h_1^{2Q+1} }{ (2Q+1) ((2Q)!)^3 } \left\| (g(k)U(T,k))^{(2Q)} \right\|. 
    \end{equation}
    Here the superscript $(q)$ refers to the $q$th-order partial derivative with respect to $k$. 

    For notational simplicity, we omit the explicit argument dependence on $k$ and $T$ whenever the text is clear. 
    Now we estimate the high-order derivative of the function $g$ and $U$. 
    From the definition of $U$, we have 
    \begin{equation}
        \frac{\ud U(t,k)}{\ud t} = -i(kL(t)+H(t)) U(t,k), \quad U(0,k) = I. 
    \end{equation}
    Differentiating this equation with respect to $k$ for $q$ times yields 
    \begin{equation}
        \frac{\ud U^{(q)}}{\ud t} = -i(kL(t)+H(t)) U^{(q)} - iq L(t) U^{(q-1)}, \quad U^{(q)}(0,k) = 0. 
    \end{equation}
    By the variation of parameters formula\REV{~\cite[Lemma A.1]{ChildsSuTranEtAl2020}, regarding $-i(kL(t)+H(t))$ as the coefficient matrix of the homogeneous part and $- iq L(t) U^{(q-1)}$ as the inhomogeneous term, we have }
    \begin{equation}
        U^{(q)}(t,k) = \int_0^t \mathcal{T} e^{-i\int_s^t (k L(s') + H(s')) \ud s' } (-iq) L(s) U^{(q-1)} (s,k) \ud s, 
    \end{equation}
    and thus 
    \begin{equation}
        \|U^{(q)}(t,k)\| \leq T  q \max_t \|L(t)\| \max_t\|U^{(q-1)}\|. 
    \end{equation}
    By taking the maximum over time on the left-hand side and iterating this inequality, we have 
    \begin{equation}\label{eqn:quadrature_error_proof_U_est}
        \|U^{(q)}\| \leq \max_t\|U^{(q)}\| \leq (Tq\max_t \|L(t)\|)^q. 
    \end{equation}
    The derivatives of $g$ are more involved. 
    We first use the high-order chain rule (a.k.a.~Fa\`a di Bruno's formula) to compute 
    \begin{align}
        \left( e^{-(1+ik)^{\beta}} \right)^{(p)} &= \sum_{\sum_{j=1}^p jq_j = p} \frac{p!}{q_1! (1!)^{q_1} q_2! (2!)^{q_2} \cdots q_p! (p!)^{q_p}} e^{-(1+ik)^{\beta}} \prod_{j=1}^p \left( \left(-(1+ik)^{\beta}\right)^{(j)} \right)^{q_j} \\
        & = \sum_{\sum_{j=1}^p jq_j = p} \frac{p!}{q_1! (1!)^{q_1} q_2! (2!)^{q_2} \cdots q_p! (p!)^{q_p}} e^{-(1+ik)^{\beta}} \nonumber\\
        &\quad\quad\quad\quad\quad\quad \times \prod_{j=1}^p \left( -i^j\beta(\beta-1)\cdots (\beta-j+1) (1+ik)^{\beta-j} \right)^{q_j}. 
    \end{align}
    Then we can estimate 
    \begin{align}
        \left|  \left( e^{-(1+ik)^{\beta}} \right)^{(p)} \right| & \leq \sum_{\sum_{j=1}^p jq_j = p} \frac{p!}{q_1! (1!)^{q_1} q_2! (2!)^{q_2} \cdots q_p! (p!)^{q_p}}  \prod_{j=1}^p \left( (j-1)! \right)^{q_j} \\
        & \leq \sum_{\sum_{j=1}^p jq_j = p} \frac{p!}{q_1!  q_2! \cdots q_p! }  . 
    \end{align}
    Notice that, for the tuples $(q_1,\ldots,q_p)$ satisfying $\sum_{j=1}^p jq_j = p$, as each $q_j$ is within $[0,p/j]$, the number of such tuples must be smaller than $p(p/2+1)(p/3+1)\cdots(p/p+1) = \binom{2p}{p}$. 
    Therefore 
    \begin{equation}\label{eqn:bound_FdB_sum}
        \left|  \left( e^{-(1+ik)^{\beta}} \right)^{(p)} \right| \leq p! \binom{2p}{p} = \frac{(2p)!}{p!}. 
    \end{equation}
    \REV{Next, for the function $g(k) = \frac{1}{C_{\beta} (1-ik) e^{(1+ik)^{\beta}} }$, which is the product of $\frac{1}{C_{\beta} (1-ik)}$ and $e^{-(1+ik)^{\beta}}$, we can use the product rule to compute its derivatives as }
    \begin{align}
        g^{(q)} &= \frac{1}{C_{\beta}} \sum_{p=0}^q \binom{q}{p} \left(\frac{1}{1-ik}\right)^{(q-p)} \left( e^{-(1+ik)^{\beta}} \right)^{(p)} \\
        & = \frac{1}{C_{\beta}} \sum_{p=0}^q \binom{q}{p} \frac{i^{q-p}(q-p)!}{(1-ik)^{q-p+1}} \left( e^{-(1+ik)^{\beta}} \right)^{(p)}, 
    \end{align}
    and 
    \begin{align}\label{eqn:quadrature_error_proof_g_est}
        |g^{(q)}| \leq \frac{1}{C_{\beta}} \sum_{p=0}^q \binom{q}{p} (q-p)! \frac{(2p)!}{p!} 
        = \frac{q!}{C_{\beta}} \sum_{p=0}^q \binom{2p}{p} \leq \frac{q!}{C_{\beta}} \sum_{p=0}^q 2^{2p} \leq \frac{4}{3C_{\beta}} 4^q q!. 
    \end{align}
    Now, we use the product formula again, as well as~\cref{eqn:quadrature_error_proof_g_est} and~\cref{eqn:quadrature_error_proof_U_est}, to obtain 
    \begin{align}
        \|(gU)^{(2Q)}\| & \leq \sum_{q=0}^{2Q} \binom{2Q}{q} |g^{(2Q-q)}| \|U^{(q)}\| \\
        & \leq \sum_{q=0}^{2Q} \binom{2Q}{q} \frac{4}{3C_{\beta}} 4^{2Q-q} (2Q-q)! (T\max_t \|L(t)\|)^q q^q. 
    \end{align}
    Since $q^q \leq e^q q!$, 
    \begin{align}
        \|(gU)^{(2Q)}\| & \leq \sum_{q=0}^{2Q} \binom{2Q}{q} \frac{4}{3C_{\beta}} 4^{2Q-q} (2Q-q)! (T\max_t \|L(t)\|)^q e^q q! \\
        & = \frac{4}{3C_{\beta}} 4^{2Q} \sum_{q=0}^{2Q} (2Q)! \left(\frac{eT\max_t \|L(t)\|}{4}\right)^q  \\
        & \leq \frac{4}{3C_{\beta}} 4^{2Q}  (2Q)! (2Q+1) \left(\frac{eT\max_t \|L(t)\|}{4}\right)^{2Q} \\
        & = \frac{4}{3C_{\beta}} (2Q)! (2Q+1) \left(eT\max_t \|L(t)\|\right)^{2Q}. 
    \end{align}
    Plugging this estimate into the quadrature error formula~\cref{eqn:quadrature_error_proof_eq1}, we have 
    \begin{align}
        & \quad \left\| \int_{mh_1}^{(m+1)h_1} g(k) U(T,k) \ud k -  \sum_{q=0}^{Q-1} c_{q,m} U(T,k_{q,m}) \right\| \\
        & \leq \frac{(Q!)^4 h_1^{2Q+1} }{ (2Q+1) ((2Q)!)^3 } \frac{4}{3C_{\beta}} (2Q)! (2Q+1) \left(eT\max_t \|L(t)\|\right)^{2Q} \\
        & = \frac{(Q!)^4 h_1^{2Q+1} }{ ((2Q)!)^2 } \frac{4}{3C_{\beta}} \left(eT\max_t \|L(t)\|\right)^{2Q} \\
        & \leq  \frac{4 }{3C_{\beta}} h_1^{2Q+1} \left(\frac{eT\max_t \|L(t)\|}{2}\right)^{2Q}. 
    \end{align}
    By summing over all the short intervals, we have 
    \begin{align}
        & \left\| \int_{-K}^{K} g(k) U(T,k) \ud k - \sum_{m = -K/h_1}^{K/h_1-1} \sum_{q=0}^{Q-1} c_{q,m} U(T,k_{q,m}) \right\| \nonumber\\
        &\quad \leq \sum_{m = -K/h_1}^{K/h_1-1} \left\| \int_{mh_1}^{(m+1)h_1} g(k) U(T,k) \ud k -  \sum_{q=0}^{Q-1} c_{q,m} U(T,k_{q,m}) \right\| \\
        &\quad \leq  \frac{8 }{3C_{\beta}} K h_1^{2Q} \left(\frac{eT\max_t \|L(t)\|}{2}\right)^{2Q}. 
    \end{align}
    This completes the proof of the first part. 

    For the second part, according to the choice of $h_1 = 1/(eT\max_t\|L(t)\|)$, we further bound 
    \begin{equation}
        \left\| \int_{-K}^{K} g(k) U(T,k) \ud k - \sum_{m = -K/h_1}^{K/h_1-1} \sum_{q=0}^{Q-1} c_{q,m} U(T,k_{q,m}) \right\|  \leq \frac{8 }{3C_{\beta}} K \frac{1}{2^{2Q}}. 
    \end{equation}
    In order to bound this by $\epsilon$, it suffices to choose 
    \begin{equation}
        Q = \left\lceil \frac{1}{\log 4} \log\left( \frac{8}{3C_{\beta}} \frac{K}{\epsilon} \right) \right\rceil = \mathcal{O}\left( \log\left( \frac{K}{\epsilon} \right) \right) = \mathcal{O}\left( \log\left( \frac{1}{\epsilon} \right) \right), 
    \end{equation}
    where the last equation is due to $K = \mathcal{O}\left(\left(\log(1/\epsilon)\right)^{1/\beta}\right)$ from~\cref{lem:truncation}. 
\end{proof}

\section{Proof of~\texorpdfstring{\cref{lem:coefficient_1norm}}{Lemma 12}}
\label{app:coefficient_1norm}

\begin{proof}[Proof of~\cref{lem:coefficient_1norm}]
    Notice that 
    \begin{equation}
        \sum_{q,m} |c_{q,m}| = \sum_{m=-K/h_1}^{K/h_1-1} \sum_q w_q |g(k_{q,m})|, 
    \end{equation}
    which is the composite Gaussian quadrature formula for the integral $\int_{-K}^{K} |g(k)| \ud k $. 
    Then 
    \REV{\begin{equation}
        \sum_{q,m} |c_{q,m}| \leq \int_{-\infty}^{\infty} |g(k)| \ud k + \mathcal{E}, 
    \end{equation}
    where $\mathcal{E}$ is the corresponding quadrature error. 
    It suffices to show that both the integral and the quadrature error are at most $\mathcal{O}(1)$. }

    \REV{From~\cref{eqn:bound_f_abs}, we have 
    \begin{align}
        |g(k)| \leq \frac{1}{C_{\beta} |1-ik| e^{|k|^{\beta}\cos(\beta\pi/2)} } \leq \frac{e^2}{2\pi \sqrt{k^2+1} (1+|k|^{\beta} \cos(\beta \pi /2)) }, 
    \end{align}
    where in the second inequality we use $C_{\beta} \geq 2\pi e^{-2}$ and $e^{x} \geq 1+x$ for $x \geq 0$. 
    To bound the integral $\int_{-\infty}^{\infty} |g(k)| \ud k = 2 \int_{0}^{\infty} |g(k)| \ud k$, we divide the interval $(0,\infty)$ into $(0,1)$ and $(1,\infty)$. 
    On the interval $(0,1)$, we simply bound 
    \begin{equation}
        |g(k)| \leq |g(0)| = \frac{e^2}{2\pi}, 
    \end{equation}
    and on the interval $(1,\infty)$, we use 
    \begin{equation}
        |g(k)| \leq \frac{e^2}{2\pi \sqrt{k^2+1} (1+|k|^{\beta} \cos(\beta \pi /2)) } \leq \frac{e^2}{2\pi k^{\beta+1} \cos(\beta \pi /2) }. 
    \end{equation}
    Then we have 
    \begin{align}
        \int_{-\infty}^{\infty} |g(k)| \ud k & = 2\int_0^1 |g(k)| \ud k + 2 \int_1^{\infty} |g(k)| \ud k \\
        & \leq 2\int_0^1 \frac{e^2}{2\pi} \ud k + 2 \int_1^{\infty} \frac{e^2}{2\pi k^{\beta+1} \cos(\beta \pi /2) } \ud k \\
        & = \frac{e^2}{\pi}\left( 1 + \frac{1}{\beta \cos(\beta\pi/2)} \right). 
    \end{align}
    Using $\cos(\beta\pi/2) = \sin((1-\beta)\pi/2) \geq 1-\beta$, we can further bound 
    \begin{equation}\label{eqn:bound_int_gabs}
        \int_{-\infty}^{\infty} |g(k)| \ud k \leq \frac{e^2}{\pi}\left( 1 + \frac{1}{\beta (1-\beta) } \right)
    \end{equation}
    which is $\mathcal{O}(1)$ for a fixed value of $\beta$. 
    }

    \REV{Now we estimate the quadrature error 
    \begin{equation}
        \mathcal{E} = \left|\int_{-K}^K |g(k)| \ud k - \sum_{m=-K/h_1}^{K/h_1-1} \sum_q w_q |g(k_{q,m})| \right| \leq \sum_{m=-K/h_1}^{K/h_1-1} \left| \int_{mh_1}^{(m+1)h_1} |g(k)| \ud k - \sum_q w_q |g(k_{q,m})| \right|, 
    \end{equation}
    which can be bounded in a similar way as the proof of~\cref{lem:quadrature} in~\cref{app:quadrature_error}. 
    Specifically, for each interval $[mh_1,(m+1)h_1]$, by~\cite[p.~98]{DavisRabinowitz1984}, we can bound the quadrature error as 
    \begin{equation}\label{eqn:proof_c1norm_eq4}
        \left| \int_{mh_1}^{(m+1)h_1} |g(k)| \ud k - \sum_q w_q |g(k_{q,m})| \right| \leq \frac{(Q!)^4 h_1^{2Q+1} }{ (2Q+1) ((2Q)!)^3 } \max_{k} \left| |g(k)|^{(2Q)} \right|. 
    \end{equation}
    Let us estimate the derivatives of $|g(k)| = \frac{1}{C_{\beta} \sqrt{k^2+1} e^{(k^2+1)^{\beta/2}\cos(\beta\arctan k) } }$. 
    Define 
    \begin{equation}
        h(k) = \log |g(k)| = -\log C_{\beta} - \frac{1}{2}\log(k^2+1) - (k^2+1)^{\beta/2}\cos(\beta\arctan k). 
    \end{equation}
    Then the high-order chain rule (a.k.a.~Fa\`a di Bruno's formula) gives 
    \begin{align}
        |g(k)|^{(p)} = \left( e^{h(k)} \right)^{(p)} = \sum_{\sum_{j=1}^p jq_j = p} \frac{p!}{q_1! (1!)^{q_1} q_2! (2!)^{q_2} \cdots q_p! (p!)^{q_p}} e^{h(k)} \prod_{j=1}^p \left( \left(h(k)\right)^{(j)} \right)^{q_j}, 
    \end{align}
    and thus 
    \begin{equation}\label{eqn:proof_c1norm_eq3}
        \left||g(k)|^{(p)} \right| \leq |g(k)| \sum_{\sum_{j=1}^p jq_j = p} \frac{p!}{q_1! (1!)^{q_1} q_2! (2!)^{q_2} \cdots q_p! (p!)^{q_p}}  \prod_{j=1}^p \left| \left(h(k)\right)^{(j)} \right|^{q_j}. 
    \end{equation}
    Let $r(k) = \sqrt{k^2+1}$ and $\theta(k) = \arctan k$, and then 
    \begin{equation}
        h(k) = -\log C_{\beta} - \log r(k) - r(k)^{\beta}\cos(\beta\theta(k)). 
    \end{equation}
    We handle the latter two terms separately. 
    }

    \REV{
    For $\log r(k)$, we have 
    \begin{equation}
        \frac{\ud}{\ud k} \log r(k) = \frac{k}{k^2+1} = \frac{1}{2}\left( \frac{1}{k+i} + \frac{1}{k-i} \right),
    \end{equation}
    so 
    \begin{equation}
        \frac{\ud^p}{\ud k^p} \log r(k) = \frac{(-1)^{p-1}(p-1)!}{2} \left( (k+i)^{-p} + (k-i)^{-p} \right) 
    \end{equation}
    and 
    \begin{equation}\label{eqn:proof_c1norm_eq1}
        \left|\frac{\ud^p}{\ud k^p} \log r(k)\right| \leq (p-1)!. 
    \end{equation}
    For $r(k)^{\beta}\cos(\beta\theta(k))$, we first use the high-order chain rule to obtain 
    \begin{align}
        \frac{\ud^p}{\ud k^p} r(k)^{\beta} &= \frac{\ud^p}{\ud k^p} (k^2+1)^{\beta/2}\\
        &= \sum_{\sum_{j=1}^p jq_j = p} \frac{p!}{q_1! (1!)^{q_1} q_2! (2!)^{q_2} \cdots q_p! (p!)^{q_p}} \left( \frac{\beta}{2}\left(\frac{\beta}{2}-1\right)\cdots \left(\frac{\beta}{2}-\sum q_j + 1\right) \right) \nonumber\\
        & \quad\quad\quad\quad\quad\quad \times (k^2+1)^{\beta/2-\sum q_j} \prod_{j=1}^p \left( \left(k^2+1\right)^{(j)} \right)^{q_j} \\
        & = \sum_{ q_1+2q_2 = p} \frac{p!}{q_1! q_2! 2^{q_2} } \left( \frac{\beta}{2}\left(\frac{\beta}{2}-1\right)\cdots \left(\frac{\beta}{2}-q_1-q_2 + 1\right) \right) (k^2+1)^{\beta/2-q_1-q_2} \left( 2k \right)^{q_1} 2^{q_2}. 
    \end{align}
    Using $|(k^2+1)^{\beta/2-q_1-q_2}  k^{q_1} | \leq 1$, we have 
    \begin{align}
        \left| \frac{\ud^p}{\ud k^p} r(k)^{\beta} \right| &\leq \sum_{ q_1+2q_2 = p} \frac{p! 2^{q_1} (q_1+q_2-1)! }{q_1! q_2! } = p! 2^p \sum_{ q=0}^{\lfloor p/2 \rfloor } \frac{ 2^{-2q} (p-q-1)! }{q! (p-2q)! } \\
        & \leq p! 2^p \sum_{ q=0}^{\lfloor p/2 \rfloor } \frac{ 2^{-2q} p^{q-1}}{q!} \leq (p-1)! 2^p e^{p/4} \leq (p-1)! 3^p. 
    \end{align}
    Similarly, 
    \begin{align}
        \frac{\ud^p}{\ud k^p} \cos(\beta \theta(k)) = \sum_{\sum_{j=1}^p jq_j = p} \frac{p!}{q_1! (1!)^{q_1} q_2! (2!)^{q_2} \cdots q_p! (p!)^{q_p}} \left( \frac{\ud^{\sum q_j}}{\ud \theta^{\sum q_j}} \cos(\beta \theta) \right) \prod_{j=1}^p \left( \left(\theta(k)\right)^{(j)} \right)^{q_j}. 
    \end{align}
    Using $\left|\left( \frac{\ud^{\sum q_j}}{\ud \theta^{\sum q_j}} \cos(\beta \theta) \right)\right| \leq 1$ and $|\left(\theta(k)\right)^{(j)}| \leq (j-1)!$, we have 
    \begin{align}
        \left| \frac{\ud^p}{\ud k^p} \cos(\beta \theta(k)) \right| \leq \sum_{\sum_{j=1}^p jq_j = p} \frac{p!}{q_1! (1!)^{q_1} q_2! (2!)^{q_2} \cdots q_p! (p!)^{q_p}} \prod_{j=1}^p \left( (j-1)! \right)^{q_j} \leq \frac{(2p)!}{p!} \leq p! 4^p,  
    \end{align}
    where the second inequality follows from the same reasoning as obtaining~\cref{eqn:bound_FdB_sum}. 
    Therefore, by the product rule, we have 
    \begin{align}
        \left| \frac{\ud^p}{\ud k^p} \left(r(k)^{\beta}\cos(\beta \theta(k))\right) \right| & \leq \sum_{j=0}^{p} \binom{p}{j} \left| \frac{\ud^j}{\ud k^j} r(k)^{\beta} \right| \left|  \frac{\ud^{p-j}}{\ud k^{p-j}} \cos(\beta \theta(k)) \right| \\
        & \leq \sum_{j=0}^{p} \frac{p!}{j!(p-j)!} (j-1)! 3^j (p-j)! 4^{p-j} \leq p! 4^p p. \label{eqn:proof_c1norm_eq2}
    \end{align}
    By~\cref{eqn:proof_c1norm_eq1} and~\cref{eqn:proof_c1norm_eq2}, we have 
    \begin{equation}
        \left| (h(k))^{(j)} \right| \leq (j-1)! + j! 4^j j \leq (j+1)! 4^j. 
    \end{equation}
    Plugging this back into~\cref{eqn:proof_c1norm_eq3}, we have 
    \begin{align}
        \left||g(k)|^{(p)} \right| &\leq \frac{1}{C_{\beta}} \sum_{\sum_{j=1}^p jq_j = p} \frac{p!}{q_1! (1!)^{q_1} q_2! (2!)^{q_2} \cdots q_p! (p!)^{q_p}}  \prod_{j=1}^p \left( (j+1)! 4^j \right)^{q_j} \\
        & = \frac{1}{C_{\beta}} \sum_{\sum_{j=1}^p jq_j = p} \frac{p! 4^p }{q_1! q_2! \cdots q_p!}  \prod_{j=1}^p \left( (j+1)^{1/j} \right)^{j q_j} \\
        & \leq \frac{1}{C_{\beta}} \sum_{\sum_{j=1}^p jq_j = p} \frac{p! 8^p }{q_1! q_2! \cdots q_p!} \leq \frac{8^p}{C_{\beta}} \frac{(2p)!}{p!} \leq \frac{p! 32^p}{C_{\beta}}. 
    \end{align}
    }

    \REV{Now, according to the quadrature error bound in~\cref{eqn:proof_c1norm_eq4}, we have 
    \begin{align}
        \left| \int_{mh_1}^{(m+1)h_1} |g(k)| \ud k - \sum_q w_q |g(k_{q,m})| \right| & \leq \frac{1}{C_{\beta}} \frac{(Q!)^4 h_1^{2Q+1} }{ (2Q+1) ((2Q)!)^2 } 32^{2Q} \leq \frac{1}{3 C_{\beta}}  h_1^{2Q+1} 16^{2Q}, 
    \end{align}
    and 
    \begin{equation}
        \mathcal{E} \leq \frac{2K}{3 C_{\beta}}  (16 h_1)^{2Q} . 
    \end{equation}
    According to the choices of $h_1 = 1/(eT\max_t\|L(t)\|)$ in~\cref{lem:quadrature}, for sufficiently large $T$ and $\max_t\|L(t)\|$, we can bound $h_1 \leq 1/32$ and 
    \begin{equation}
        \mathcal{E} \leq  \frac{2K}{3 C_{\beta}} \frac{1}{2^{2Q}}. 
    \end{equation}
    Furthermore, the choice of $Q$ in~\cref{lem:quadrature} assures that $\frac{8K}{3 C_{\beta}} \frac{1}{2^{2Q}} \leq \epsilon$, so we have $\mathcal{E} \leq \epsilon/4 \leq \mathcal{O}(1)$, which completes the proof. 
    }
\end{proof}

\section{Proof of~\texorpdfstring{\cref{lem:quadrature_error_inhomo}}{Lemma 13}}
\label{app:discretization_error_inhomo}

\begin{proof}[Proof of~\cref{lem:quadrature_error_inhomo}]
    We first bound the discretization error of the variable $k$. 
    Notice that 
    \begin{align}
        & \left\| \int_0^T \mathcal{T}e^{-\int_s^T A(s')\ud s'} b(s) \ud s - \int_0^T \sum_{m_1 = -K/h_1}^{K/h_1-1} \sum_{q_1=0}^{Q_1-1} c_{q_1,m_1} U(T,s,k_{q_1,m_1})  b(s)  \ud s \right\| \\
        &\quad \leq \|b\|_{L^1} \max_s \left\| \mathcal{T}e^{-\int_s^T A(s')\ud s'} - \sum_{m_1 = -K/h_1}^{K/h_1-1} \sum_{q_1=0}^{Q_1-1} c_{q_1,m_1} U(T,s,k_{q_1,m_1})  \right\|. 
    \end{align}
    It suffices to bound the worst-case discretization error by $\mathcal{O}(\epsilon/\|b\|_{L^1})$. 
    According to~\cref{lem:truncation} and~\cref{lem:quadrature}, we may choose 
    \begin{equation}
        K = \mathcal{O}\left( \left(\log\left(1+\frac{\|b\|_{L^1}}{\epsilon}\right)\right)^{1/\beta} \right), \quad h_1 = \frac{1}{eT \max_t\|L(t)\|}, \quad Q_1 = \mathcal{O}\left( \log\left( 1+\frac{\|b\|_{L^1}}{\epsilon} \right) \right).  
    \end{equation}

    The second part of the error can be bounded as 
    \begin{align}
            &\left\| \int_0^T \sum_{m_1 = -K/h_1}^{K/h_1-1} \sum_{q_1=0}^{Q_1-1} c_{q_1,m_1} U(T,s,k_{q_1,m_1})  b(s)  \ud s  \right. \nonumber \\
            &\qquad \left. - \sum^{T/h_2-1}_{m_2 = 0} \sum_{q_2=0}^{Q_2-1} \sum_{m_1 = -K/h_1}^{K/h_1-1} \sum_{q_1=0}^{Q_1-1} c'_{q_2,m_2} c_{q_1,m_1} U(T,s_{q_2,m_2},k_{q_1,m_1})  \ket{b(s_{q_2,m_2})} \right\| \\
            &\quad \leq \|c\|_1 \max_k \left\| \int_0^T U(T,s,k)  b(s)  \ud s - \sum^{T/h_2-1}_{m_2 = 0} \sum_{q_2=0}^{Q_2-1}  c'_{q_2,m_2} U(T,s_{q_2,m_2},k)  \ket{b(s_{q_2,m_2})} \right\|.  \label{eqn:proof_inhomo_error_eq3}
    \end{align}
    Let $h(s) = U(T,s,k)b(s) $. Then by~\cite[p.~98]{DavisRabinowitz1984}, we have 
    \begin{equation}\label{eqn:proof_inhomo_error_eq2}
        \left\| \int_0^T U(T,s,k)  b(s)  \ud s - \sum^{T/h_2-1}_{m_2 = 0} \sum_{q_2=0}^{Q_2-1}  c'_{q_2,m_2} U(T,s_{q_2,m_2},k)  \ket{b(s_{q_2,m_2})} \right\| \leq \frac{T (Q_2!)^4 h_2^{2Q_2} }{ (2Q_2+1) ((2Q_2)!)^3 } \max\left\| h^{(2Q_2)} \right\|. 
    \end{equation}
    Here the superscript $(q)$ denotes the $q$th-order derivative with respect to $s$, and we omit the explicit dependence on the arguments whenever the text is clear. 
    We now bound the derivatives of $h$. 
    Notice that $U(T,s,k) = \mathcal{T} e^{-i \int_0^{T-s} (k \widetilde{L}(s') + \widetilde{H}(s') ) \ud s' }$, where $\widetilde{L}(s) = L(T-s)$ and $\widetilde{H}(s) = H(T-s)$. 
    Consider $V(s) = \mathcal{T} e^{-i \int_0^{s} (k \widetilde{L}(s') + \widetilde{H}(s') ) \ud s' }$. 
    Notice that $\max_s \|V^{(p)}\| = \max_s \|U^{(p)}(T,s,k)\|$ for each $p$, and similarly for $L,\widetilde{L}$ and $H,\widetilde{H}$, respectively. 
    By the definition of $V$, 
    \begin{equation}
        \frac{\ud V}{\ud s} = - i( k \widetilde{L} + \widetilde{H} ) V, 
    \end{equation}
    we have 
    \begin{equation}
        V^{(p)} = -i \sum_{j=0}^{p-1} \binom{p-1}{j} (k\widetilde{L} + \widetilde{H})^{(p-1-j)} V^{(j)},
    \end{equation}
    and thus 
    \begin{align}
        \max \|V^{(p)}\| & \leq \sum_{j=0}^{p-1} \binom{p-1}{j} \left( K \max \|L^{(p-1-j)}\| + \max \|H^{(p-1-j)}\| \right) \max \|V^{(j)}\| \\
        & \leq 2K \sum_{j=0}^{p-1} \binom{p-1}{j} \max \|A^{(p-1-j)}\|  \max \|V^{(j)}\|. 
    \end{align}
    By the definition of $\Lambda = \sup_{p \geq 0, t \in [0,T]} \|A^{(p)}\|^{1/(p+1)} $, we have 
    \begin{equation}
        \max \|V^{(p)}\| \leq 2K\Lambda \sum_{j=0}^{p-1} \binom{p-1}{j} \Lambda^{p-1-j}  \max \|V^{(j)}\|. 
    \end{equation}
    By induction, we can prove that 
    \begin{equation}\label{eqn:proof_inhomo_error_eq1}
        \max \|V^{(p)}\| \leq \left( (p+1) K\Lambda \right)^p. 
    \end{equation}
    Specifically, the cases $p=0,1$ are straightforward. 
    Supposing that the estimate holds for values less than $p$, we have 
    \begin{align}
            \max \|V^{(p)}\| &\leq 2K\Lambda \sum_{j=0}^{p-1} \binom{p-1}{j} \Lambda^{p-1-j}   \left( (j+1) K\Lambda \right)^j \\
            & \leq 2K\Lambda \sum_{j=0}^{p-1} \binom{p-1}{j} \Lambda^{p-1-j}   \left( p K\Lambda \right)^j \\
            & = 2K\Lambda \left(p K\Lambda + \Lambda \right)^{p-1} \\
            & \leq \left((p+1) K\Lambda\right)^p. 
    \end{align}
    Using the product rule, we have
    \begin{equation}
            \max \|h^{(2Q_2)}\| \leq \sum_{j=0}^{2Q_2} \binom{2Q_2}{j} \max \|b^{(2Q_2-j)}\| \max \|V^{(j)}\|. 
    \end{equation}
    By the definition $\Xi = \sup_{ p\geq 0, t \in [0,T] } \|b^{(p)}\|^{1/(p+1)} $ and~\cref{eqn:proof_inhomo_error_eq1}, we have 
    \begin{align}
            \max \|h^{(2Q_2)}\| &\leq \sum_{j=0}^{2Q_2} \binom{2Q_2}{j} \Xi^{2Q_2-j+1} \left( (j+1) K\Lambda \right)^j \\
            & \leq \Xi \sum_{j=0}^{2Q_2} \binom{2Q_2}{j} \Xi^{2Q_2-j} \left( (2Q_2+1) K\Lambda \right)^j \\
            & \leq \Xi \left( (2Q_2+1) K\Lambda + \Xi \right)^{2Q_2} \\
            & \leq (2Q_2+1)^{2Q_2} K^{2Q_2} (\Lambda+\Xi)^{2Q_2+1}. 
    \end{align}
    Plugging this back into~\cref{eqn:proof_inhomo_error_eq2} and using $q! \leq e^q q! $ and $(2q)! \geq 2^q (q!)^2 $, we have 
    \begin{align}
             & \quad \left\| \int_0^T U(T,s,k)  b(s)  \ud s - \sum^{T/h_2-1}_{m_2 = 0} \sum_{q_2=0}^{Q_2-1}  c'_{q_2,m_2} U(T,s_{q_2,m_2},k)  \ket{b(s_{q_2,m_2})} \right\| \\
             & \leq \frac{T (Q_2!)^4 h_2^{2Q_2} }{ (2Q_2+1) ((2Q_2)!)^3 } (2Q_2+1)^{2Q_2} K^{2Q_2} (\Lambda+\Xi)^{2Q_2+1}  \\
             & \leq \frac{T (Q_2!)^4 h_2^{2Q_2} }{ (2Q_2+1) ((2Q_2)!)^3 } e (2Q_2)^{2Q_2} K^{2Q_2} (\Lambda+\Xi)^{2Q_2+1} \\
             & \leq \frac{T (Q_2!)^4 h_2^{2Q_2} }{ (2Q_2+1) ((2Q_2)!)^2 } e^{2Q_2+1} K^{2Q_2} (\Lambda+\Xi)^{2Q_2+1} \\
             & \leq  \frac{e T (\Lambda+\Xi)}{ 2Q_2+1 } \left(\frac{e h_2 K (\Lambda+\Xi)}{2} \right)^{2Q_2}. 
    \end{align}
    By choosing 
    \begin{equation}
        h_2 = \frac{1}{eK(\Lambda+\Xi)}
    \end{equation}
    and dropping the term $2Q_2+1$ in the denominator, we have 
    \begin{equation}
        \left\| \int_0^T U(T,s,k)  b(s)  \ud s - \sum^{T/h_2-1}_{m_2 = 0} \sum_{q_2=0}^{Q_2-1}  c'_{q_2,m_2} U(T,s_{q_2,m_2},k)  \ket{b(s_{q_2,m_2})} \right\| \leq \frac{e T (\Lambda+\Xi)}{ 4^{Q_2} }. 
    \end{equation}
    Plugging this back into~\cref{eqn:proof_inhomo_error_eq3}, we can bound the second part of the discretization error by 
    \begin{equation}
        \|c\|_1 \frac{e T (\Lambda+\Xi)}{ 4^{Q_2} } = \mathcal{O}\left( \frac{T (\Lambda+\Xi)}{ 4^{Q_2} } \right). 
    \end{equation}
    In order to bound it by $\epsilon/2$, it suffices to choose 
    \begin{equation}
        Q_2 = \mathcal{O}\left( \log\left(\frac{T(\Lambda+\Xi)}{\epsilon}\right) \right),
    \end{equation}
which completes the proof.
\end{proof}

\end{document}